\documentclass[american,aps,pra,reprint,superscriptaddress,nofootinbib]{revtex4-2}

\usepackage{amsmath,amssymb}

\usepackage{color}
\usepackage{graphicx}
\usepackage[american]{babel}
\usepackage[utf8]{inputenc}
\usepackage{times}
\usepackage{braket} 				
\usepackage{amsthm}
\usepackage{amsfonts}
\usepackage{mathrsfs}
\usepackage{subcaption}
\usepackage{quantikz}
\usepackage{mathtools}
\usepackage{bbold}
\usepackage{titlesec}

\newcommand{\ketbra}[2]{|#1\rangle\!\langle#2|}

\definecolor{mygrey}{gray}{0.35}
\definecolor{myblue}{rgb}{0.2,0.2,0.8}
\definecolor{myzard}{cmyk}{0,0,0.05,0}
\definecolor{mywhite}{rgb}{1,1,1}
\definecolor{myred}{rgb}{0.9,0.1,0.}
\definecolor{goldenyellow}{rgb}{1.0, 0.87, 0.0}
\definecolor{cornellred}{rgb}{0.7, 0.11, 0.11}
\usepackage[colorlinks=true,citecolor=myblue,linkcolor=myblue,urlcolor=myblue]{hyperref}

\newtheoremstyle{customStyle1}  
{0pt}       
{0pt}       
{\normalfont}   
{\parindent}        
{\em}  
{. --}   	 
{.5em}       
{\thmname{#1}\thmnumber{ #2}\thmnote{ (#3)}}  

\newcounter{theorems}
\newtheorem{thm}[theorems]{Theorem}
\newtheorem{prop}[theorems]{Proposition}
\newtheorem{cor}[theorems]{Corollary}

\newtheorem{lem}[theorems]{Lemma}

\newtheorem{ex.}{Example}[theorems]

\newtheorem*{cor*}{Corollary}
\newtheorem*{thm*}{Theorem}
\newtheorem*{prop*}{Proposition}
\newtheorem*{lem*}{Lemma}
\newtheorem*{rem*}{Remark}

\usepackage{cancel}

\newcommand{\id}{{\mathbb{1}}}

\newcommand{\norm}[1]{\left\lVert#1\right\rVert}
\DeclareMathOperator{\trace}{Tr}
\newcommand{\Tr}[1]{\trace\left[#1\right]}
\newcommand{\partTr}[2]{\trace_{#1}\left[#2\right]}

\DeclareMathOperator{\diag}{diag}

\DeclarePairedDelimiter\floor{\lfloor}{\rfloor}

\newcommand{\DI}{\mathcal{D}\mathcal{I}}
\newcommand{\DIS}{\mathcal{D}\mathcal{I}\mathcal{S}}
\newcommand{\CI}{\mathcal{M}\mathcal{I}\mathcal{O}}
\newcommand{\CIS}{\mathcal{M}\mathcal{I}\mathcal{O}\mathcal{S}}
\newcommand{\CDI}{\mathcal{C}\mathcal{D}\mathcal{I}}

\newcommand{\suc}{\text{succ}}
\newcommand{\free}{\text{free}}
\newcommand{\IM}{\mathcal{I}\mathcal{M}}

\newcommand\newsubcap[1]{\phantomcaption%
	\caption*{\figurename~\thefigure\thesubfigure: #1}} 

\definecolor{dartmouthgreen}{rgb}{0.05, 0.5, 0.06}

\begin{document}
	\title{On the Role of Coherence in Shor's Algorithm}
	\author{Felix Ahnefeld}
	\email{felix.ahnefeld@uni-ulm.de}
	\affiliation{Institute of Theoretical Physics, 			Universit{\"a}t Ulm, Albert-Einstein-Allee 11, D-89069 Ulm, Germany}
	\author{Thomas Theurer}
	\email{thomas.theurer@ucalgary.ca}
	\affiliation{Department of Mathematics and Statistics, Institute for Quantum
Science and Technology, University of Calgary, AB, Canada T2N 1N4}
	\author{Dario Egloff}
	\email{dario.egloff@mailbox.tu-dresden.de}
	\affiliation{Institute of Theoretical Physics, Technical University Dresden, D-01062 Dresden, Germany}
	\author{Juan Mauricio Matera}
	\email{matera@fisica.unlp.edu.ar}
	\affiliation{IFLP-CONICET, Departamento de F{\'i}sica, Facultad de Ciencias Exactas, Universidad Nacional de La Plata, C.C. 67, La Plata 1900, Argentina}
	
	\author{Martin B. Plenio}
	\email{martin.plenio@uni-ulm.de}
	\affiliation{Institute of Theoretical Physics, Universit{\"a}t Ulm, Albert-Einstein-Allee 11, D-89069 Ulm, Germany}

	\begin{abstract}
Shor's factoring algorithm provides a super-polynomial speed-up over all known classical factoring algorithms. Here, we address the question of which quantum properties fuel this advantage. We investigate a sequential variant of Shor's algorithm with a fixed overall structure and identify the role of coherence for this algorithm quantitatively. We analyze this protocol in the framework of dynamical resource theories, which capture the resource character of operations that can create and detect coherence. This allows us to derive a lower and an upper bound on the success probability of the protocol, which depend on rigorously defined measures of coherence as a dynamical resource. We compare these bounds with the classical limit of the protocol and conclude that within the fixed structure that we consider, coherence is the quantum resource that determines its performance by bounding the success probability from below and above. Therefore, we shine new light on the fundamental role of coherence in quantum computation.
		
	\end{abstract}
	\date{\today}
	\maketitle

	\section{Introduction} 
	Factoring large integers is considered to be a notoriously hard problem on a classical device. No classical factoring algorithm with polynomial run time is known, and the assumption that none exists lies at the heart of the widely used RSA cryptosystem~\cite{RSA}. Therefore, Shor's discovery of a quantum algorithm capable of factoring in polynomial time~\cite{Shor1995} not only attracted interest in this algorithm itself but the field of quantum computation in general: It is an example of a quantum algorithm that provides a super-polynomial computational speed-up over its best known classical counterpart (see also Refs.~\cite{Deutsch1992, Simon1997,Knill1998,Harrow2009}).
Since quantum devices are governed by laws different to those of classical physics, in principle, it might not seem too surprising that they can outperform them in certain applications. But which properties of quantum mechanics not present in classical physics fuel the speed-up in Shor's algorithm? And can they be used to explain speed-ups for the solution of other problems too? It is known that the presence of an unbounded amount of multi-partite entanglement is necessary for exponential speed-ups in circuit-based pure state quantum computation because every protocol that does not exhibit this property can be simulated efficiently on a classical device~\cite{Josza2003}. 	
This result, therefore, describes a necessary condition for exponential speed-ups in \textit{arbitrary} protocols but {\em not} a sufficient condition as the presence of unbounded entanglement does not gaurantee efficient quantum computation. This and the lack of a connection of entanglement and classical simulability in the case of mixed states might give a hint that deeper concepts underpin the computational speed-up. 

Here, we go one step further and instead of asking whether a resource is necessary to obtain speed-ups or describing its evolution during a protocol~\cite{Josza2003,Hillery2016,Anand2016,Shi2017,Liu2019YC}, we explore the speed-up that it actually grants. To start this exploration, we retreat from the general computational setting and focus on a specific algorithm with a fixed overall structure, namely a variant
of Shor’s algorithm introduced by Parker and Plenio~\cite{Parker2000}. The focus allows us to present lower and upper bounds on the performance of this algorithm that hold for mixed states too and are expressed in terms of coherence measures, which are derived in the framework of quantum resource theories~\cite{Streltsov2017}.

Quantum resource theories, see for example Refs.~\cite{Vedral1997,Vedral998,Horodecki2003,Aberg2006,Gour2008,Horodecki2013,Brandao2013, Veitch2014,Grudka2014,Baumgartz2014,delRio2015,Streltsov2017,Chitamber2019}, are  mathematical tools developed to describe the resource character of quantum properties in a mathematically and operationally well defined manner. Their central idea is to impose additional restrictions on the laws of quantum mechanics, which single out certain properties as precious resources. The resource content of various physical objects such as states, operations, or measurements can then be quantified rigorously via resource measures that cannot increase under actions that are still permitted in the presence of the constraints. Furthermore, one can study which physical operations may still be implemented in the presence of these restrictions and at what cost they can be overcome when supplied with resourceful quantum objects. This allows for an investigation of which resource is responsible for what quantum advantage. Besides the insights that our results give on Shor's algorithm, they also show that (dynamical) resource theories are applicable to problems of practical relevance. Whilst this is often used as a motivation to study resource theories, quantitative relations between coherence and performance beyond variations of discrimination, exclusion, and detection games~\cite{Napoli2016,Piani2016,Takagi2019b,Takagi2019,Skrzypczyk2019a,Skrzypczyk2019b,Uola2019,Mori2020,Ducuara2020a,Ducuara2020b,Uola2020,Masini2021} are surprisingly rare~\cite{Hillery2016,Matera2016,Biswas2017,Marvian2020}.

After a short introduction to the relevant aspects of resource theories and Shor's algorithm, we present our main results. First, we carefully motivate and describe what algorithm we are considering and how this allows us to investigate the role of coherence. This is crucial because we need to fix the overall structure of our algorithm: The most general approach to an investigation of the speed-up quantum resources grant in factoring would be to compare an ideal quantum algorithm (given fixed resources) to an ideal classical algorithm. Since it is unknown what such algorithms are, this is however out of reach. Instead, we will focus on the quantum part in Shor's algorithm, namely the order-finding protocol, and fix its core, the modular exponentiation, whilst varying the remainder. This approach provides enough freedom while giving enough structure to observe interesting quantitative connections. We conclude with a discussion and outlook and refer to the SM for proofs and further details.

\section{Quantum resource theories} 
In this section, we give a brief introduction to the resource theoretical notions that we will use in this work. We will restrict ourselves to finite-dimensional quantum systems, which we label by large Latin letters. Quantum states will be denoted by small Greek letters, quantum channels, i.e., linear and completely positive and trace-preserving (CPTP) maps that transform quantum states by large Greek letters. Additionally, if clear from the context, super-channels, i.e., linear maps that transform channels into channels will be labeled by large Latin letters as well. 

Generally, resource theories emerge from restrictions that are frequently motivated experimentally. Here we focus on constraints concerning the ability to create and detect coherence, but the concepts can be analogously applied to other restrictions. We begin by fixing the incoherent basis $\{\ket{i}\}_i$, i.e., the basis with respect to which we are going to describe coherence. Since we are considering circuit-based quantum computation, the computational basis in which we are encoding and extracting our classical information is the natural choice: If we never create coherence with respect to the computational basis, we are essentially reduced to the (classical) application of stochastic matrices onto probability vectors. 

A quantum state $\sigma$ is now considered incoherent and equivalent to a probability vector iff it is diagonal in the incoherent basis, i.e., iff $\Delta(\sigma)=\sigma$, where 
\begin{equation}
	\Delta(\rho)=\sum_i \ketbra{i}{i} \rho \ketbra{i}{i}
\end{equation}
denotes total dephasing in the incoherent (computational) basis $\{\ket{i}\}_i$. We denote the set of incoherent states by $\mathcal{I}$ and call the maximal set of channels $\Phi$ that cannot create coherence, i.e., turn an incoherent state into a coherent one, the \textit{maximally incoherent} channels and denote it by $\CI$~\cite{Aberg2006,Liu2017,Diaz2018,Theurer2019}. This set constitutes of all channels $\Phi$ that satisfy $\Phi\Delta=\Delta\Phi\Delta$. To exploit coherence, we do not only need to create it, but also use it. Using coherence is only possible if we have access to measurements that can detect it in the sense that its presence makes a difference in measurement statistics~\cite{Theurer2019,Yadin2016,Smirne2018}. 

As detailed in Ref.~\cite{Theurer2019}, it is possible to identify all instruments that cannot detect coherence with the \textit{detection-incoherent} channels $\mathcal{D}\mathcal{I}$, i.e., all channels $\Phi$ satisfying $\Delta \Phi=\Delta\Phi\Delta$~\cite{Meznaric2013,Liu2017,Theurer2019,Xu2020}. After we defined the set of channels that cannot create/detect coherence and are thus considered free for the respective task, we can use them to build dynamical resource theories~\cite{Dana2017,Zhuang2018,Theurer2019,Wang2019a,Wang2019b,Liu2020, Liu2019,Gour2019a,Gour2021,Saxena2020,Gour2020,Takagi2019,Bauml2019,Li2020, Takagi2020, Takagi2021} introduced only relatively recently to quantify how well non-free channels can create/detect coherence. The missing pieces are the super-channels that map quantum channels to quantum channels. A super-channel $S$ can be represented by two quantum channels $\Theta_1,\Theta_2$ that are used as pre-and post-processing, i.e.,
\begin{align}
	S[\Lambda]=\Theta_2(\mathbb{1}\otimes \Lambda)\Theta_1.
\end{align}
This definition is natural in the context of circuit quantum computation, but can also be shown to be the most general one consistent with an operational interpretation~\cite{Chiribella2008}. We now also divide the super-channels into free and non-free: A minimal requirement is that a free super-channel maps free channels to free channels, otherwise it would be possible to create non-free operations from free ones at no cost. Since both $\DI$ and $\CI$ are closed under sequential and parallel concatenation, we take the circuit-based approach and define a super-channel as free iff it can be represented by a free pre-and postprocessing~\cite{PhysRevLett.123.150401, gour2020dynamical}. The set of free super-channels in the resource theory concerning the creation/detection of coherence will be labeled by $\CIS/\DIS$. This concept allows us to compare the value of channels: A channel $\Theta$ is at least as valuable as a channel $\Lambda$ if there exists a free super-channel $S$ such that $S[\Theta]=\Lambda$. In general, it is difficult to decide whether such an $S$ exists, which is why one considers (dynamical) resource measures. These are functionals $M$ that map quantum operations to the non-negative numbers and satisfy (i) monotonicity: $M(\Theta)\ge M(S[\Theta])$ for all free super-channels $S$, i.e., they respect the preorder that the free super-channels impose on the channels and therefore their relative value, (ii) faithfulness: $M(\Theta)=0$ iff $\Theta$ is free, (iii) convexity.

\begin{figure*}\label{fig:Shor}
     \centering
     \begin{subfigure}[b]{0.4\textwidth}
         \centering
         \scalebox{0.68}{
         \begin{quantikz}
\lstick{$\ket{0}$} &\gate{H}\gategroup[5,steps=4,style={rounded corners,fill=goldenyellow!25, inner xsep=2pt}, background,label style={label position=above,anchor=north,yshift=-0.2cm}]{}\gategroup[4,steps=1,style={rounded corners,fill=cornellred!35, inner xsep=2pt}, background,label style={label position=above,anchor=north,yshift=-0.2cm}]{}&\qw &\ctrl{1} & \gate[4]{\mathcal{F}^\dagger}\gategroup[4,steps=1,style={rounded corners,fill=dartmouthgreen!35, inner xsep=2pt}, background,label style={label position=above,anchor=north,yshift=-0.2cm}]{}  &\meter{} \\
\lstick{$\ket{0}$} &\gate{H}&\qw &\ctrl{1} & &\meter{} \\
\lstick{$\ket{0}$} &\gate{H}&\qw &\ctrl{1}& &\meter{} \\
\lstick{$\ket{0}$} &\gate{H}&\qw &\ctrl{1}& &\meter{} \\
&&\lstick{$\ket{1}$}&\gate{U_c} \qwbundle[alternate]{}&\gate{\text{Discard}}\qwbundle[alternate]{}
\end{quantikz}
}
         \newsubcap{The quantum subroutine of Shor's order-finding algorithm.}
         \label{fig:ShorA}
     \end{subfigure}%
     \begin{subfigure}[b]{0.6\textwidth}
         \centering
		 \scalebox{0.68}{
		 \begin{quantikz}
	&\gategroup[1,steps=14,style={rounded corners,fill=dartmouthgreen!35, inner xsep=2pt}, label style={label position=above,anchor=north,yshift=-0.2cm}]{{\sc Classical Control and Classical Post-Processing}}&\gate{\text{A}} \cw &&&&&&&&&&&&&&&&&&&& \\
	\lstick{$\ket{0}$}\qw &\qw \gategroup[2,steps=6,style={rounded corners,fill=blue!20, inner xsep=2pt},background, label style={label position=below,anchor=north,yshift=-0.2cm}]{{\sc Block 1}}&\gate{H}&\ctrl{1}& \gate{R_1^\prime} \vcw{-1} & \gate{H}&\meter{} \vcw{-1}
	&& \lstick{$\ket{0}$} & \qw \gategroup[2,steps=6,style={rounded corners,fill=blue!20, inner xsep=2pt},background, label style={label position=below,anchor=north,yshift=-0.2cm}]{{\sc Block L}}&\gate{H} & \ctrl{1}&\gate{R_L^\prime} \vcw{-1} &\gate{H} &\meter{} \vcw{-1} &&&&&&&&&& \\
	\lstick{$\ket{1}$}&\qwbundle[alternate]{}&\qwbundle[alternate]{}&\gate{\mathcal{U}_1} \qwbundle[alternate]{}&\qwbundle[alternate]{} &\qwbundle[alternate]{}&\qwbundle[alternate]{} &\qwbundle[alternate]{}\ \ldots  & &\qwbundle[alternate]{} &\qwbundle[alternate]{} &\gate{\mathcal{U}_L} \qwbundle[alternate]{} &\qwbundle[alternate]{} &\gate{\text{Discard}}\qwbundle[alternate]{}&&&&&&&&&
\end{quantikz}
}
		 \newsubcap{A sequential variant of the order-finding algorithm, where the $R'_n$ denote phase gates that depend on the outcomes of the previous measurements and $\mathcal{U}_l=U^{2^{L-l}}_B$. See the main text and the SM for further details.}
		 \label{fig:ShorB}
     \end{subfigure}
\end{figure*}

We will now proceed to define the two measures that we will use to quantify the connection between coherence and the performance of Shor's algorithm. The robustness of coherence~\cite{Napoli2016} is defined as
\begin{align}
	C(\rho)=\min\left\{s\ge0: \frac{\rho+s \tau}{1+s}\in \mathcal{I}, \tau \text{ a state}\right\}.
\end{align} 
From this, we can define a dynamical measure that describes how well a channel can create coherence (i.e., with respect to $\CI$) via 
\begin{align}
	\mathscr{C}(\Theta)=\max_{\tau\in \mathcal{I}} C(\Theta\tau),
\end{align}
which is a resource generation capacity~\cite{Bennett2003,Mani2015,Xi2015,Garcia2016,Bu2017a,Bu2017b}.
For the detection-incoherent setting, the  NSID-measure~\cite{Theurer2019} (non-stochastic in detection) is a dynamical measure that describes how well a channel can detect coherence (i.e., with respect to $\DI$), and is given by
\begin{equation}\label{nSID}
	\tilde{M}_\diamond(\Lambda)=\min_{\Phi \in \DI} \max_{\sigma} \norm{\Delta(\Lambda-\Phi)\sigma}_1.
\end{equation}
Furthermore, we show in the SM that an intuitive candidate for a measure, namely
$\mathscr{D}(\Lambda)=\max_\rho \norm{\Delta\Lambda(\id-\Delta)\rho}_1$,
fails to form a measure in the $\DI$ setting, as it violates monotonicity. \newline

\section{Shor's algorithm} 
Let $N$ denote an integer to factor and assume $N$ to be large. The factorization problem can be reduced to the order-finding problem: given integers $N$ and $x$ where $x<N$ and $x$ coprime to $N$, the order $r$ is defined as the smallest integer such that $x^r=1 \, (\text{mod}\,N)$ (see Ref.~\cite{Shor1995} and the SM for more information).
Solving order-finding for a randomly chosen $x$ with the above properties allows to solve factoring with high probability, and it is exactly what the quantum parts of the various versions of Shor's algorithm accomplish efficiently. 

For the standard quantum order-finding protocol one uses two quantum systems $A$ and $B$ of dimension $q$ and $N$ respectively, where system $A$ consists of $L$ qubits with $N^2<q=2^L<2N^2$.  Furthermore, one defines a unitary by $U_B\ket n=\ket{x n \,(\text{mod}\, N)}$ that acts on system $B$ and the modular exponentiation via 
\begin{equation}
	U_c=\sum_{n=0}^{q-1} \ketbra{n}{n}_A \otimes U_B^n.
\end{equation} 
Important from a resource theoretical perspective, $U_c$ is both in $\DI$ and in $\CI$, i.e., it can neither produce nor detect coherence and is thus considered free in both resource theories.
As shown in Fig.~\ref{fig:ShorA}, an order-finding protocol works then as follows: Initialize system $AB$ in the state $\ket{0}_A^{\otimes L}\ket{1}_B$, first apply Hadamard gates to each qubit, apply  $U_c$, followed by an inverse Fourier transform $\mathcal{F}^\dagger $ on $A$ and then a measurement in the computational basis. This allows us to estimate a randomly chosen eigenvalue of $U_B$ with sufficiently high probability from the measurement outcome via the continued fraction algorithm and thus deduce $r$. 
Since both the quantum part (and in particular modular exponentiation and the inverse Fourier transform) as well as the classical pre- and post-processing can be implemented efficiently, this allows to factor in polynomial runtime ~\cite{Shor1995,Parker2000,Hardy1984,Nielsen2016}. A particular implementation of the Fourier transformStreltsov2017 given by sequentially applied controlled phase gates and Hadamard gates~\cite{Griffiths1995} allows to derive an equally efficient variant of the order-finding protocol that requires only a single control qubit which is being recycled~\cite{Parker2000}, see Fig.~\ref{fig:ShorB}. 

\section{Results}
We now describe the setup to which our results apply, namely the order-finding protocol depicted in Fig.~\ref{fig:ShorB}, and connect its performance to coherence. The quantum advantage in this protocol is obviously not emerging from the classical control and post-processing, so we will keep this part fixed. Now looking at a single block, we remind that the controlled unitary $\mathcal{U}_l=U^{2^{L-l}}_B$ as well as the phase gate $R'_l$ (see the SM for details) can neither create nor detect coherence and are thus free in both resource theories. Therefore, we will keep them fixed as well and treat them as a black-box that we can probe. The remaining ingredients of each block become the main focus of study: If we would replace the initial state $H\ket{0}=\ket{+}$ of the control qubit, which is a maximally coherent state~\cite{Baumgartz2014}, with an incoherent one, the block would be seriously flawed, in the sense that it does not encode information about $r$, since the black-box only affects the coherences of the control qubit (see the SM for more information). Incoherent and maximally coherent states are extreme cases, and to connect the performance of the algorithm quantitatively to coherence, we will investigate the impact on efficiency if we replace the initial control state with a partially coherent state. Since every quantum state can be identified with its replacements channel, we replace it with a fixed qubit channel $\Theta_l$ that is used to create an initial (partially coherent) control qubit state from an incoherent state $\sigma_l$. We further allow $\Theta_l$ to be transformed by arbitrary super-channels $S_1^{(l)}\in\CIS$ since this is free from a resource perspective and ensures that we use the resource at hand appropriately.
In this spirit, $S_1^{(l)}$ allows for a fair comparison of different resourceful operations. Note that another approach would be to optimize over different $U_l$ (see Refs.~\cite{Biswas2017,Masini2021} for related approaches in different settings).
\begin{figure}[t]
	\scalebox{0.69}{
	\begin{quantikz}
	&&&&&&\vcw{2}&&&&\vcw{2}&&&&&&&&&&&&&&\\
	&\gategroup[5,steps=10,style={rounded corners,fill=blue!20, inner xsep=2pt},background, label style={label position=below,anchor=north,yshift=-0.1cm}]{}&&&&&&&&&&\\
	\lstick{$\sigma_l$}&\qw &\gate[2,nwires={2}]{\Phi_1^{(l)}} \gategroup[2,steps=1,style={dashed,rounded corners,fill=cornellred!40, inner xsep=2pt},background, label style={label position=below,anchor=north,yshift=-0.1cm}]{}&\gate{\Theta_l} &\gate[2]{\Phi_2^{(l)}} \gategroup[2,steps=1,style={dashed,rounded corners,fill=cornellred!40, inner xsep=2pt},background, label style={label position=below,anchor=north,yshift=-0.1cm}]{}&\ctrl{3}&\gate{R_l^\prime}&\gate[2,nwires={2}]{\Phi_3^{(l)}} \gategroup[2,steps=1,style={dashed,rounded corners,fill=cornellred!40, inner xsep=2pt},background, label style={label position=below,anchor=north,yshift=-0.1cm}]{}&\gate{\Lambda_l}&\gate[2]{\Phi_4^{(l)}} \gategroup[2,steps=1,style={dashed,rounded corners,fill=cornellred!40, inner xsep=2pt},background, label style={label position=below,anchor=north,yshift=-0.1cm}]{}&\meter{}\\
	&&\gategroup[1,steps=3,style={dashed,rounded corners,fill=cornellred!40, inner xsep=2pt},background, label style={label position=below,anchor=north,yshift=-0.1cm}]{\sc{$S_1^{(l)}$}}&\qw&\qw &&&\gategroup[1,steps=3,style={dashed,rounded corners,fill=cornellred!40, inner xsep=2pt},background, label style={label position=below,anchor=north,yshift=-0.1cm}]{\sc{$S_2^{(l)}$}}& \qw&\qw&&&& \\
	&&&&&& &&&& &&&&&&&&&&&&&&\\
	\qwbundle[alternate]{}&\qwbundle[alternate]{}&\qwbundle[alternate]{}&\qwbundle[alternate]{}&\qwbundle[alternate]{}&\gate{\mathcal{U}_l}\qwbundle[alternate]{}&\qwbundle[alternate]{}&\qwbundle[alternate]{}&\qwbundle[alternate]{}&\qwbundle[alternate]{}&\qwbundle[alternate]{}&\qwbundle[alternate]{}
\end{quantikz}
	}
	\caption{A single modified block of the sequential order-finding algorithm with super-channels to make optimal use of the resources in the protocol.}\label{fig:SingleBlock} 
\end{figure}
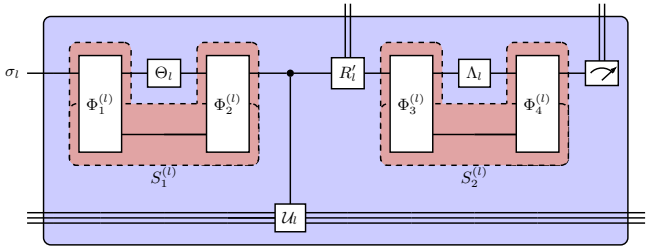
Furthermore, after the application of $\mathcal{U}_l$, we must extract the desired information, which is encoded exclusively in the coherences of the control qubit, hence we must detect coherence exactly in the sense that it makes a difference in the measurement statistics. The application of a Hadamard gate, which maximizes the NSID-measure among all qubit channels, is thus an extremal case too~\cite{Theurer2019}.

In contrast, a channel that cannot detect coherence would not be able to recover any of the available information on the prime factors. The ability to detect coherence, therefore, plays a vital role after the application of $\mathcal{U}_l$, and to investigate its precise contribution, we replace $H$ with a fixed channel $\Lambda_l$, that interpolates between the optimal $H$ and a completely incoherent measurement. We now allow to apply an arbitrary super-channel $S_2^{(l)}\in \DIS$ that is unitality preserving (we comment on this requirement in the SM), for the same reasoning as for $S_1^{(l)}$. The resulting block is depicted in Fig.~\ref{fig:SingleBlock}. To simplify our analysis for the main text, we further assume here that in each block, we use the same channel $\Theta/\Lambda$ for the creation/detection of coherence (see the SM for the more general version). For fixed $\Theta$ and $\Lambda$, we then define $P^{\suc}(\Theta,\Lambda)$  to be the probability (maximized over the $S_1^{(l)}$, $S_2^{(l)}$, and $\sigma_l$) that a single run of this order-finding protocol leads to the correct order and bound it by the following Theorem.
\begin{thm}\label{DynamicBound}
	The success probability of the order-finding protocol as described above with qubit operations $\Theta$ and unital $\Lambda$ for creation and detection respectively is lower bounded by
	\begin{align}\label{eq:main}
		P^{\suc}(\Theta,\Lambda) \ge& \frac{4}{\pi^2} \frac{\varphi(r)}{r} \left(\frac{1+\mathscr{C}(\Theta)\tilde{M}_\diamond(\Lambda)}{2} \right)^L,
	\end{align}
	where $\varphi(r)$ denotes Euler's totient function. 
\end{thm}
The presence of $\mathscr{C}(\Theta)$ is intuitive as it quantifies the ability of $\Theta$ to create coherence in the control qubit~\cite{Garcia2016,Bu2017a,Bu2017b}, which is exactly what we use the channel $\Theta$ for. We note that for any qubit channel $\Theta$, we have $\mathscr{C}(\Theta)\le1$, with equality if and only if $\Theta$ can create a maximally coherent qubit state~\cite{Bu2017b}. Moreover, for qubit operations $\Lambda$, $\tilde{M}_\diamond(\Lambda)\le 1$, and the bound is saturated for a Hadamard gate~\cite{Theurer2019}. The measures enter the bound on equal footing, which indicates that the ability to create and detect coherence are equally important, as one would intuitively expect. In case both $\Theta$ and $\Lambda$ are Hadamard gates, we thus recover the bound presented in Refs.~\cite{Shor1995,Parker2000}, which is used to prove the polynomial runtime of the algorithm. If the abilities to create and detect coherence decrease, this influences our bound exponentially in $L$. This suggests that the polynomial runtime of the fully coherent protocol becomes degraded exponentially in $L$ by the lack of coherence and the ability to detect it. However, one needs to ask whether the performance actually decreases exponentially with less coherence, or if only our bound does so. To address this question, we continue to present a sufficiently general upper bound.
\begin{thm}\label{thm:main2}
The success probability of the order-finding protocol as described above with qubit operations $\Theta$ and unital $\Lambda$ for creation and detection respectively is upper bounded by
\begin{equation}
\resizebox{.91\hsize}{!}{$P^{\suc}(\Theta,\Lambda)\leq  \min \left\lbrace \varphi(r)\left(1+\floor{\tfrac{2^L}{r^2}} \right) \left(\frac{1+\mathscr{C}(\Theta) \tilde{M}_\diamond(\Lambda)}{2}\right)^L , 1  \right\rbrace,$}
\end{equation}
where $\varphi(r)$ denotes Euler's totient function.
\end{thm}
We notice that this bound depends on both the problem and the employed coherence. The bound becomes trivial if the first term exceeds unit probability, which is sensitively dependent on the ratio of $2^L$ and $r$. Nevertheless, this is a rather gentle restriction on our upper bound, which can be justified by comparing the bounds on the resourceful success probability with the classical limit of the algorithm. We define the classical limit as the corresponding protocol if we are only allowed to use operations that cannot detect or create coherence, i.e., if both $\Theta$ and $\Lambda$ are free in their respective resource theories. In this case, we are in a classical regime and all states and operations can be reduced to probability vectors and stochastic matrices. The success probability is then determined by the uniform measurement statistic and the probability that the post-processing works, which results in
\begin{equation}
\resizebox{.88\hsize}{!}{$2\tfrac{\varphi(r)}{2^L}\floor{\tfrac{2^L-1}{2r^2}} \le P^\suc(\Theta_{\text{free}},\Lambda_{\text{free}})\le 2\tfrac{\varphi(r)}{2^L} \left(1+2\floor{\tfrac{2^L}{r^2}}\right)$},
\end{equation}
as we show in the SM. If we compare the bounds on the classical limit of the success probability with the one in Thm.~\ref{thm:main2}, we see that the same prefactor occurs. In this sense, the slightly limited upper bound in Thm.~\ref{thm:main2} can be regarded as an artifact of the problem-dependence.
If the fixed protocol does not perform well in the classical limit (which is the case of interest), we conclude that coherence is the quantum resource that determines the success probability by bounding it from below and above.

\section{Discussion and outlook}
In our work, we have used resource theories to derive quantitative upper and lower bounds on the success probability of the quantum part of a sequential
version of Shor’s algorithm in terms of measures of (dynamical)
coherence. Since the full algorithm repeats the quantum part until
it succeeds, this also quantifies the total run time and speed-up in
terms of the available resources. It is a novelty of our approach that
we do not only observe how a resource evolves or depletes during
an algorithm~\cite{Hillery2016,Anand2016,Shi2017,Liu2019YC} but determine quantitatively the performance advantage
that it grants. Moreover, our approach differs from Ref.~\cite{Josza2003}, where a necessary condition for the presence of a resource (here entanglement) to admit a speed-up in pure state quantum computing was derived. The argumentation of Ref.~\cite{Josza2003} is based on the observation that a quantum protocol with limited multi-partite entanglement operating on pure states can be simulated efficiently on a classical device. As already pointed out in Ref.~\cite{Josza2003}, this approach is incapable of establishing a (quantitative) sufficient condition for the contribution of entanglement as a resource as the presence of certain forms of large scale multi-partite entanglement can permit efficient classical simulation when employing a suitable mathematical data structure such as the stabilizer formalism~\cite{gottesman1997}. 

In contrast, we derive bounds that hold even for mixed states and show quantitatively
that coherence is necessary and sufficient to achieve an advantage over the classical limit of the investigated algorithm with a fixed overall structure. This, however, comes
at the price that at present these quantitative connections are tied to a specific family of factoring algorithms. Furthermore, we remark that whilst the way we fixed the overall structure of the protocol and our choice of the free operations is natural, well-motivated, and models the operations that are available to a classical computer, other choices may be considered too. Indeed, introducing restrictions that model the capabilities of a classical computer more accurately is an open problem that would lead to different (and potentially more involved) resource theories. As an example, one can additionally restrict the ability to preserve coherence~\cite{Saxena2020}, or more generally states~\cite{Rosset2018}. It is an interesting open question whether other restrictions and the corresponding resources would lead to relations comparable to those we found; see for instance the discussion in the SM why we did not choose operations that can neither create nor detect coherence as free.

A closely related question is to what extent the overall structure of the protocol can be generalized whilst still obtaining meaningful bounds. As we discuss in the SM in more detail, our findings hold for the standard parallel version of Shor’s algorithm as depicted in Fig.~\ref{fig:ShorA} too if the first register is in a product state and if the inverse Fourier transform is implemented in a way that leads to the sequential version. 
Indeed, generalizing the structure and choosing other free operations may reveal additional resources to underpin the efficiency of the quantum processor. One may for example argue that the implementation of the modular exponentiation, which is assumed to be free in our framework, does carry a cost. Relaxing this assumption may establish entanglement as a resource that bounds the efficiency of the protocol. However, as incoherent operations such as the modular exponentiation can convert coherence to entanglement~\cite{Streltsov2015,Egloff2018,Theurer2020} it may also be possible to reduce the resource entanglement to coherence when it comes to computation.  In summary, our results depend on the choice of free operations and overall structure and we do not claim that coherence is {\em the} quantum resource for factoring alone, but we show that it is {\em a} quantum resource that lower and upper bounds the performance. In fact, it might well be that other resources not captured in our framework contribute (in other factoring algorithms) too. Exploring this is an interesting starting point for future work.

Furthermore, using our technique to fix the structure of a protocol and to define a free limit, one can investigate the role of quantum resources in other quantum algorithms too. Since general statements about the role of quantum resources in computation are often out of reach, such an algorithm and implementation specific approach might lead to further insights into the value of quantum resources in computation, which might help to understand separation between classical and quantum computing.

\begin{acknowledgments}
This work was supported in part by the BMBF project PhotonQ (Förderkennzeichen 13N15762). JMM acknowledges support from CONICET Argentina. 
TT acknowledges support from a 2020 Eyes High Postdoctoral Match-Funding Fellowship. TT’s research conducted for this paper is partially supported by the Pacific Institute for the Mathematical Sciences (PIMS). The research and findings may not reflect those of the Institute. For figures and numerics, we utilized Refs.~\cite{Kay2018,Lofberg2004,Sturm1999,cvx}.
\end{acknowledgments}
%

\appendix
\pagebreak
\widetext
\newpage


	\begin{center}
		\textbf{On the Role of Coherence in Shor's Algorithm\\
			Supplemental Material}
	\end{center}

	In this Supplemental Material, we give the proofs of the results presented in the main text and some further information. This includes additional dynamic resource measures and their properties.

\section{On the appearing measures}
In this section, we present properties of the resource measures employed in our analysis of the performance of Shor's algorithm. To begin with, we discuss the functional 
\begin{align}
	\mathscr{D}(\Lambda)=\max_{\rho} \norm{\Delta\Lambda(\id-\Delta)\rho}_1,
\end{align}
which is interesting from a resource theoretical perspective and we will see later that $\mathscr{D}(\Lambda)$ appears naturally when connecting the success probability of the investigated order-finding protocol with the ability to detect coherence. Moreover, it seems to be a natural candidate for a resource measure under detection incoherent operations. However, it is not monotonic under $\DIS$ as we show here. To relate the performance of Shor's algorithm with a rigorous dynamical measure in Thm.~\ref{Thm:LowerBound} and Thm.~\ref{thm:UpperBoundSM} we make use of the fact that the functional $\mathscr{D}$ shares sufficient similarities with the NSID measure $\tilde{M}_\diamond(\Lambda)=\min_{\Phi\in \DI}\max_{\rho} \norm{\Delta(\Lambda-\Phi)\rho}_1$. In particular, that $\mathscr{D}$ provides an upper bound on $\tilde{M}_\diamond$ and that the two functionals coincide on qubit channels.

\begin{prop}\label{prop:Properties}
	Let $\Phi\in \DI$. The functional $\mathscr{D}(\Lambda)=\max_{\rho} \norm{\Delta\Lambda(\id-\Delta)\rho}_1$ has the following properties:
	\begin{enumerate}
		\item $\mathscr{D}(\Lambda)=0 \Leftrightarrow \Lambda \in\DI$,
		\item $\mathscr{D}(\Lambda^{C\leftarrow A} \otimes \id_B)=\mathscr{D}(\Lambda^{C\leftarrow A})$,
		\item $\mathscr{D}(\Phi\Lambda)\le\mathscr{D}(\Lambda)$,
		\item $\mathscr{D}$ is convex.
	\end{enumerate}
\end{prop}
\begin{proof}
	Let us begin by pointing out that convexity in the argument follows from the convexity of the trace norm itself. Notice that for any $\Lambda\in \DI$, i.e., $\Delta\Lambda=\Delta\Lambda\Delta$, it follows that the functional vanishes. Furthermore, for any detecting channel $\Lambda$ there exists some state $\rho$ such that $\Delta\Lambda(\id-\Delta)\rho\neq 0$, which proves faithfulness since $\norm{\cdot}_1$ is a norm.\newline
	The functional behaves monotonically under post-processing with free channels even with an identity channel attached in parallel. Using~\cite[Lem.~14]{Theurer2019} for the inequality, we see that
	\begin{equation}
		\begin{split}
			\mathscr{D}(\Lambda^{C\leftarrow A} \otimes \id_B)&= \max_{\rho_{AB}} \norm{\Delta_{CB} (\Lambda^{C\leftarrow A}\otimes \id_B)(\id_{AB}-\Delta_{AB})\rho_{AB}}_1\\
			&=\max_{\rho_{AB}} \norm{\Delta_{CB} (\Lambda^{C\leftarrow A}\otimes \id_B)(\id_A\otimes \id_B-\Delta_A\otimes \Delta_B)\rho_{AB}}_1 \\
			&=\max_{\rho_{AB}} \norm{\left[(\Delta_C\Lambda^{C\leftarrow A}-\Delta_C\Lambda^{C\leftarrow A}\Delta_A)\otimes\id_B \right](\id_A\otimes \Delta_B)\rho_{AB}}_1 \\
			&\leq \max_{\ket{\psi},\ket{i}} \norm{(\Delta_C\Lambda^{C\leftarrow A}-\Delta_C\Lambda^{C\leftarrow A}\Delta_A)\ketbra{\psi}{\psi}_A\otimes \ketbra{i}{i}}_1\\
			&=\max_{\ket{\psi},\ket{i}} \norm{(\Delta_C\Lambda^{C\leftarrow A}-\Delta_C\Lambda^{C\leftarrow A}\Delta_A)\ketbra{\psi}{\psi}_A}_1 \norm{ \ketbra{i}{i}}_1\\
			&=\max_{\ket{\psi}} \norm{(\Delta_C\Lambda^{C\leftarrow A}-\Delta_C\Lambda^{C\leftarrow A}\Delta_A)\ketbra{\psi}{\psi}_A}_1, 
		\end{split}
	\end{equation}
	which coincides with $\mathscr{D}(\Lambda^{C\leftarrow A})$ due to the convexity of the trace norm. The reverse inequality follows from restricting $\rho_{AB}$ to product states in the first line.
	
	Additionally, the properties of the trace norm allow us to write
	\begin{equation}
		\begin{split}
			\mathscr{D}(\Phi\Lambda)&=\max_{\rho} \norm{\Delta\Phi\Lambda(\id-\Delta)\rho}_1\\
			&=\max_{\rho} \norm{\Delta\Phi\Delta\Lambda(\id-\Delta)\rho}_1 \\
			&\leq \max_{\rho} \norm{\Delta\Lambda(\id-\Delta)\rho}_1.
		\end{split}
	\end{equation}
\end{proof}

To show that $\mathscr{D}$ is a resource measure, we would need in addition that $\mathscr{D}(\Lambda\Phi)\le\mathscr{D}(\Lambda)$. However, this condition is violated in general for $\Phi\in\DI$. To prove this, we first describe how we can evaluate $\mathscr{D}$ numerically.

\begin{prop}\label{theo:semi_pre}
	Consider a quantum channel $\Theta^{C\leftarrow B}$ and let $N=\dim(C)$. Let further $(s_{m,n})_{m,n}$ be the matrix of dimension $2^N\times N$ that contains as rows all $N$-dimensional  vectors $\vec{s}_m$  whose entries are $\pm1$. The numerical value of $\mathscr{D}(\Theta^{C\leftarrow B} )$ is then equivalent to the maximum of the solutions of the following $2^N$ semidefinite programs (each for a fixed $m$)
	\begin{align}\label{eq:semied}
		\begin{split}
			\text{maximize:}\quad\, &t_m=\sum_{n=0}^{N-1} s_{m,n} \bra{n}_C \Theta^{C\leftarrow B} \left(\id-\Delta\right)\sigma \ket{n}_C,\\		
			\text{subject to:}\quad & \sigma_B\geq 0, \\
			& \Tr{\sigma_B}=1.
		\end{split}
	\end{align} 
\end{prop}
\begin{proof}
	We use that for real $f_n$
	\begin{equation}
		\sum_n|f_n|=\max_{\vec{s}_m} (\vec{s}_m\cdot\vec{f}),
	\end{equation}
	where the vectors $\vec{s}_m$ have been introduced in the statement of the Proposition. In addition
	\begin{align}
		\mathscr{D}(\Theta^{C\leftarrow B} )= \max_{\sigma}\norm{\Delta \Theta^{C\leftarrow B}  \left(\id-\Delta\right)\sigma}_1=\max_{\sigma}\sum_n \left|\bra{n}\Theta^{C\leftarrow B}  \left(\id-\Delta\right)\sigma\ket{n}\right|.
	\end{align}
\end{proof}
This method of evaluating $\mathscr{D}$ is certainly not the most efficient. However, with the help of the following Proposition, it allows us to disprove monotonicity.

\begin{prop} \label{prop:optimPhi}
	Let $\Theta^{C\leftarrow B}$ be a quantum channel and $A$ a third and fixed quantum system. Denote by $\mathcal{S}$ the set of all diagonal matrices of dimension $\text{dim}(C)$ with diagonal elements $\pm1 $, and by $M_A$ the matrix on system $A$ with all entries equal to one. Furthermore, we define 
	\begin{align}
		\mathcal{X}=\{X=\left(M_A-\id_A\right)\otimes\Theta^\dagger S: S\in\mathcal{S}\}.
	\end{align}
	The solution of the maximization problem 
	\begin{align}
		\max_{\sigma,\Phi^{B\leftarrow A}\in\mathcal{DI}}\norm{\Delta \Theta^{C\leftarrow B} \Phi^{B\leftarrow A} \left(\id-\Delta\right)\sigma}_1
	\end{align}
	is then given by the maximum of the solutions of the following (finite number of) semidefinite programs to be evaluated for all fixed $X\in\mathcal{X}$
	\begin{align}
		\begin{split}
			\text{maximize:}\quad\, &\Tr{X Y_{AB}},\\		
			\text{subject to:}\quad & Y\ge0,\\
			&\partTr{B}{Y}=\Delta \sigma_A,\ \\
			&\sigma_A\geq 0, \\
			& \Tr{\sigma_A}=1, \\
			&\diag\left(\bra{i}_AY\ket{j}_A\right)=0 \ \forall i\ne j.
		\end{split}
	\end{align} 
\end{prop}
\begin{proof}
	We write quantum states as 
	$\rho=\sum_{i,j}\rho_{ij}\ketbra{i}{j}$ and the action of quantum channels as $\Phi(\ketbra{i}{j})=\sum_{k,l} \Phi^{ij}_{kl}\ketbra{k}{l}$. With this notation, we find
	\begin{align}
		\norm{\Delta \Theta \Phi \left(\id-\Delta\right)\rho}_1&=\norm{\Delta \Theta \left[\sum_{i\ne j}\sum_{kl}\Phi_{kl}^{ij}\ketbra{k}{l}\rho_{ij} \right]}_1 \nonumber \\
		&=\norm{\Delta \Theta \left[\sum_{i\ne j} \bra{i}_A \left(\sum_{op} \rho_{op} \ketbra{o}{p}_A\otimes\sum_{kl} \Phi_{kl}^{op}\ketbra{k}{l}_B\right)\ket{j}_A \right]}_1\nonumber \\
		&=\norm{\Delta \Theta \left[\sum_{i\ne j} \bra{i}_A \left( \left(\id^A \otimes \Phi^{B\leftarrow \tilde{A}} \right)\sum_{op} \rho_{op} \ketbra{oo}{pp}_{A\tilde{A}}\right)\ket{j}_A \right]}_1.
	\end{align}
	Using \cite[Lem.~12]{Masini2021} we thus have
	\begin{align}
		\max_{\sigma,\Phi^{B\leftarrow A}\in\mathcal{DI}}\norm{\Delta \Theta^{C\leftarrow B} \Phi^{B\leftarrow A} \left(\id-\Delta\right)\sigma}_1
		=\max_{Y_{AB}\in\mathcal{Y}} \norm{\Delta \Theta^{C\leftarrow B} \left[\sum_{i\ne j} \bra{i}_A Y_{AB}\ket{j}_A \right]}_1,
	\end{align}
	where the set $\mathcal{Y}$ is defined as
	\begin{align}
		\mathcal{Y}:= &\{Y_{AB} | Y\ge0,\ \text{Tr}_B(Y)=\Delta \sigma_A,\ \diag\left(\bra{i}_AY\ket{j}_A\right)=0 \ \forall i\ne j,\ \sigma_A  \text{ quantum state} \}
	\end{align}
	and therefore characterized by semidefinite constraints. Using the absolute value technique from Prop.~\ref{theo:semi_pre}, we can solve this optimization problem via a set of SDPs: With the definitions from the Proposition, we note that $\sum_{i\ne j} \bra{i}_A Y_{AB}\ket{j}_A=\text{Tr}_A \left[\left(M_A-\id_A\right)\otimes\id_B Y_{AB}\right]$ and therefore 
	\begin{align}
		&\max_{Y_{AB}\in\mathcal{Y}} \left|\left|\Delta \Theta^{C\leftarrow B} \left[\sum_{i\ne j} \bra{i}_A Y_{AB}\ket{j}_A \right]\right|\right|_1 \nonumber \\
		&=\max_{S \in\mathcal{S}}\max_{Y_{AB}\in\mathcal{Y}} \text{Tr}\left(S  \Theta^{C\leftarrow B} \left[\sum_{i\ne j} \bra{i}_A Y_{AB}\ket{j}_A \right]\right) \nonumber \\
		&=\max_{S \in\mathcal{S}}\max_{Y_{AB}\in\mathcal{Y}} \text{Tr}\left(\left[\Theta^\dagger S  \right]  \left[\sum_{i\ne j} \bra{i}_A Y_{AB}\ket{j}_A \right]\right) \nonumber \\
		&=\max_{S \in\mathcal{S}}\max_{Y_{AB}\in\mathcal{Y}} \text{Tr}\left(\left[\Theta^\dagger S  \right] \text{Tr}_A\left[\left(M_A-\id_A\right)\otimes\id_B Y_{AB}\right]\right) \nonumber \\
		&=\max_{S \in\mathcal{S}}\max_{Y_{AB}\in\mathcal{Y}} \text{Tr}\left(\left[\id_A\otimes\Theta^\dagger S  \right] \left(M_A-\id_A\right)\otimes\id_B Y_{AB}\right)\nonumber\\
		&=\max_{X\in\mathcal{X}} \max_{Y_{AB}\in\mathcal{Y}} \text{Tr}\left(X  Y_{AB}\right).
	\end{align}
\end{proof}
Due to this Proposition, for every fixed system $A$, we can evaluate 
\begin{align}
	\tilde{\mathscr{D}}_A(\Theta^{C\leftarrow B}):=\max_{\Phi^{B\leftarrow A}\in\DI} \mathscr{D}(\Theta^{C\leftarrow B} \Phi^{B\leftarrow A})= \max_{\Phi^{B\leftarrow A}\in\DI} \max_{\sigma_A} \norm{\Delta\Theta^{C\leftarrow B} \Phi^{B\leftarrow A}(\id_A-\Delta_A)\sigma_A}_1
\end{align}
numerically by solving a collection of SDPs. We now move to the equivalence on qubit channels, which we will use to connect the performance of Shor's algorithm to the ability to detect coherence.

\begin{lem}\label{lem:DetectabilityNSID}
	Let $\Lambda$ be any qubit channel defined in the index representation as $\Lambda(\ketbra{n}{m})=\sum_{kl} \Lambda_{kl}^{nm} \ketbra{k}{l}$.
	
	Then the functional $\mathscr{D}$
	\begin{enumerate}
		\item  is  given by $\mathscr{D}(\Lambda)=2|\Lambda_{00}^{01}|$,
		\item coincides with the NSID measure~\cite{Theurer2019} $\tilde{M}_\diamond(\Lambda)=\min_{\Phi \in \DI} \max_\rho \norm{\Delta(\Lambda-\Phi)\rho}_1$.
	\end{enumerate}
\end{lem}
\begin{proof}
	Due to convexity of the trace norm, the optimization in the definition of $\mathscr{D}(\Lambda)$ over all states can be reduced to pure states $\rho=\ketbra{\psi}{\psi}$. We expand a pure qubit state as $\ket{\psi}=p_0 \ket{0}+p_1 e^{i\phi} \ket{1}$ and write $\Lambda_{kl}^{nm}=|\Lambda_{kl}^{nm}|e^{i\lambda_{kl}^{nm}}$. According to \cite[Lem.~6]{Masini2021},  $|\lambda_{00}^{01}|=-|\lambda_{00}^{10}|=:\lambda$ and $|\Lambda_{00}^{01}|=|\Lambda_{00}^{10}|$. For the first part, from $\Delta\Lambda(\id-\Delta)$ being diagonal it follows then straightforward that
	\begin{equation}
		\begin{split}
			\mathscr{D}(\Lambda)&= \max_{\rho} \norm{\Delta\Lambda(\id-\Delta)\rho}_1 =\max_{\rho} \,\, 2 \left| \bra{0}\Delta\Lambda(\id-\Delta)\rho \ket{0} \right |\\
			&=\max_{\lbrace p_n, \phi \rbrace} 2\left| \sum_{n\neq m} \sqrt{p_np_m} e^{i\phi(n-m)} \Lambda_{00}^{nm} \right| =\max_{\lbrace p_n, \phi \rbrace} 2\left| \sum_{n\neq m} \sqrt{p_np_m} e^{i\phi(n-m)} |\Lambda_{00}^{nm}| e^{i\lambda_{00}^{nm}} \right| \\
			&=\max_{\lbrace p_n, \phi \rbrace} 2\left| \sum_{n\neq m} \sqrt{p_np_m} e^{i(\phi+\lambda)(n-m)} |\Lambda_{00}^{nm}| \right|\\
			&=\max_{\lbrace p_n, \phi \rbrace}  4\sqrt{p_0p_1} \cos(\phi+\lambda) |\Lambda_{00}^{01}|, \\
			&=2 |\Lambda_{00}^{01}|.
		\end{split}
	\end{equation}
	Secondly, for a qubit map $\Lambda$, we find that
	\begin{equation}
		\begin{split}
			\tilde{M}_\diamond(\Lambda)&=\min_{\Phi \in \DI} \max_\rho \norm{\Delta(\Lambda-\Phi)\rho}_1\\
			&= \min_{\Phi \in \DI} \max_\rho \, 2 \left| \bra{0} \Delta(\Lambda-\Phi)\rho \ket{0} \right| \\
			&= \min_{\Phi \in \DI} \max_{\lbrace p_n,\phi \rbrace } \,2 \left| \sum_n p_n (\Lambda_{00}^{nn}-\Phi_{00}^{nn}) +\sum_{n\neq m}\sqrt{p_np_m} e^{i\phi(n-m)} \Lambda_{00}^{nm}   \right|\\
			&= \min_{\Phi \in \DI} \max_{\lbrace p_n,\phi \rbrace } \,2 \left| \sum_n p_n (\Lambda_{00}^{nn}-\Phi_{00}^{nn}) +\sum_{n\neq m} \sqrt{p_np_m} e^{i(\phi+\lambda)(n-m)} |\Lambda_{00}^{nm}|  \right|.  
		\end{split}
	\end{equation}
	
	Now we first consider the inner optimization problem, i.e., we fix $\Phi$. The first sum always evaluates to a real number, and the phase $\phi$ only appears in the second sum. Let us assume that the first sum is positive. The optimum over $\phi$ is then obviously achieved for $\phi=-\lambda$. If the first sum is negative, the optimum is $\phi=\pi-\lambda$. In both cases, we have 
	\begin{equation}
		\begin{split}
			\tilde{M}_\diamond(\Lambda)&=  \min_{\Phi \in \DI} \max_{\lbrace p_n\rbrace } \, 2 \left(\left| \sum_n p_n (\Lambda_{00}^{nn}-\Phi_{00}^{nn})\right| + 2\sqrt{p_0p_1} |\Lambda_{00}^{01}| \right) \\
			&\ge \max_{\lbrace p_n\rbrace } \, 4\sqrt{p_0p_1} |\Lambda_{00}^{01}| = 2| \Lambda_{00}^{01}|\\
			&=\mathscr{D}(\Lambda).
		\end{split}
	\end{equation}
	Since $\Lambda \Delta \in \DI$, we also have 
	\begin{equation}
		\begin{split}
			\tilde{M}_\diamond(\Lambda)&=\min_{\Phi \in \DI} \max_\rho \norm{\Delta(\Lambda-\Phi)\rho}_1\\
			&\le \max_\rho \,\norm{\Delta(\Lambda-\Lambda\Delta)\rho}_1 \\
			&= \mathscr{D}(\Lambda)
		\end{split}
	\end{equation}
	and find 
	\begin{equation}
		\tilde{M}_\diamond(\Lambda)=\mathscr{D}(\Lambda) =2| \Lambda_{00}^{01}|
	\end{equation}
	for qubit channels $\Lambda$.
\end{proof}

Furthermore, this allows us to prove that $\mathscr{D}(\Lambda)=\max_{\rho} \norm{\Delta\Lambda(\id-\Delta)\rho}_1$ fails to form a measure as defined in the main text.

\begin{prop}
	There exist $\Phi \in \DI, \Theta$ such that $\mathscr{D}(\Theta)<\mathscr{D}(\Theta\Phi)$.
\end{prop}
\begin{proof}
	Let $\Theta^{B\leftarrow B}$ be defined via the two Kraus operators
	\begin{align}
		K_1=\frac{1}{\sqrt{2}}\begin{pmatrix}
			i & 1 \\
			0 & 0
		\end{pmatrix}, \quad K_2=\frac{1}{\sqrt{2}}\begin{pmatrix}
			0 & 0 \\
			1 & i
		\end{pmatrix},
	\end{align}
	i.e., $|B|=2$ (and it is straightforward to check that this defines indeed a CPTP map). With the help of Lem.~\ref{lem:DetectabilityNSID}, we find $\mathscr{D}(\Theta)=1$. Choosing $|A|=3$, we can use  Prop.~\ref{prop:optimPhi} to evaluate $\tilde{\mathcal{D}}_A(\Theta^{B\leftarrow B})$ numerically. Moreover, it is possible to extract optimal $\Phi^{B\leftarrow A}$ and $\sigma_A$ from the solution of the semidefinite program. An optimal choice consists of
	\begin{align}
		\sigma_A=\frac{1}{3}\begin{pmatrix}
			1 & 1 & 1 \\
			1 & 1 & 1 \\
			1 & 1 & 1
		\end{pmatrix},
	\end{align}
	which is a maximally coherent state and the quantum operation $\Phi^{C\leftarrow A}$ given by its Choi state
	\begin{align}
		J_{\Phi}=(\id_A\otimes \Phi) \sum_{n,m=1}^3\ketbra{nn}{mm}=:\sum_{n,m=1}^3\ketbra{n}{m}_A\otimes X_{nm}
	\end{align}
	where
	\begin{align}
		X_{nn}=\begin{pmatrix}
			1/2 & i/6 \\
			-i/6& 1/2
		\end{pmatrix}
	\end{align}
	and
	\begin{align}
		X_{nm}= \begin{pmatrix}
			0 & -i/3 \\
			i/3& 0
		\end{pmatrix}
	\end{align}
	for $n\ne m$.
	It is straightforward to check that $J_{\Phi}$ is hermitian, has eigenvalues $(0,0,0,1,1,1)$, and
	\begin{align}
		\text{Tr}_B\left(J_\Phi\right)=\id_A,
	\end{align}
	i.e., $\Phi^{C\leftarrow A}$ is CPTP. Moreover, due to
	\begin{align}
		&J_{\Delta\Phi\Delta}=\Delta_{AB} J_\Phi=\frac{1}{2}\id_{AB}=\Delta_B J_\Phi=J_{\Delta\Phi},
	\end{align}
	$\Phi\in \mathcal{DI}$. We are not going to prove optimality of $\Phi$ and $\sigma$, for example by deriving the dual program, but rather note that
	\begin{align}
		\tilde{\mathcal{D}}_A(\Theta^{B\leftarrow B})=&\max_{\tilde{\Phi}^{B\leftarrow A}\in\mathcal{DI}}\max_{\tilde{\sigma}_A}\norm{\Delta \Theta^{B\leftarrow B} \tilde{\Phi}^{B\leftarrow A} \left(\id-\Delta\right)\tilde{\sigma}_A}_1 \nonumber \\
		\ge& \norm{\Delta \Theta^{B\leftarrow B} \Phi^{B\leftarrow A} \left(\id-\Delta\right)\sigma_A}_1 \nonumber \\
		=&\norm{\Delta \Theta^{B\leftarrow B}\ \text{Tr}_A \left[\left(\left(\left(\id-\Delta\right)\sigma_A\right)^T\otimes \id_B \right) J_\Phi\right]}_1 \nonumber \\
		=&\norm{\Delta \Theta^{B\leftarrow B}\ \text{Tr}_A \left[\left(\left(\id-\Delta\right)\sigma_A\otimes \id_B\right) J_\Phi\right]}_1 \nonumber \\
		=& \norm{\Delta \Theta^{B\leftarrow B} \frac{1}{3}\sum_{\substack{n,m=0\\ n\ne m}}^{3} X_{ij} }_1 \nonumber \\
		=&\frac{2}{3} \norm{\Delta \Theta^{B\leftarrow B} \begin{pmatrix}
				0 & -i \\
				i &  0
		\end{pmatrix}}_1 \nonumber \\
		=& \frac{2}{3} \norm{ \begin{pmatrix}
				1 & 0 \\
				0 &  -1
		\end{pmatrix}}_1 \nonumber \\
		=&\frac{4}{3}>1,
	\end{align}
	which finishes the proof by giving an explicit example.
\end{proof}

Whilst this is not the purpose of this Letter, we note that it is straightforward to show that the family of functionals
\begin{align}
	\tilde{\mathscr{D}}_A(\Theta^{C\leftarrow B})= \max_{\Phi^{B\leftarrow A}\in\DI} \max_{\sigma_A} \norm{\Delta\Theta^{C\leftarrow B} \Phi^{B\leftarrow A}(\id_A-\Delta_A)\sigma_A}_1
\end{align}
defines resource measures in the detection-incoherent setting in itself. We notice similarities to the measures in Ref.~\cite{Masini2021}, but leave further investigations on these measures, for example on an operational interpretation, to future work.

\section{Shor's factorization algorithm}	
In this section we review Shor's algorithm, beginning with the basic prerequisites in number theory, moving on to Shor's protocol and a sequential version introduced in Ref.~\cite{Parker2000}. Additionally, the fine-tuned interplay of the quantum part and the classical post-processing via the continued fraction algorithm is discussed in detail, paving the way for a discussion of the protocols investigated in this work. Some notable examples of further reading on Shor's algorithm are the articles~\cite{Shor1995, Parker2000} and the textbook~\cite{Nielsen2016}, on which the following brief review is based. \newline

\subsection{Reduction to order-finding} The first step in Shor's algorithm is the reduction of the integer factorization problem to the so-called order-finding problem~\cite{Shor1995}. Let $N$ denote the integer to be factorized, which consists of $m$ distinct prime factors and can be represented in an $n$ bit string. Furthermore, let $x$ be an integer $1\leq x<N$ with $x$ coprime to $N$, i.e., $x$ and $N$ share no common factor. The order-finding problem is then to find the \textit{smallest} integer $r$ such that $x^r= 1 \,\text{mod}\, N$. This integer $r$ is called the order of $x$ modulo $N$. The reduction of factoring to order-finding results from the following two statements. We omit the proofs at this point, for further reading see for example Ref.~\cite{Nielsen2016}. 

\begin{lem}\label{gcdLEmma}
	Given a composite (with more than one distinct prime factor), odd integer $N$ and an integer solution $a$ with $1\leq a<N$ to the equation $a^2=1 \,\text{mod}\, N$, that is non-trivial, i.e., $a\neq \pm 1 \,\text{mod} \, N$, then at least one of $\text{gcd}(a\pm 1,N)$ is a non-trivial factor of $N$.
\end{lem}\vspace{0.4cm}
\begin{lem}\label{OrderEven}
	For a uniformly chosen $x$ in the range $1\leq x<N$ and coprime to $N$, the probability that the order $r$ of $x$ modulo $N$ is even and non-trivial is bounded by
	\begin{equation}
		P(r \,\,\text{even},\,\text{and}\,\, x^{r/2}\neq -1\,\text{mod}\, N)>1-\frac{1}{2^m},
	\end{equation}
	where $m$ is the number of distinct prime factors of $N$.
\end{lem}\vspace{0.4cm}
With this at hand, a factorization algorithm is given by the following procedure: In a first step, catch exceptions like $N$ having two as a (multiple) prime factor and check if $N$ is a composite integer, i.e., has more than one distinct prime factor. This can be done efficiently on a classical device, see Ref.~\cite{Nielsen2016}. These two steps guarantee that the prerequisites of Lem.~\ref{gcdLEmma} and Lem.~\ref{OrderEven} are satisfied. In the next step, choose a random $x$, check if it is coprime to $N$, otherwise, repeat until it is. The bottleneck of the algorithm is the order-finding, but assuming we can solve this in polynomial time, determine the order $r$ and subsequently check if it is even and non-trivial (which has sufficiently high probability due to Lem.~\ref{OrderEven}). If so,
compute $a=x^{r/2}$ (note that $x^{r/2}$ cannot be $1 \text{ mod } N$ due to the definition of the order) and use Lem.~\ref{gcdLEmma} to find a factor of $N$, otherwise, repeat. The algorithm is run until all prime factors have been found.
Since the greatest common divisor can be computed efficiently in polynomial time in the input length $n$ (for example using Euclid's algorithm), having a polynomial time algorithm for order-finding results in a polynomial time algorithm for factorization. 
\newline
\subsection{Order-finding \`a la Shor} 
Shor's coup of an efficient order-finding protocol, depicted schematically in Fig.~\ref{fig:ShorDecomposition}, is at the heart of the factorization algorithm. This \textit{standard} protocol for order-finding utilizes two quantum systems $A$ and $B$ of dimension $q$ and $N$ respectively, where system $A$ consists of $L$ qubits such that $N^2<q=2^L<2N^2$, with $N$ being the number to factor. Along with the classical post-processing via the continued fraction algorithm, the quantum part of the protocol can be separated into three essential ingredients: preparation of an initial state, then the so-called modular exponentiation, and a  measurement. The modular exponentiation is defined by the controlled-like unitary
\begin{equation}\label{eq:ControlledUnitary}
	U_c=\sum_{n=0}^{q-1} \ketbra{n}{n}_A \otimes U_B^n,
\end{equation}
where $U_B \ket{n}_B= \ket{xn\, \text{mod}\, N}_B$. Note that the modular exponentiation can be implemented in polynomial time~\cite{Shor1995, PhysRevA.54.147,PhysRevA.54.1034,zalka1998fast}. The modular exponentiation encodes information about the order $r$ into the state of system $A$, only requiring knowledge about $x$ and the number $N$ to be factored. The encoding of this information depends on the initial state of the auxiliary system $B$, and a convenient choice is the state $\ket{1}_B$. Let us emphasize that other incoherent states can be used as well. For instance in Ref.~\cite{Parker2000}, it is shown that choosing a normalized maximally mixed initial state $\id_B$ will increase the runtime of the algorithm at most polynomially. In fact, for factorization problems of the form $N=pq$, where $p$ and $q$ are primes, the increase is asymptotically negligible.
After performing the modular exponentiation, the auxiliary system is discarded. For our purposes, the action of the modular exponentiation on system $A$ will be fixed and labeled by $\mathcal{E}$. This channel $\mathcal{E}$ admits the following simple structure. 
\begin{lem}\label{Encoding}
	If system $B$ is in the state $\ket{1}_B$, then the effect of the modular exponentiation on system $A$ is given by
	\begin{equation}
		\mathcal{E}(\rho_A)= \frac{1}{r}\sum_{j=0}^{r-1} \mathcal{E}_j(\rho_A) \quad \text{with} \quad \mathcal{E}_j(\rho_A)= R_{j/r} \rho_A R_{j/r}^\dagger,
	\end{equation}
	where the $R_{j/r}$ denote rotations around multiples of the fraction of $r$, i.e., $R_{j/r}=\sum_n e^{2\pi i\tfrac{j}{r} n} \ketbra{n}{n}$. 
\end{lem}
\begin{proof}
	Notice that by definition of the order-finding problem $x^r=1\, \text{mod}\, N$. It follows that $U^r_B=\id_{B}$, since $\forall n$ we find $U^r_B \ket{n}_B=\ket{x^r n \, \text{mod}\, N}_B=\ket{(x^r\, \text{mod}\, N)(n \, \text{mod}\, N)\,\text{mod}\, N}_B=\ket{n}_B$. Hence, orthonormal eigenstates $\ket{\psi_j}_B$ of $U_B$ are simply given by
	\begin{equation}
		\ket{\psi_j}_B= \frac{1}{\sqrt{r}} \sum_{l=0}^{r-1} e^{-2\pi i l \tfrac{j}{r}} \ket{x^l\, \text{mod}\, N}_B,
	\end{equation}
	with corresponding eigenvalues of $e^{2\pi i \tfrac{j}{r}}$. This allows us to expand the auxiliary state as $\ket{1}_B= \tfrac{1}{\sqrt{r}} \sum_{j=0}^{r-1} \ket{\psi_j}_B$. With this at hand, it is straightforward to calculate
	\begin{equation}
		\begin{split}
			\mathcal{E}(\rho_A)&=\text{Tr}_B\left[{U_c(\rho_A\otimes \ketbra{1}{1}_B) U_c^\dagger}\right] \\
			&=\text{Tr}_B\left[ \sum_{n,m}\rho_{nm} \ketbra{n}{m}_A \otimes  \frac{1}{r}\sum_{j,j^\prime=0}^{r-1} U^n_B \ketbra{\psi_j}{\psi_{j^{\prime}}}_B(U^m_B)^\dagger\right] \\
			&=\text{Tr}_B\left[ \sum_{n,m}  \rho_{nm} \ketbra{n}{m}_A \otimes  \frac{1}{r}\sum_{j,j^\prime=0}^{r-1} e^{2\pi i(n \tfrac{j}{r}- m\tfrac{j^{\prime}}{r})} \ketbra{\psi_j}{\psi_{j^{\prime}}}_B\right] \\
			&=\frac{1}{r}\sum_{j=0}^{r-1} \sum_{n,m} \rho_{nm} e^{2\pi i \tfrac{j}{r} (n-m)} \ketbra{n}{m}_A =\frac{1}{r} \sum_{j=0}^{r-1} \mathcal{E}_j(\rho_A).\qedhere
		\end{split}
	\end{equation}
\end{proof}
Let us emphasize the resemblance of $\mathcal{E}$ to a symmetry operation that gives rise to the resource theory of asymmetry~\cite{Gour_2008, PhysRevA.80.012307, Marvian_2013}. In this particular case, the symmetry group elements are simple rotations, being uniformly weighted to define the symmetry operation $\mathcal{E}$. This symmetry group gives rise to the resource theory of coherence as a special case~\cite{Marvian_2016, Piani2016, Napoli2016}. Any incoherent state is left invariant under the action of $\mathcal{E}$, i.e., an incoherent state is symmetric with respect to the symmetry group, thereby naturally selecting a set of free states. On the contrary, any coherent state will encode information about $r$, thus being useful at least in principle for the task of order-finding. Analyzing the protocol in the framework of coherence theory is a natural consequence. Concretely, in this work the performance of the protocol will be quantitatively linked to the ability to create and detect coherence.

Furthermore, it has to be noted that not every single rotation $\mathcal{E}_j$ encodes the order $r$ the way we wish. In fact, any rotation $\mathcal{E}_j$ where $\text{gcd}(j,r)\neq 1$ is equivalent to a rotation around an angle depending on a factor of $r$ rather than $r$ itself. Fortunately, this is sufficiently rare to still allow for an efficient post-processing strategy that estimates $r$ from the measurement statistics efficiently. After the modular exponentiation, a measurement of system $A$ in the Fourier basis produces an outcome $k$ that is forwarded to the continued fraction algorithm (CFA), which will then compute a continued fraction decomposition of $k/q$.

The continued fraction algorithm computes the decomposition of a number $x$ in the following iterative form: the sum of its closest integer part and the reciprocal of another number, which is then written as the sum of its closets integer part and another reciprocal, and so on, see for example Ref.~\cite{Nielsen2016}. This decomposition is typically denoted as
\begin{align}
	x=[a_0,a_1,a_2,\ldots]=a_0+\frac{1}{a_1+\frac{1}{a_2+\frac{1}{\ldots}}},
\end{align}
where the list is finite for rational $x$, i.e., $x=[a_0,a_1,\ldots, a_n ]$, and infinite otherwise. The so-called convergents, or specifically the $m$-th convergent of $x$, is defined by $[a_0,a_1,\ldots,a_m]$. The post-processing of measurement results will be done by computing the convergents of $k/q$. Some measurement results give sufficiently good approximations to some $j/r$ that allow recovering the latter fraction from $k/q$ by using the CFA to compute the convergents, where one matches $j/r$.

To clarify which measurement outcomes do so, we continue with the following result from number theory involved in the study of Diophantine Approximation, i.e., approximations of irrational numbers by rational ones. The following statement can be found in various textbooks on number theory, see for example Ref.~\cite{Miller2006}. The first part is also treated in the textbook~\cite{Nielsen2016}, and for completeness, we give a short proof of the statement based on Ref.~\cite{Miller2006}.
\begin{thm}\label{thm::Approximation}
	Let $x$ be a positive  number and $p/q$ a positive rational number. If 
	\begin{equation}\label{cond1}
		\left|x-\frac{p}{q} \right|<\frac{1}{2q^2}
	\end{equation}
	then $p/q$ is a convergent of $x$. Conversely if $p/q$ is a convergent of $x$, then 
	\begin{equation}\label{cond2}
		\left|x-\frac{p}{q} \right|\leq \frac{1}{q^2}.
	\end{equation}
\end{thm}
\begin{proof}
	Let $\tfrac{p_n}{q_n}$ denote the convergents to the continued fraction decomposition of $x$. The sequence $(q_n)_n$ is increasing~\cite{Miller2006}, thus there exists some integer $n$ such that $q_{n}\leq q< q_{n+1}$. Now assume that $\tfrac{p}{q}$ satisfies the inequality~(\ref{cond1}) but is not a convergent to the continued fraction algorithm, i.e., $ \tfrac{p}{q}\neq \tfrac{p_n}{q_n}\, \forall n$. The convergents $\tfrac{p_n}{q_n}$ are precisely the best approximations to $x$ in the second sense, thus, $|qx-p|<|q_nx-p_n|$ implies $q>q_{n+1}$~\cite{Miller2006}.
	Therefore, if $\tfrac{p}{q}$ is not a convergent with $q< q_{n+1}$ (if $q=q_{n+1}$ there is nothing to show) we find $|q_nx-p_n|\leq|qx-p|<\tfrac{1}{2q}$, since $\tfrac{p}{q}$ satisfies Eq.~\eqref{cond1} by assumption. This yields
	\begin{equation}
		\begin{split}
			\frac{1}{qq_n}&\leq \frac{|pq_n-qp_n|}{qq_n}= \left|\frac{p}{q} -\frac{p_n}{q_n} \right| \leq \left|x -\frac{p_n}{q_n} \right|+ \left|x -\frac{p}{q} \right| \\
			&< \frac{1}{2qq_n}+\frac{1}{2q^2},
		\end{split}
	\end{equation}
	and thus $q_n>q$, which is a contradiction to $q_n\leq q<q_{n+1}$. Therefore, we find that $q=q_n$ and consequentially $p=p_n$, which concludes the first part of the statement.\newline
	For the second part, we can make use of the so-called complete quotients $a_i^\prime$, see for example Ref.~\cite{Miller2006}, defined as $a_i^\prime=[a_i,a_{i+1},...]$ which allows us to express arbitrary $x$ as
	\begin{equation}
		x=\frac{a_{i+1}^\prime p_i+p_{i-1}}{a_{i+1}^\prime q_i+q_{i-1}},
	\end{equation}
	in terms of an arbitrary convergent $\tfrac{p_i}{q_i}$. Then it follows
	\begin{equation}
		\begin{split}
			\left|x-\frac{p_i}{q_i} \right| &= \left|\frac{a_{i+1}^\prime p_i+p_{i-1}}{a_{i+1}^\prime q_i+q_{i-1}}-\frac{p_i}{q_i} \right|= \left|\frac{\left(a_{i+1}^\prime p_i+p_{i-1}\right) q_i-p_i\left(a_{i+1}^\prime q_i+q_{i-1}\right)}{q_i\left(a_{i+1}^\prime q_i+q_{i-1}\right)}\right| \\
			&= \left| \frac{p_{i-1}q_i-p_iq_{i-1}}{q_i\left(a_{i+1}^\prime q_i+q_{i-1}\right)} \right| = \left| \frac{(-1)^i}{q_i\left(a_{i+1}^\prime q_i+q_{i-1}\right)} \right| \\
			&\leq \frac{1}{q_i q_{i+1}},
		\end{split}
	\end{equation}
	where we used in the last line that $a_{i+1}^\prime q_i+q_{i-1}\geq q_{i+1}$.
	Lastly, since $q_{i+1}\geq q_{i}$ every convergent and thus also the particular convergent $p/q$ satisfies the inequality $\left|x -\tfrac{p}{q}\right|\leq\tfrac{1}{q^2}$. Recall that in the case of a rational $x$, i.e., a simple finite continued fraction expansion $x=[a_0,a_1,...,a_n]$, we define the denominator of the $n+1$ convergent simply as $q_n$, such that the proof also holds for rational $x$.
\end{proof}
This Theorem provides a sufficient and necessary condition on the absolute difference of the number $x$ and a rational approximation
$p/q$ such that said approximation is a convergent of $x$ in the continued fraction decomposition. Coming back to the
question of which measurement outcomes are useful, we employ the following Corollary.

\begin{cor}\label{cor:goodEstimate}
	Let $k$  be an integer with $0\leq k < q$ where $N^2<q=2^L<2N^2$ that satisfies the inequality 
	\begin{equation}\label{eq:DifferenceBeta}
		\left| \frac{j}{r}-\frac{k}{q} \right|\leq\frac{\beta}{2q}
	\end{equation}
	for some coprime $(j,r)$ with $0<j<r$ and $\beta=\tfrac{q-1}{r^2}$. 
	Then the continued fraction expansion of $k/q$ will yield
	$j/r$ and thereby $r$, as a convergent.
\end{cor}
\begin{proof}
	According to the first part of Thm.~\ref{thm::Approximation}, any integer $k$ that satisfies $\left|\tfrac{j}{r}-\tfrac{k}{q}\right|<\tfrac{1}{2r^2}$ will yield $j/r$ as a convergent.  Obviously $\tfrac{\beta}{2q}<\tfrac{1}{2r^2}$. In particular, since $\beta>1$ all integers $k$ that obey the inequality $\left|\tfrac{j}{r}-\tfrac{k}{q}\right|< \tfrac{1}{2q}$ will yield $j/r$ as a convergent.
\end{proof}
This justifies the choice of the dimension of quantum system A with $\text{dim}(A)=q$ at the beginning of the discussion. Let us emphasize that extending the margin of error like in Cor.~\ref{cor:goodEstimate} for a $\beta>1$ has allowed to sharpen Shor's original bound (which basically utilizes a weaker bound with $\beta=1$) on the coherent protocol, see for example Refs.~\cite{Gerjuoy2005,Bourdon}. Looking at the following result, the reason why the post-processing via the CFA works well for a measurement result as in Cor.~\ref{cor:goodEstimate} can be better understood.

\begin{lem}\label{Unique} 
	Consider fixed integers $N$ and $q>N^2$.
	
	i) Assume you have a fixed integer $0\leq k<q$. Then there exists at most one pair of integers $(j,r)$ with $1\leq r <N$, $0\leq j<r$, and $\text{gcd}(j,r)=1$ such that $\left| \tfrac{j}{r}-\tfrac{k}{q} \right|<\tfrac{1}{2q}$.
	
	ii) Assume that you have a pair of integers $(j,r)$ with $1\leq r <N$ and $0\leq j<r$. Then there exists an integer $0\leq k<q$ such that $\left| \tfrac{j}{r}-\tfrac{k}{q} \right|<\tfrac{1}{2q}$ is satisfied.
\end{lem}
\begin{proof}
	We begin with i). Assume that there exist two distinct fractions $\tfrac{j\prime}{r^\prime}\neq \tfrac{j}{r}$ that satisfy $\left| \tfrac{j}{r}-\tfrac{k}{q} \right|<\tfrac{1}{2q}$ and $\left| \tfrac{j'}{r'}-\tfrac{k}{q} \right|<\tfrac{1}{2q}$.
	It follows that 
	\begin{align}
		\left| \tfrac{j^\prime}{r^\prime}-\tfrac{j}{r}\right|=\left| \tfrac{j^\prime}{r^\prime}-\tfrac{k}{q}+\tfrac{k}{q}-\tfrac{j}{r}\right|<\tfrac{1}{q}  <\tfrac{1}{N^2}.
	\end{align}
	On the other hand $\left| \tfrac{j^\prime}{r^\prime}-\tfrac{j}{r}\right|=\left| \tfrac{j^\prime r-j r^\prime}{rr^\prime}\right|>\tfrac{1}{N^2}$, since $r,r^\prime<N$ and there exists an integer $i$ such that $|j^\prime r-j r^\prime|=|i|\geq 1$. By contradiction the two fractions are identical.
	
	For ii), we note that the distance between neighboring fractions $\frac{k}{q}$ is given by $\frac{1}{q}$. Therefore, there always exists a $k'$ such that $\left|\tfrac{j}{r}-\tfrac{k'}{q}\right|\leq \tfrac{1}{2q}$. However, equality can only hold if $r$ is a power of 2, in which case there exists a $k$ such that $\tfrac{k}{q}$ samples $\tfrac{j}{r}$ exactly.
\end{proof}

Combining the results of Cor.~\ref{cor:goodEstimate} and Lem.~\ref{Unique} tells us, that given a single rotation $\mathcal{E}_j$ as defined in Lem.~\ref{Encoding} and with $j$ coprime to $r$, there always exists a measurement outcome $k$ that will yield r via the continued fraction algorithm. 
With this at hand, we define the following two sets for a fixed $j$ coprime to $r$ 
\begin{equation}\label{eq:Sets}
	\begin{split}
		&\mathcal{K}_1^{j}=\left\{ k: 0\leq k< q\, \wedge \,  \left|\frac{j}{r}-\frac{k}{q}\right|< \frac{1}{2q} \right\},\\
		&\mathcal{K}_2^{j}=\left\{ k: 0\leq k< q \, \wedge \, \left|\frac{j}{r}-\frac{k}{q}\right|\leq\frac{1}{r^2} \right\}.
	\end{split}
\end{equation}
Additionally we define the sets $\mathcal{K}_1,\mathcal{K}_2$ as $\mathcal{K}_i=\cup_j \mathcal{K}_i^j$, where the union is formed over all $j$ smaller than and coprime to $r$. The set $\mathcal{K}_1$ contains all measurement outcomes that will yield the correct order $r$ by putting the outcome in the continued fraction algorithm. The second set $\mathcal{K}_2$ consists of all outcomes that obey the necessary condition to be a convergent of the CFA according to Thm.~\ref{thm::Approximation}, i.e., it contains all outcomes that will yield the correct $r$ via the CFA but potentially also outcomes that do not. Let us conclude this preliminary discussion by noting what happens for an unknown and randomly chosen $j$ (or equivalently a uniformly weighted $\mathcal{E}_j$, as we got here) during the post-processing. Suppose for the sampled $\mathcal{E}_j$, $j$ and $r$ share a common factor. The post-processing will then maximally yield a factor of $r$ and thus fail. This case is however rare: the probability that a randomly chosen $j$ is coprime to $r$ is given by $\varphi(r)/r$, where $\varphi(r)$ denotes Euler's totient function.
This ratio between Euler's totient function and its argument is bounded by $\tfrac{\varphi(r)}{r}>\tfrac{\delta}{\log\log r}>\tfrac{\delta}{\log\log N}$, for some positive constant $\delta$, according to a well-known result by Hardy~\cite[Theorem 328]{Hardy1984}. In fact, the latter inequality is asymptotically tight for infinitely many values of $r$. 

\subsection{Sequential order-finding}
Furthermore, we have to discuss a sequential version of Shor's original order-finding protocol that allows reducing the number of qubits drastically for large factorization problems. The protocol is based on a semi-classical implementation of the combination of an inverse quantum Fourier transform and a measurement in the computational basis following directly afterward (see Refs.~\cite{Griffiths1995,Parker2000}): Assume the inverse Fourier transform is implemented via its \textit{standard} decomposition into Hadamard gates and controlled rotations as depicted in Fig.~\ref{fig:inverseFourier}, see also Ref.~\cite{Nielsen2016}.
Fig.~\ref{fig:ShorDecomposition} therefore shows an implementation of Shor's algorithm in which the measurement outcome has to be reordered in reverse order. As explained in detail in Ref.~\cite{Griffiths1995}, it is then possible to do the measurement on the first qubit directly after the first Hadamard gate belonging to the inverse Fourier transform was implemented (the second Hadamard gate in the figure)  and use its outcome to classically control all the following rotations that depend on this qubit. A similar argument holds for the other qubits as well: after the respective Hadamard gate in their line, one can directly measure them and control all following rotations classically depending on the outcome. Since all the controlled rotations in one line lead to an effective rotation, in this way, one can replace them with a single effective classically controlled rotation $R_l'$ that depends on the previous measurement outcomes. This is shown in Fig.~\ref{fig:ShorDecompositionComp}.

Thereby, the gates and measurements on the individual qubits are performed sequentially, which allows to split Shor's protocol into blocks, see Fig.~\ref{BlockFormHadamards}, that each utilize only a single control qubit on which the Hadamard gates and the classically controlled rotations $R_l^\prime$ are performed. The single control qubit can be \textit{recycled} after each block, such that the total amount of qubits required decreases to $\log N+1$. Due to this decomposition, Shor's original protocol and the sequential version lead to identical measurement statistics if the auxiliary systems are initialized in the same state.  \newline

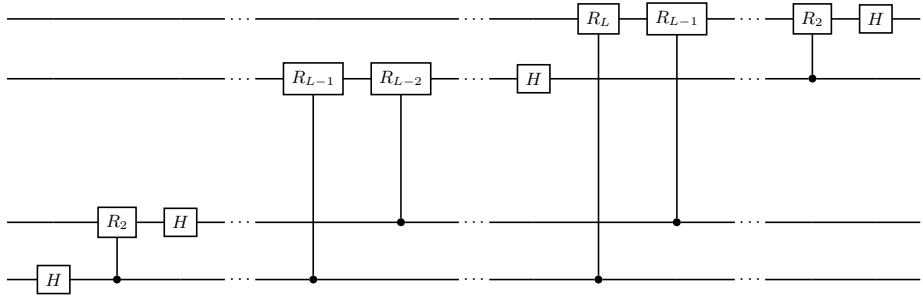
\begin{figure}[ht]
	\centering
	\scalebox{0.75}{
	\begin{quantikz}
	  & \qw& \qw& \qw &\ \ldots \ \qw  & \qw &  \qw &\ \ldots \ \qw & \qw & \gate{R_L} & \gate{R_{L-1}} &  \ \ldots \ \qw & \gate{R_2} &\gate{H} &  \qw\\
	& \qw& \qw & \qw &\ \ldots \ \qw  & \gate{R_{L-1}} & \gate{R_{L-2}} & \ \ldots \ \qw &\gate{H} &  \qw & \qw  & \ \ldots \ \qw & \ctrl{-1} &\qw & \qw\\
	 \ &&&&&&&&&&&&\\	
	 \ &&&&&&&&&&&&\\
	 \ &&&&&&&&&&&&\\ 
	 & \qw& \gate{R_2}&\gate{H} & \ \ldots \ \qw  & \qw & \ctrl{-4} & \ \ldots \ \qw &\qw & \qw& \ctrl{-5} & \ \ldots \ \qw & \qw & \qw & \qw\\
	 &\gate{H} & \ctrl{-1}& \qw &    \ \ldots \ \qw  & \ctrl{-5} & \qw & \ \ldots \ \qw &\qw & \ctrl{-6} & \qw& \ \ldots \ \qw & \qw & \qw & \qw\\
\end{quantikz}
}
	\caption{Standard decomposition of the inverse Fourier transform into Hadamard gates and controlled rotations $R_l$. The controlled rotation $R_l$ adds a phase of $-2\pi i/2^l$ to $\ket{1}$ and leaves $\ket{0}$ unchanged. An additional initial reordering of the qubits in reverse order is not shown. }
	\label{fig:inverseFourier}
\end{figure}

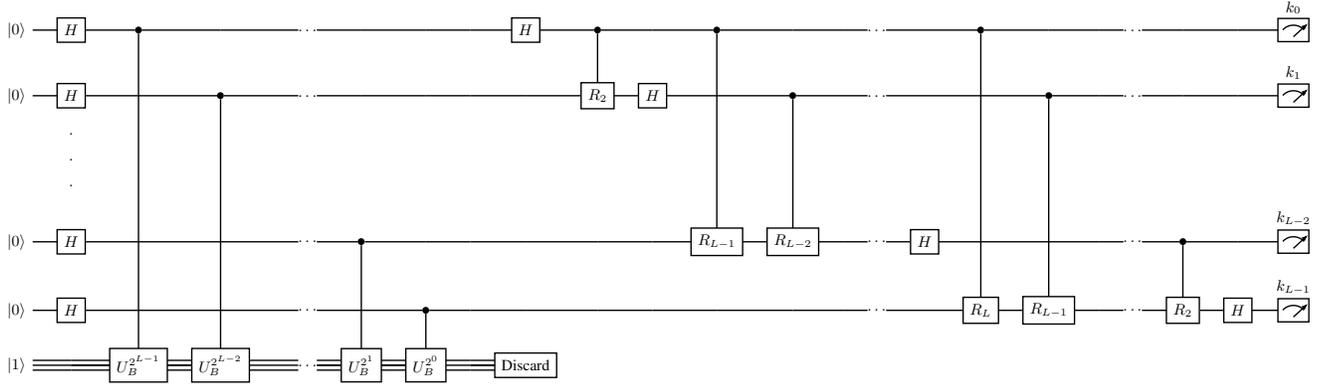
\begin{figure}[ht]
	\centering
	\scalebox{0.65}{
	\begin{quantikz}
	\lstick{$\ket{0}$} 	 & \gate{H} 			  & \ctrl{7} 												&\qw		 											&\qw						&\ldots	\qw						&\qw												&\qw												&\qw					& \gate{H}	&\ctrl{1}	&\qw		&\ctrl{5}		&\qw			&\qw		&\ldots\qw	&\qw		&\ctrl{6}	&\qw			&\qw	&\ldots\qw	&\qw		&\qw		&\meter{$k_0$}	\\
	\lstick{$\ket{0}$} 	 & \gate{H} 			  & \qw						 								& \ctrl{6}												&\qw						&\ldots	\qw						&\qw	 											&\qw												&\qw					&\qw		&\gate{R_2}	&\gate{H}	&\qw			&\ctrl{4}		&\qw		&\ldots\qw	&\qw		&\qw		&\ctrl{5}		&\qw	&\ldots\qw	&\qw		&\qw		&\meter{$k_1$}	\\
	&	.					  &															& 														&							&								&  													&													&						&			&			&			&				&				&			&			&			&			&				&		&			&			&			&			\\
	&	.					  &															&														&							&								&													&													&						&			&			&			&				&				&			&			&			&			&				&		&			&			&			&			\\
	&	.					  &															&														&							&								&													&													&						&			&			&			&				&				&			&			&			&			&				&		&			&			&			&			\\
	\lstick{$\ket{0}$} 	 & \gate{H} 			  & \qw														&\qw 													&\qw						&\ldots	\qw						&\ctrl{2}											&\qw												&\qw					&\qw		&\qw		&\qw		&\gate{R_{L-1}}	&\gate{R_{L-2}}	&\qw		&\ldots\qw	&\gate{H}	&\qw		&\qw			&\qw	&\ldots\qw	&\ctrl{1}	&\qw		&\meter{$k_{L-2}$}	\\
	\lstick{$\ket{0}$} 	 & \gate{H} 			  & \qw														&\qw 													&\qw						&\ldots	\qw						&\qw												&\ctrl{1}											&\qw					&\qw		&\qw		&\qw		&\qw			&\qw			&\qw		&\ldots\qw	&\qw		&\gate{R_L}	&\gate{R_{L-1}}	&\qw	&\ldots\qw	&\gate{R_2}	&\gate{H}	&\meter{$k_{L-1}$}	\\
	\lstick{$\ket{1}$} & \qwbundle[alternate]{} & \gate{U_B^{2^{L-1}}}\qwbundle[alternate]{}	&\gate{U_B^{2^{L-2}}}\qwbundle[alternate]{}	&\qwbundle[alternate]{}  	&\ldots	\qwbundle[alternate]{}	&\gate{U_B^{2^{1}}}\qwbundle[alternate]{}	&\gate{U_B^{2^{0}}}\qwbundle[alternate]{}	&\qwbundle[alternate]{}	& \gate{\text{Discard}}\qwbundle[alternate]{}			&			&			&				&				&			&			&			&			&				&		&			&			&			&			\\
\end{quantikz}
}
	\caption{Decomposition of Shor's algorithm, with the inverse Fourier transform decomposed into  Hadamard gates and controlled rotations $R_l$. A controlled rotation $R_l$ adds a phase of $-2\pi i/2^l$ to $\ket{1}$ and leaves $\ket{0}$ unchanged. This leads to a total outcome $k=\sum_{i=0}^{L-1} 2^i k_i$. }
	\label{fig:ShorDecomposition}
\end{figure}

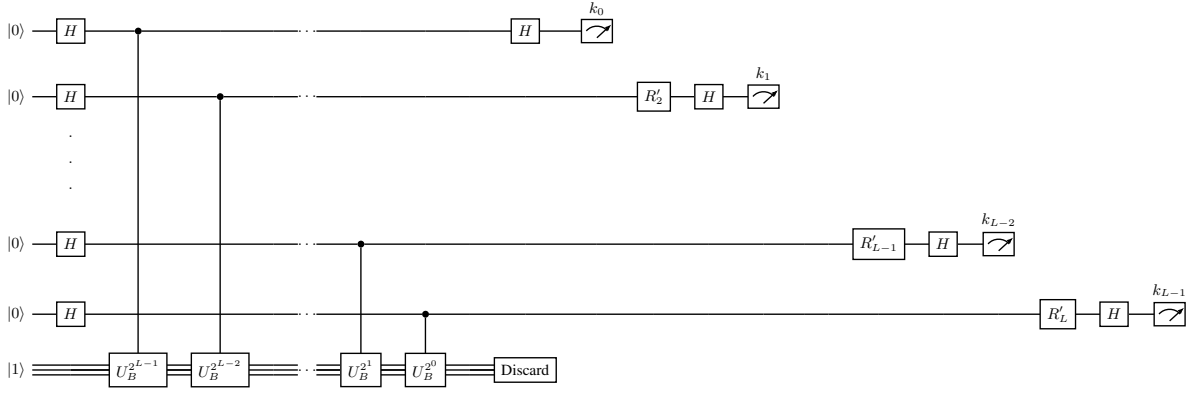
\begin{figure}[ht]
	\scalebox{0.65}{
	\begin{quantikz}
	\lstick{$\ket{0}$} 	 & \gate{H} 			  & \ctrl{7} 												&\qw		 											&\qw						&\ldots	\qw						&\qw												&\qw												&\qw					& \gate{H}	&\meter{$k_0$}	&		&		&		&	&	&		&	&			&	&	&			\\
	\lstick{$\ket{0}$} 	 & \gate{H} 			  & \qw						 								& \ctrl{6}												&\qw						&\ldots	\qw						&\qw	 											&\qw												&\qw					&\qw	&\qw	&\gate{R_2'}	&\gate{H}				&\meter{$k_1$}		&		&	&		&		&		&	&	&			\\
	&	.					  &															& 														&							&								&  													&													&						&			&			&			&				&				&			&			&			&			&				&		&			&					\\
	&	.					  &															&														&							&								&													&													&						&			&			&			&				&				&			&			&			&			&				&		&			&					\\
	&	.					  &															&														&							&								&													&													&						&			&			&			&				&				&			&			&			&			&				&		&			&						\\
	\lstick{$\ket{0}$} 	 & \gate{H} 			  & \qw														&\qw 													&\qw						&\ldots	\qw						&\ctrl{2}											&\qw												&\qw					&\qw		&\qw		&\qw		&\qw	&\qw	&\qw		&\qw	&\gate{R_{L-1}'}&\gate{H}			&\meter{$k_{L-2}$}			& 	&	& 		\\
	\lstick{$\ket{0}$} 	 & \gate{H} 			  & \qw														&\qw 													&\qw						&\ldots	\qw						&\qw												&\ctrl{1}											&\qw					&\qw		&\qw		&\qw		&\qw			&\qw			&\qw		&\qw	&\qw		&\qw	&\qw	&\gate{R_{L}'}	&\gate{H}	&\meter{$k_{L-1}$}	\\
	\lstick{$\ket{1}$} & \qwbundle[alternate]{} & \gate{U_B^{2^{L-1}}}\qwbundle[alternate]{}	&\gate{U_B^{2^{L-2}}}\qwbundle[alternate]{}	&\qwbundle[alternate]{}  	&\ldots	\qwbundle[alternate]{}	&\gate{U_B^{2^{1}}}\qwbundle[alternate]{}	&\gate{U_B^{2^{0}}}\qwbundle[alternate]{}	&\qwbundle[alternate]{}	&	\gate{\text{Discard}}\qwbundle[alternate]{}		&			&			&				&				&			&			&			&			&				&		&			&					\\
\end{quantikz}
}
	\caption{Rewriting Shor's algorithm using classically controlled effective rotations $R'_l=\sum_{n=0}^1 e^{-2\pi i n\phi'_l}\ketbra{n}{n}$ that depend on the outcomes of previous measurements via $\phi'_l=\sum_{a=2}^l k_{l-a}/2^a$.}
	\label{fig:ShorDecompositionComp}
\end{figure}

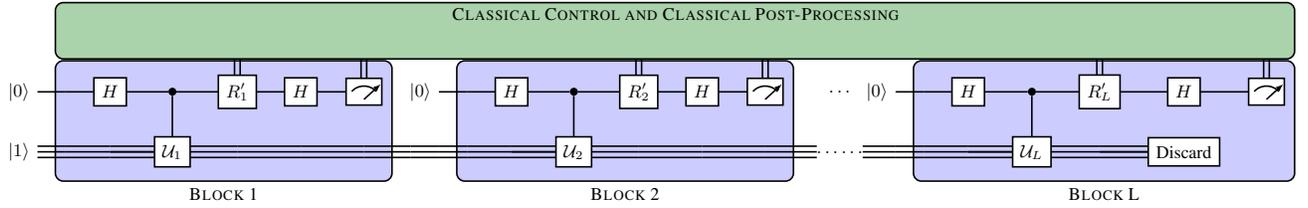
\begin{figure}[ht]
	\scalebox{0.75}{
\begin{quantikz}
&\gategroup[1,steps=22,style={rounded corners,fill=dartmouthgreen!35, inner xsep=2pt}, label style={label position=above,anchor=north,yshift=-0.2cm}]{{\sc Classical Control and Classical Post-Processing}}&\gate{\text{A}} \cw &&&&&&&&&&&&&&&&&&&& \\
\lstick{$\ket{0}$}&\qw\gategroup[2,steps=6,style={rounded corners,fill=blue!20, inner xsep=2pt},background, label style={label position=below,anchor=north,yshift=-0.2cm}]{{\sc Block 1}}&\gate{H}&\ctrl{1}& \gate{R_1^\prime} \vcw{-1} & \gate{H}&\meter{} \vcw{-1} && \lstick{$\ket{0}$} &\qw \gategroup[2,steps=6,style={rounded corners,fill=blue!20, inner xsep=2pt},background, label style={label position=below,anchor=north,yshift=-0.2cm}]{{\sc Block 2}}& \gate{H} &\ctrl{1} &\gate{R_2^\prime} \vcw{-1} &\gate{H}& \meter{}\vcw{-1} & \hspace{0.3cm}\ldots\hspace{0.3cm} & \lstick{$\ket{0}$}&\qw \gategroup[2,steps=6,style={rounded corners,fill=blue!20, inner xsep=2pt},background, label style={label position=below,anchor=north,yshift=-0.2cm}]{{\sc Block L}}&\gate{H}& \ctrl{1}&\gate{R_L^\prime} \vcw{-1} &\gate{H} &\meter{} \vcw{-1}  \\
\lstick{$\ket{1}$}&\qwbundle[alternate]{}&\qwbundle[alternate]{}&\gate{\mathcal{U}_1} \qwbundle[alternate]{}&\qwbundle[alternate]{}&\qwbundle[alternate]{}&\qwbundle[alternate]{}&\qwbundle[alternate]{} &\qwbundle[alternate]{}&\qwbundle[alternate]{}&\qwbundle[alternate]{}&\gate{\mathcal{U}_2}\qwbundle[alternate]{} &\qwbundle[alternate]{} &\qwbundle[alternate]{} &\qwbundle[alternate]{} &\qwbundle[alternate]{} \ldots\ldots &\qwbundle[alternate]{}&\qwbundle[alternate]{} &\qwbundle[alternate]{} &\gate{\mathcal{U}_L} \qwbundle[alternate]{} &\qwbundle[alternate]{} &\gate{\text{Discard}}\qwbundle[alternate]{}&
\end{quantikz}
}
	\caption{Sequential order-finding protocol using the semi-classical version of the Fourier transform. The modular exponentiation factors into single qubit controlled-operations given by $\mathcal{U}_l=U_B^{2^{L-l}}$ and the classically controlled rotations $R_l^\prime=\sum_n e^{-2\pi i n \phi_l^\prime}\ketbra{n}{n}$, where the phases $\phi_l^\prime$ depend on the previous measurement outcomes $k_l$ via $\phi_l^\prime=\sum_{a=2}^{l} k_{l-a}/2^{a}$, see Refs.~\cite{Griffiths1995,Parker2000}.}
	\label{BlockFormHadamards}
\end{figure}

\section{Choosing the free operations}

As discussed in the main text, we fix the overall protocol that we investigate and vary only parts of it. Here, we will explain our choices a bit more in detail.
First, we assume that the post-processing of measurement results after a single round is achieved by the continued fraction algorithm. 
If this fails, then we restart the algorithm and perform the post-processing without accounting for the previous outcomes, thereby ignoring possible correlations between results of failed trials. In general, this is not the best possible post-processing strategy. An example of a more involved strategy can be found in Ref.~\cite{Nielsen2016}. Nevertheless, for simplicity, we assume this fixed post-processing involves only the outcome of individual trials since we are not optimizing over post-processing strategies anyway.
The ability to create, then utilize, and finally detect coherence is a key feature in the protocol. Imposing constraints on these abilities can be done naturally within the framework of dynamical resource theories of coherence.

Notice that expressing the protocol in the form of Fig.~\ref{BlockFormHadamards} makes it clear that except for the Hadamard gates, the protocol utilizes only incoherent input states, channels $\mathcal{U}_l$ and $R_l^\prime$ that can neither detect nor create coherence, and measurements in the incoherent basis. 
Replacing Hadamard gates by quantum channels $S_1^{(l)}[\Theta_l]$ and $S_2^{(l)}[\Lambda_l]$ respectively results in the protocol depicted in Fig.~\ref{BlockForm}.
If no particular block is considered, we omit the label $l$ and refer to the channels for creation and detection simply as $\Theta$ and $\Lambda$.

\begin{figure}[ht]
	\centering
	\scalebox{0.64}{
	\begin{quantikz}
&\gategroup[1,steps=22,style={rounded corners,fill=dartmouthgreen!35, inner xsep=2pt}, label style={label position=above,anchor=north,yshift=-0.2cm}]{{\sc Classical Control and Classical Post-Processing}}&\gate{\text{A}} \cw &&&&&&&&&&&&&&&&&&&& \\
\lstick{$\sigma_1$}&\qw\gategroup[2,steps=6,style={rounded corners,fill=blue!20, inner xsep=2pt},background, label style={label position=below,anchor=north,yshift=-0.2cm}]{{\sc Block 1}}&\gate{S_1^{(1)}[\Theta_1]}&\ctrl{1}& \gate{R_1^\prime} \vcw{-1} & \gate{S_2^{(1)}[\Lambda_1]}&\meter{} \vcw{-1} && \lstick{$\sigma_2$} &\qw \gategroup[2,steps=6,style={rounded corners,fill=blue!20, inner xsep=2pt},background, label style={label position=below,anchor=north,yshift=-0.2cm}]{{\sc Block 2}}& \gate{S_1^{(2)}[\Theta_2]} &\ctrl{1} &\gate{R_2^\prime} \vcw{-1} &\gate{S_2^{(2)}[\Lambda_2]}& \meter{}\vcw{-1} & \hspace{0.3cm}\ldots\hspace{0.3cm} & \lstick{$\sigma_L$}&\qw \gategroup[2,steps=6,style={rounded corners,fill=blue!20, inner xsep=2pt},background, label style={label position=below,anchor=north,yshift=-0.2cm}]{{\sc Block L}}&\gate{S_1^{(L)}[\Theta_L]}& \ctrl{1}&\gate{R_L^\prime} \vcw{-1} &\gate{S_2^{(L)}[\Lambda_L]} &\meter{} \vcw{-1}  \\
\lstick{$\ket{1}$}&\qwbundle[alternate]{}&\qwbundle[alternate]{}&\gate{\mathcal{U}_1} \qwbundle[alternate]{}&\qwbundle[alternate]{}&\qwbundle[alternate]{}&\qwbundle[alternate]{}&\qwbundle[alternate]{} &\qwbundle[alternate]{}&\qwbundle[alternate]{}&\qwbundle[alternate]{}&\gate{\mathcal{U}_2}\qwbundle[alternate]{} &\qwbundle[alternate]{} &\qwbundle[alternate]{} &\qwbundle[alternate]{} &\qwbundle[alternate]{} \ldots\ldots &\qwbundle[alternate]{}&\qwbundle[alternate]{} &\qwbundle[alternate]{} &\gate{\mathcal{U}_L} \qwbundle[alternate]{} &\qwbundle[alternate]{} &\gate{\text{Discard}}\qwbundle[alternate]{}&
\end{quantikz}
}
	\caption{Circuit representation of the order-finding protocol using only channels $\Theta_l$ and $\Lambda_l$ to create and detect coherence. The outcomes of an (incoherent) projective measurement in the computational basis are forwarded to the classical control and post-processing unit, which re-initializes the single control qubit, classically controls the rotations $R_l^\prime$ to implement the inverse Fourier transform, and lastly computes the continued fraction decomposition to yield an estimate of the order $r$.}
	\label{BlockForm}
\end{figure}
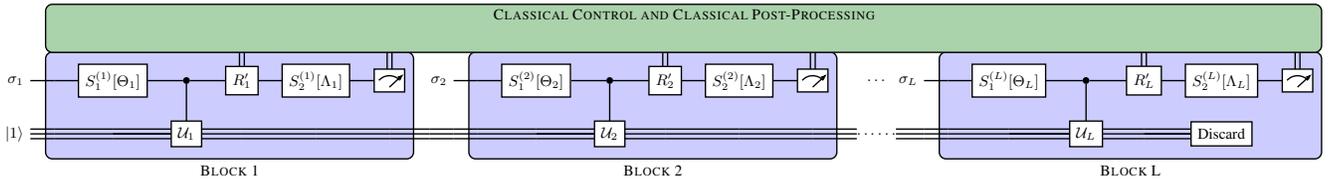

Let us now explain why the symmetry of the fully \textit{coherent} protocol that uses the same channel to create and detect coherence (i.e., the Hadamard gate) has to be broken in the more general case: The ability to create and detect coherence are two fundamentally different properties a quantum channel can possess, which in turn gives rise to two different resources that are generally not interconvertible (e.g., a channel $\Gamma(\sigma)=\rho\trace(\sigma)$ can prepare coherence if $\rho$ is chosen suitably, but not detect, whilst a destructive measurement in the Fourier basis can detect but not prepare coherence). The Hadamard gate can however create a maximally coherent state (by applying it to $\ket{0}$), but also maximizes the NSID-measure~\cite{Theurer2019}. Therefore, it plays a dual role, i.e., it both creates and detects coherence.

As mentioned in the main text, the choice of free channels follows naturally. If $\Theta$ is incapable of creating coherence, no information about the order can be encoded. If $\Lambda$ cannot detect coherence, none of this information can influence the measurement statistics. Thus, the lack of either ingredient renders the protocol practically ``useless" by reducing it to a random number generator independent of the order that it is supposed to estimate, and is moreover classically simulable. Therefore, the choices of free channels are maximally incoherent channels $\CI$~\cite{Aberg2006} and detection-incoherent channels $\DI$~\cite{Theurer2019}, also known as non-activating~\cite{Liu2017}.  Let us mention that this random number generator gives rise to different probability distributions depending on the structure of the free channels $\Theta_{\free}$ and $\Lambda_{\free}$. Details will follow in the next section.\newline
It is tempting to choose the set of creation-detection incoherent channels $\CDI$ as the set of free channels, i.e., the channels that can neither create nor detect coherence, also known as dephasing-covariant channels~\cite{Chitamber2016Comp,Marvian2016,Chitamber2016}, classical~\cite{Meznaric2013}, or commuting~\cite{Liu2017}. This would keep the symmetry of the protocol and seems to be an intuitive choice as it leads to a ``fully classical" protocol. 
However, it does not lead to a consistent connection between operational advantages and deployed resources:
Imagine we would use an channel $\Lambda \in \DI$ with $\Lambda \not\in \CI$ for detection. Although not granting any operational advantage, this channel has to be considered resourceful. In contrast, our choice of different sets of free channels naturally leads to an operationally meaningful use of resources. \newline\newline

Furthermore, the channel $\Lambda$ utilized in the detection scheme is assumed to be a unital map. This assumption is physically motivated: The measurement statics of the incoherent measurement are uniquely determined by the pre-measurement populations. To be maximally sensitive to information about $r$, we want that the deviation of the measurement statistics from a flat distribution purely depends on the coherences that $\Lambda$ mapped to populations, and not on a reshuffling of populations that does not include information about $r$.
Without knowing $r$, we can choose both free super-channels $S_1$ and $S_2$ such that this is the case iff $\Lambda$ is unital. Since the state before $\mathcal{U}_l$ is still independent of $r$, we can always choose $S_1$ such that its populations are equal to a maximally mixed state, without affecting the coherences (because the phases of the coherences are still independent of $r$ and therefore known). After $\mathcal{U}_l$, the phase of the coherences depends however on $r$, and we can thus not alter the populations without varying the coherences (or knowing $r$). Thus, if $\Lambda$ were not unital but could detect coherence, the following might happen: The measurement statistics depend stronger on the population reshuffling than on the detected coherences. In this case, we would perform worse than with a free channel that leads to equally distributed random numbers and therefore on average produces better guesses of $r$ than random numbers that are weighted in a way that does not depend on $r$. To avoid this, we must choose $\Lambda_l$ to be unital and similarly choose the super-channels $S_2^{(l)}$ to be unitality-preserving.

\section{Success probability}
The success probability of the order-finding protocol, consisting of the quantum part combined with
the continued fraction algorithm, can now be expressed. To ease up the notation, we make use of the equivalence between Shor's original version and the sequential version. That way, there is no need to laboriously track the back-action of the measurements in each block on the auxiliary system, which allows us to express the success probability compactly. Recall that the detection part, i.e., the \textit{standard} implementation of the inverse Fourier transform (see Fig.~\ref{fig:inverseFourier}), was altered only by replacing the Hadamard gates with channels $S_2^{(l)}[\Lambda_l]$. Let us denote the resulting channel by $F_{S_2^{(l)}[\Lambda_l]}$. Furthermore, we use $\tilde{\sigma}=\bigotimes_{l} \sigma_l$ and the POVM elements $M_k=\bigotimes_l M_{k_l}$. With this  notation, the incoherent measurement $\mathbb{M}= \lbrace M_k \rbrace_k$ results in the measurement statistics 
\begin{equation}\label{eq:succProb}
	\begin{split}
		p_k(S_1^{(l)}[\Theta_l],S_2^{(l)}[\Lambda_l]; \tilde{\sigma}, \mathbb{M})&=\Tr{M_k \Delta F_{S_2^{(l)}[\Lambda_l]} \mathcal{E}\bigotimes_{l=1}^{L} S_1^{(l)}[\Theta_l]\tilde{\sigma}},
	\end{split}
\end{equation}
where $\mathcal{E}$ denotes the uniformly weighted rotations described in Lem.~\ref{Encoding}. After completing all blocks, the measurement outcome $k$ is forwarded to the CFA, which will return the order $r$ with a probability of $P(k\to r\,|\, \text{CFA})$. Therefore, the probability that the order-finding protocol in Fig.~\ref{BlockForm} succeeds, is given by
\begin{equation}
	P^{\suc}(S_1^{(l)}[\Theta_l],S_2^{(l)}[\Lambda_l]; \tilde{\sigma}, \mathbb{M})=\sum_k P(k\to r\,|\, \text{CFA})\,\,p_k(S_1^{(l)}[\Theta_l],S_2^{(l)}[\Lambda_l]; \tilde{\sigma}, \mathbb{M}).
\end{equation}
Since all incoherent input states $\sigma_l$, incoherent measurements, and free super-channels $S^{(l)}_1$ and $S^{(l)}_2$ are available at no cost, we choose them optimally (but without knowledge of $r$ and in a way that is implementable efficiently), which ensures that the available resources are used adequately. The resulting success probability is then given by

\begin{equation}
	\begin{split}
		P^{\suc}(\Theta_l,\Lambda_l)&= \max_{\substack{\tilde{\sigma} \in \mathcal{I} \\\mathbb{M} \in \IM}} \sup_{\substack{S_1^{(l)} \in \CIS \\ S_2^{(l)} \in \DIS}} P^{\suc}(S_1^{(l)}[\Theta_l],S_2^{(l)}[\Lambda_l]; \tilde{\sigma}, \mathbb{M}) \\
		&= \max_{\substack{\tilde{\sigma} \in \mathcal{I} \\\mathbb{M} \in \IM}} \sup_{\substack{S_1^{(l)} \in \CIS \\ S_2^{(l)} \in \DIS}}  \sum_k P(k\to r\,|\, \text{CFA})\,\,p_k(S_1^{(l)}[\Theta_l],S_2^{(l)}[\Lambda_l]; \tilde{\sigma}, \mathbb{M}).
	\end{split}
\end{equation}
Since every incoherent POVM $\mathbb{M}$ is equivalent to a detection-incoherent channel followed by a projective measurement $\mathbb{P}$ in the incoherent basis~\cite{Theurer2019}, the optimization over the measurement can be absorbed into the optimization of the detection-incoherent super-channel, i.e.,
\begin{equation}\label{SuccessProb}
	P^{\suc}(\Theta_l,\Lambda_l)=\max_{\tilde{\sigma} \in \mathcal{I}} \sup_{\substack{S_1^{(l)} \in \CIS \\ S_2^{(l)} \in \DIS}}\sum_k P(k\to r \,|\, \text{CFA})\,p_k(S_1^{(l)}[\Theta_l],S_2^{(l)}[\Lambda_l];\tilde{\sigma},\mathbb{P}).
\end{equation}
In general, this expression seems hard to evaluate exactly. However, in the following section, we will provide bounds allowing us to compare performance and resource content. \newline

\section{Proof of the results in the main text}
In this section we give the proofs of the results presented in the main text, i.e., we derive bounds on the success probability given in Eq.~\eqref{SuccessProb}. 

\subsection{Preliminaries}\label{ProofPreliminary}
We start by presenting a bound on a product that we will later use to obtain a lower bound on the performance of the order-finding protocol.\vspace{0.4cm}
\begin{lem}\label{4PiBound}
	For positive numbers $\lbrace a_l \rbrace_l$ with $0\leq a_l\leq 1\, \forall \, l$ the following inequalities hold: 
	\begin{equation}
		\frac{4}{\pi^2} \prod_{l=1}^{L} \frac{1}{2} \left[1+a_l \right]\leq \prod_{l=1}^{L} \frac{1}{2}\left[1+a_l \cos\left(\frac{\pi}{2^l}\right) \right]\leq \prod_{l=1}^{L} \frac{1}{2} \left[1+a_l \right].
	\end{equation}
\end{lem}
\begin{proof}
	Since $a_l\geq 0$, the upper bound holds trivially. For the lower bound, notice that the term $0\leq \cos\left(\frac{\pi}{2^l}\right) <1$ rapidly converges to one for increasing $l$. Thereby, it is reasonable that the deviation from the simple upper bound is small. First rewrite the product as 
	\begin{equation}
		\begin{split}
			\prod_{l=1}^{L} \frac{1}{2}\left[ 1+a_l \cos\left(\frac{\pi}{2^l}\right) \right] &= \prod_{l=1}^{L} \frac{1}{2} \left[1+\cos\left(\frac{\pi}{2^l}\right)\right]\prod_{l=1}^{L} \left[ \frac{1+a_l \cos\left(\frac{\pi}{2^l}\right)}{1+\cos\left(\frac{\pi}{2^l}\right)} \right]\\
			&= \prod_{l=1}^{L} \frac{1}{2} \left[1+\cos\left(\frac{\pi}{2^l}\right)\right] 
			\prod_{l=1}^{L} \left[\frac{1+a_l}{2} +\frac{1-a_l}{2} \frac{1-\cos\left(\frac{\pi}{2^l}\right)}{1+\cos\left(\frac{\pi}{2^l}\right)}\right]\\
			&\geq \prod_{l=1}^{L} \frac{1}{2} \left[1+\cos\left(\frac{\pi}{2^l}\right)\right] \prod_{l=1}^{L} \left[\frac{1+a_l}{2}\right] \\
			&=\prod_{l=1}^{L} \cos^2\left(\frac{\pi}{2^{l+1}}\right) \prod_{l=1}^{L} \left[\frac{1+a_l}{2}\right].
		\end{split}
	\end{equation}
	Now utilize a special case of the Viète-Euler product formula, see for example Ref.~\cite{MORENO201390}, $\frac{\sin(x)}{x}=\prod_{l=1}^{\infty} \cos \left( \frac{x}{2^l}\right)$ with $ x=\pi/2$ which results in
	\begin{equation}
		\frac{4}{\pi^2}=\prod_{l=1}^{\infty} \cos^2 \left( \frac{\pi}{2^{l+1}}\right)=\prod_{l=1}^{L} \cos^2 \left( \frac{\pi}{2^{l+1}}\right)\prod_{l=L+1}^{\infty} \cos^2 \left( \frac{\pi}{2^{l+1}}\right)\leq  \prod_{l=1}^{L} \cos^2 \left( \frac{\pi}{2^{l+1}}\right),
	\end{equation}
	which concludes the proof. Notice that the last inequality is asymptotically tight for $L\to \infty$.
\end{proof}

Let us proceed by introducing a particular super-channel $S_2$ for the detection scheme. The channel $S_2[\Lambda]$ mimics a key property of the Hadamard gate that will allow us to mimic a key property of the inverse Fourier transform such that the protocol yields useful measurement outcomes with high probability.

\begin{lem}\label{HadamardLike}
	Let $\Lambda$ be a qubit quantum channel, defined in the index representation as 
	\begin{equation}
		\Lambda(\ketbra{n}{m} )=\sum_{kl} \Lambda_{kl}^{nm} \ketbra{k}{l}.
	\end{equation}
	There exists an implementable super-channel $S_2\in \DIS$, such that
	\begin{equation}\label{HadamardLikeEquation}
		\Delta S_2[\Lambda](\ketbra{n}{m})=\sum_{k=0}^{1} |\Lambda_{kk}^{nm}| e^{\pi i k(n-m)} \ketbra{k}{k}.
	\end{equation}
	It suffices to choose a super-channel of the form $S[\Lambda]=\Lambda\Phi_2$. We refer to the action of the channel $S_2[\Lambda]$ on any state as Hadamard-like, or shortly the channel is Hadamard-like.
\end{lem}
\begin{proof}
	Let us use the notation $\Lambda_{kl}^{nm}=|\Lambda_{kl}^{nm}| e^{i\lambda_{kl}^{nm}}$ and choose $\Phi_2$ as the channel corresponding to the unitary $\sum_n e^{i\lambda_{00}^{01}n} \ketbra{n}{n}$. In the following, we will see that this choice satisfies our requirements.
	Note first that $\Lambda_{kk}^{nn}\geq 0\, \forall k,n$, and therefore property~\eqref{HadamardLikeEquation} holds for populations. Moreover
	\begin{align}
		\bra{0}\left( \Lambda \Phi_2\ketbra{0}{1}\right) \ket{0}=|\Lambda_{00}^{01}|
	\end{align}
	as required, and due to trace preservation
	\begin{align}
		\bra{1}\left( \Lambda \Phi_2\ketbra{0}{1}\right) \ket{1}=|\Lambda_{00}^{10}| e^{i(\lambda_{11}^{01}-\lambda_{00}^{01})}=-|\Lambda_{00}^{01}|,
	\end{align}
	i.e., $e^{i(\lambda_{11}^{01}-\lambda_{00}^{01})}=-1=e^{i\pi 1(0-1)}$, which finishes this case too. Finally,
	\begin{align}
		\Lambda \Phi_2\ketbra{1}{0}=\left( \Lambda \Phi_2\ketbra{0}{1}\right)^\dagger,
	\end{align}
	from which the remainder of the proof follows.
\end{proof}

\subsection{A lower bound}\label{LowerBoundProof}
A lower bound on the success probability~\eqref{SuccessProb} is essential to bound the runtime of the algorithm. For the coherent protocol, which utilizes Hadamard gates, it has been shown that the success probability is lower bounded by a function that is slowly growing in the number $N$ to factor~\cite{Shor1995}. In this section, we prove a similar bound for less resourceful channels, that will include the coherent bound derived by Shor as a limiting case.

For the lower bound on Eq.~\eqref{SuccessProb} discussed in the following, we can simply choose a specific set of free super-channels $S_1^{(l)}$ and $S_2^{(l)}$, which are depicted in Fig.~\ref{fig4}.
\begin{figure}[ht]
	\centering
	\scalebox{0.95}{
	\begin{quantikz}
&&&&&&\vcw{2}&&&&\vcw{2}&\\
&\gategroup[3,steps=10,style={rounded corners,fill=blue!20, inner xsep=2pt},background, label style={label position=below,anchor=north,yshift=-0.2cm}]{}&\\
\lstick{$\sigma_l$}&\qw &\qw &\gate{\Theta_l} \gategroup[1,steps=2,style={
rounded corners,fill=cornellred!40, inner xsep=2pt},
background]{{\sc $S_1^{(l)}$}} &\gate{\Phi_2^{(l)}} &\ctrl{1}&\gate{R_l^\prime}&\gate{\Phi_3^{(l)}}\gategroup[1,steps=2,style={
rounded corners,fill=cornellred!40, inner xsep=2pt},
background]{{\sc $S_2^{(l)}$}} &\gate{\Lambda_l}&\qw &\meter{}\\
\qwbundle[alternate]{}&\qwbundle[alternate]{}&\qwbundle[alternate]{}&\qwbundle[alternate]{}&\qwbundle[alternate]{}&\gate{\mathcal{U}_l}\qwbundle[alternate]{}&\qwbundle[alternate]{}&\qwbundle[alternate]{}&\qwbundle[alternate]{}&\qwbundle[alternate]{}&\qwbundle[alternate]{}&\qwbundle[alternate]{}
\end{quantikz}}
	\caption{The particular super-channels that are employed for each individual block to derive the lower bound on the success probability in Thm.~\ref{Thm:LowerBound}.}
	\label{fig4}
\end{figure}
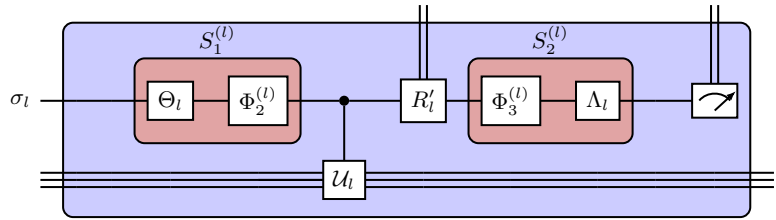

The super-channels $S_2^{(l)}$ employed in the detection part will be the ones that lead to Hadamard-like channels (see Lem.~\ref{HadamardLike}), whereas the $S_1^{(l)}$ will be introduced in the following Theorem.

\begin{thm}\label{Thm:LowerBound} 
	The success probability of the order-finding protocol with qubit channels $\Theta_l$ and unital $\Lambda_l$ is bounded by
	\begin{equation}\label{ThmSM}
		P^{\suc}( \Theta_l ,\Lambda_l) \geq \frac{4}{\pi^2} \frac{\varphi(r)}{r} \prod_{l=1}^{L} \left(\frac{1+\mathscr{C}(\Theta_l) \tilde{M}_\diamond(\Lambda_l)}{2}\right),
	\end{equation}
	where $\mathscr{C}$ denotes the cohering power with respect to the robustness of coherence, $\tilde{M}_\diamond$ is the NSID-measure, both introduced in the main text, and $\varphi(r)$ denotes Euler's totient function.
\end{thm}
\begin{proof}
	Let us consider an idealized version of the order-finding protocol first. Assume that instead of a symmetry channel $\mathcal{E}$, derived in Lem.~\ref{Encoding}, only a single rotation $\mathcal{E}_j $, where $j$ is coprime to $r$, is utilized. Let us denote the success probability of this order-finding protocol by $\tilde{P}^{\suc}_j(\Theta_l,\Lambda_l)$, which is given by
	\begin{equation}
		\tilde{P}^{\suc}_j(\Theta_l,\Lambda_l)=\max_{\tilde{\sigma} \in \mathcal{I}} \sup_{\substack{S_1^{(l)} \in \CIS \\ S_2^{(l)} \in \DIS}}\sum_k P(k\to r \,|\, \text{CFA})\,p_k^{(j)}(S_1^{(l)}[\Theta_l],S_2^{(l)}[\Lambda_l];\tilde{\sigma},\mathbb{P}),
	\end{equation}
	where $p_k^{(j)}(S_1^{(l)}[\Theta_l],S_2^{(l)}[\Lambda_l];\tilde{\sigma},\mathbb{P})=\Tr{P_k\Delta F_{S_2^{(l)}[\Lambda_l]}\mathcal{E}_j \bigotimes_{l}S_1^{(l)}[\Theta_l]\tilde{\sigma}}$ (recall the notations introduced around Eq.~\eqref{eq:succProb}).
	
	One way of obtaining a compact lower bound is the following: Instead of accounting for all possible measurement outcomes which may or may not yield the correct $r$ via the classical post-processing, i.e., all outcomes contained in the set $\mathcal{K}_2^j$ in \eqref{eq:Sets}, we focus on the set $\mathcal{K}_1^j$. Since the set $\mathcal{K}_1^j$ contains exactly one outcome, we use this single measurement outcome $k^\prime$ obeying $|\tfrac{j}{r}-\tfrac{k^\prime}{q}|<1/(2q)$ to provide a lower bound according to
	\begin{equation}\label{BoundIdea}
		\begin{split}
			\tilde{P}^{\suc}_j(\Theta_l,\Lambda_l)&=\max_{\tilde{\sigma} \in \mathcal{I}} \sup_{\substack{S_1^{(l)} \in \CIS \\ S_2^{(l)} \in \DIS}}\sum_k P(k\to r \,|\, \text{CFA})\,p_k^{(j)}(S_1^{(l)}[\Theta_l],S_2^{(l)}[\Lambda_l];\tilde{\sigma},\mathbb{P})\\ 
			&\geq  \max_{\tilde{\sigma} \in \mathcal{I}} \sup_{\substack{S_1^{(l)} \in \CIS \\ S_2^{(l)} \in \DIS}} P( k^\prime \to r \,|\, \text{CFA})\,p_{k^\prime}^{(j)}(S_1^{(l)}[\Theta_l],S_2^{(l)}[\Lambda_l];\tilde{\sigma},\mathbb{P})\\
			&= \max_{\tilde{\sigma} \in \mathcal{I}} \sup_{\substack{S_1^{(l)} \in \CIS \\ S_2^{(l)} \in \DIS}} p_{k^\prime}^{(j)}(S_1^{(l)}[\Theta_l],S_2^{(l)}[\Lambda_l];\tilde{\sigma},\mathbb{P}),
		\end{split}
	\end{equation}
	where in the third line we used the results of Cor.~\ref{cor:goodEstimate} and Lem.~\ref{Unique}.
	To further simplify this bound, we make use of particular super-channels $S_1^{(l)}, S_2^{(l)}$ depicted in Fig.~\ref{fig4}. For the detection we choose $S_2^{(l)}[\Lambda_l]=\Lambda_l\Phi_2^{(l)}$ such that we obtain a Hadamard-like channel (see Lem.~\ref{HadamardLike}). For the adjustment of the channels $\Theta_l$ we do the following: after $\Theta_l$ was applied to $\sigma_l$, we perform a rotation removing the relative phase of the qubit state $\Theta_l(\sigma_l)$. Let us denote this rotation by $\mathcal{R}_1^{(l)}$. Then we apply the map $\tilde{\Phi}(\rho)=\tfrac{1}{2} (\rho+\sigma_x \rho\sigma_x)$. This post-processing of $\Theta_l(\sigma_l)$ results in  a state of the form $S_1^{(l)}[\Theta_l](\sigma_l)=\tfrac{1}{2} \mathbb{1}+c_l\sigma_x$ where $c_l\geq 0$, which is then used to probe $\mathcal{E}_j$. Importantly, it does not carry any intrinsic phases that may interfere with the detection of the phases induced by $\mathcal{E}_j$. Enforcing uniformly distributed populations (which are preserved since $\Lambda_l$ is unital by assumption) will ensure that the deviation in the measurement statistics caused by coherence can be maximized. Choosing super-channels defined in such a way, i.e., $S_1^{(l)}[\Theta_l]=\Phi_1^{(l)}\Theta_l=\tilde{\Phi}\mathcal{R}_1^{(l)}\Theta_l$ and $S_2^{(l)}[\Lambda_l]=\Lambda_l \Phi_2^{(l)}$, we obtain from Eq.~\eqref{BoundIdea} that
	
	\begin{equation}
		\begin{split}
			\tilde{P}^{\suc}_j(\Theta_l,\Lambda_l)
			&\ge \max_{\tilde{\sigma} \in \mathcal{I}} \sup_{\substack{S_1^{(l)} \in \CIS \\ S_2^{(l)} \in \DIS}} p_{k^\prime}^{(j)}(S_1^{(l)}[\Theta_l],S_2^{(l)}[\Lambda_l];\tilde{\sigma},\mathbb{P}) \\
			&\ge \max_{\tilde{\sigma} \in \mathcal{I}} \  p_{k^\prime}^{(j)}(\Phi_1^{(l)}\Theta_l,\Lambda_l \Phi_2^{(l)};\tilde{\sigma},\mathbb{P}) \\
			&=\max_{\tilde{\sigma} \in \mathcal{I}} \ \Tr{P_{k^\prime}\Delta F_{\Lambda_l \Phi_2^{(l)}}\mathcal{E}_j \bigotimes_{l}\Phi_1^{(l)}\Theta_l\tilde{\sigma}}.
		\end{split}
	\end{equation}
	At this point, we notice that we can express $\mathcal{E}_j$ (see  Lem.~\ref{Encoding}) as a tensor product: expanding $n$ into its binary representation, i.e., $n=n_1n_2...n_L=\sum_{l=1}^{L}n_l2^{L-l}$, we find 
	\begin{align}
		R_{j/r}=&\sum_{n=0}^{2^{L}-1} e^{2\pi i \frac{j}{r}n}\ketbra{n}{n} \\
		=&\sum_{n_1=0}^{1}\ldots \sum_{n_L=0}^{1} e^{2\pi i \frac{j}{r}\sum_{l=1}^{L}n_l2^{L-l}}\ketbra{n_1n_2...n_L}{n_1n_2...n_L} \\
		=&\bigotimes_{l=1}^L \sum_{n_l=0}^{1} e^{2\pi i \frac{j}{r} n_l 2^{L-l}}\ketbra{n_l}{n_l} \\
		=& \bigotimes_{l=1}^L R_{j/r}^{(l)},
	\end{align}
	with $R_{j/r}^{(l)}= \sum_{n_l=0}^{1} e^{2\pi i \frac{j}{r} n_l 2^{L-l}}\ketbra{n_l}{n_l}$. We thus define $\mathcal{E}_j^{(l)}(\rho):=R_{j/r}^{(l)}\rho \left(R_{j/r}^{(l)}\right)^\dagger$ and notice that, with the equivalence of Figs.~\ref{fig:ShorDecomposition}, \ref{fig:ShorDecompositionComp}, and~\ref{BlockFormHadamards} (and $\tilde{\sigma}=\bigotimes_{l}\sigma_l$), 
	\begin{align}
		\Tr{P_{k^\prime}\Delta F_{\Lambda_l \Phi_2^{(l)}}\mathcal{E}_j \bigotimes_{l=1}^L\Phi_1^{(l)}\Theta_l\tilde{\sigma}}=&\Tr{P_{k^\prime}\Delta F_{\Lambda_l \Phi_2^{(l)}}\bigotimes_{l=1}^L\left(\mathcal{E}_j^{(l)} \Phi_1^{(l)}\Theta_l\sigma_l\right)} \nonumber\\
		=&\prod_{l=1}^L \Tr{P_{k'}^{(l)} \Delta \Lambda_l \Phi_2^{(l)} R'_l \mathcal{E}_j^{(l)} \Phi_1^{(l)} \Theta_l\sigma_l},
	\end{align}
	where $P_{k'}^{(l)}=\ketbra{k'_{l-1}}{k'_{l-1}}$ for a total $k'=\sum_{l=0}^{L-1} 2^l k_l'$ (see Figs.~\ref{fig:ShorDecomposition} and~\ref{fig:ShorDecompositionComp}). Here it is important that we understand the product as ordered, since $R'_l$ depends on all previous measurement outcomes.
	Recall that $c_l=\left| \left[\Theta_l(\sigma_l)\right]_{01}\right|$, with which
	\begin{align}
		\Delta \Lambda_l \Phi_2^{(l)} R'_l \mathcal{E}_j^{(l)} \Phi_1^{(l)} \Theta_l\sigma_l=&\Delta \Lambda_l \Phi_2^{(l)} R'_l \mathcal{E}_j^{(l)} \left[\frac{1}{2}\id+c_l\sigma_x\right] \nonumber \\
		=&\Delta \Lambda_l \Phi_2^{(l)} R'_l \left[\frac{1}{2}\id+c_l\left(e^{-2 \pi i\frac{j}{r} 2^{L-l}} \ketbra{0}{1}+h.c. \right)\right] \nonumber \\
		=&\Delta \Lambda_l \Phi_2^{(l)}  \left[\frac{1}{2}\id+c_l\left(e^{-2 \pi i \left(\frac{j}{r} 2^{L-l}-\sum_{a=2}^l k'_{l-a}/2^a\right)} \ketbra{0}{1}+h.c. \right)\right] \nonumber \\
		=&\frac{1}{2}\id+c_l\left(e^{-2 \pi i \left(\frac{j}{r} 2^{L-l}-\sum_{a=2}^l k'_{l-a}/2^a\right)} \sum_{b_l=0}^1 \left|(\Lambda_l)_{b_l b_l}^{01}\right| e^{i\pi b_l}\ketbra{b_l}{b_l}+h.c. \right).
	\end{align}
	Since the robustness of coherence coincides with the $l_1$ measure of coherence for qubits, see Ref.~\cite{Napoli2016}, we find
	\begin{align}\label{eq:RobustnessQubit}
		\max_{\sigma_l \in \mathcal{I}}c_l= \max_{\sigma_l \in \mathcal{I}} \left| \left[\Theta_l(\sigma_l)\right]_{01}\right| =\max_{\sigma_l \in \mathcal{I}} C(\Theta_l \sigma_l)/2=\mathscr{C}(\Theta_l)/2,
	\end{align}
	where $\mathscr{C}$ denotes the cohering power with respect to the robustness, and $\left|(\Lambda_l)_{0 0}^{01}\right|=\left|(\Lambda_l)_{1 1}^{01}\right|$~\cite[Prop.~6]{Masini2021} and thus
	\begin{align}
		&\max_{\tilde{\sigma} \in \mathcal{I}} \ \Tr{P_{k^\prime}\Delta F_{\Lambda_l \Phi_2^{(l)}} \bigotimes_{l}\mathcal{E}_j^{(l)}\Phi_1^{(l)}\Theta_l\tilde{\sigma}} \nonumber \\
		&= \prod_{l=1}^L \max_{\sigma_l \in \mathcal{I}} \left[\frac{1}{2}+c_l \left|(\Lambda_l)_{k_{0}' k_{0}'}^{01}\right| \left(e^{-2 \pi i \left(\frac{j}{r} 2^{L-l}-\sum_{a=1}^l k'_{l-a}/2^a\right)}   +h.c. \right)  \right] \nonumber \\
		&= \prod_{l=1}^L \frac{1}{2}\left[1+\mathscr{C}(\Theta_l) \left|(\Lambda_l)_{0 0}^{01}\right| \left(e^{-2 \pi i \left(\frac{j}{r} 2^{L-l}-\sum_{a=1}^l k'_{l-a}/2^a\right)}   +h.c. \right)  \right] .
	\end{align}
	Following the usual procedure (see for example Ref.~\cite{Nielsen2016}), we note that $\sum_{a=1}^{l}k_{l-a}^\prime/2^a=2^{-l}\sum_{b=0}^{l-1}k_b' 2^b$ and $e^{2\pi i 2^{-l}\sum_{b=l}^{L-1}k_b' 2^b}=1$. Therefore,
	\begin{align}
		e^{-2 \pi i \left(\frac{j}{r} 2^{L-l}-\sum_{a=1}^l k'_{l-a}/2^a\right)} =e^{-2 \pi i 2^{L-l}(\frac{j}{r}-\frac{k'}{2^L})}
	\end{align}
	and if we consider the worst-case scenario we find
	\begin{align}
		\max_{\tilde{\sigma} \in \mathcal{I}} \ \Tr{P_{k^\prime}\Delta F_{\Lambda_l \Phi_2^{(l)}}\mathcal{E}_j \bigotimes_{l}\Phi_1^{(l)}\Theta_l\tilde{\sigma}}=& \prod_{l=1}^L \frac{1}{2}\left[1+\mathscr{C}(\Theta_l) \left|(\Lambda_l)_{0 0}^{01}\right| 2\cos\left(2 \pi 2^{L-l} \left(\frac{j}{r}- \frac{k'}{2^L}\right)\right)  \right] \nonumber \\
		\ge&\inf_{|\chi|<\tfrac{1}{2q}}\prod_{l=1}^L \frac{1}{2}\left[1 +\mathscr{C}(\Theta_l) \left|(\Lambda_l)_{0 0}^{01}\right| 2 \cos\left(2 \pi 2^{L-l} \chi \right)  \right],
	\end{align}
	where in the last line we used our assumption that $k'\in \mathcal{K}_1^j$. Since  $2|\Lambda_{00}^{01}|=\tilde{M}_\diamond(\Lambda)$, as detailed in Lem.~\ref{lem:DetectabilityNSID}, it follows that
	\begin{align}
		\tilde{P}^{\suc}_j(\Theta_l,\Lambda_l)&\ge \inf_{|\chi|<\tfrac{1}{2q}}\prod_{l=1}^L \frac{1}{2}\left[1 +\mathscr{C}(\Theta_l) \tilde{M}_\diamond(\Lambda_l) \cos\left(2 \pi 2^{L-l} \chi \right)  \right] \nonumber \\
		&= \prod_{l=1}^L \frac{1}{2}\left[1 +\mathscr{C}(\Theta_l) \tilde{M}_\diamond(\Lambda_l) \cos\left( \pi 2^{-l}  \right)  \right] \nonumber \\
		&\overset{\text{Lem.}~\ref{4PiBound}}{\geq}   \frac{4}{\pi^2}\prod_{l=1}^L \frac{1}{2}\left[1 +\mathscr{C}(\Theta_l) \tilde{M}_\diamond(\Lambda_l)  \right].
	\end{align}
	
	Now remember that up to here, we assumed that we replaced $\mathcal{E}$ with $\mathcal{E}_j$. This is of course not possible since it would require knowledge of $r$. To get back to the original protocol, we note that applying $\mathcal{E}$ corresponds to applying  $\mathcal{E}_j$ with $j\in 0,...,r-1$ chosen uniformly at random. The number of such $j$ with $\text{gcd}(j,r)=1$ is given by $\varphi(r)$, where $\varphi(r)$ denotes Euler's totient function. The overall success probability is therefore bounded by
	
	\begin{equation}
		\begin{split}
			P^{\suc}(\Theta_l,\Lambda_l)&=\frac{1}{r}\sum_{j=0}^{r-1}\tilde{P}^{\suc}_j(\Theta_l,\Lambda_l) \\
			&\ge \left(\frac{\varphi(r)}{r}\right) \frac{4}{\pi^2}\prod_{l=1}^L \frac{1}{2}\left[1 +\mathscr{C}(\Theta_l) \tilde{M}_\diamond(\Lambda_l)  \right].
		\end{split}
	\end{equation}
	Particularly, if the same channels are utilized in each block the bound simplifies to 
	\begin{equation}
		P^{\suc}(\Theta,\Lambda)\ge \frac{4}{\pi^2} \left(\frac{\varphi(r)}{r}\right) \left[\frac{1+\mathscr{C}(\Theta) \tilde{M}_\diamond(\Lambda)}{2} \right]^L.
	\end{equation}
\end{proof}
Euler's totient function grows almost linearly in its argument and is strictly bounded by $\varphi(r)>\tfrac{\delta r}{\log\log r}>\tfrac{\delta r}{\log\log N}$ for some $\delta>0$, where $\delta \approx e^{-\gamma}$ with $\gamma$ being the Euler-Mascheroni constant, see for instance Ref. \cite[Theorem 328]{Hardy1984}, which connects this bound to the original bound derived by Shor~\cite{Shor1995,Parker2000}. For a perfectly coherent protocol, this bound would take the form $P^{\suc}> \tfrac{4}{\pi^2} \tfrac{\delta}{\log\log r}$, which equals the bound originally obtained by Shor~\cite{Shor1995}. In the following works, see for example Refs.~\cite{Gerjuoy2005,Bourdon}, it has been shown that for the fully coherent protocol, the factor $\tfrac{4}{\pi^2}\approx 0.4$ can be pushed to about $0.9$ (at least in an average case) by a more careful, yet tedious, analysis. The basic idea behind these proofs is to consider not only the set $\mathcal{K}_1$ as useful outcomes but to stretch the definition of said set as it has been outlined in Cor.~\ref{cor:goodEstimate}.
Since continuity in the dynamical measures $\mathscr{C}(\Theta)$ and $\tilde{M}_\diamond(\Lambda)$ is to be expected, it would not be surprising if the bound in Eq.~\eqref{ThmSM} can be sharpened analogously. For now, we leave this to future work.

\subsection{Classical limit}
As already pointed out, the classical limit of the protocol uses only free channels $\Theta_{\free}^{(l)}$ and $\Lambda_{\free}^{(l)}$ and corresponds to a random number generator. It returns a number in the range $0\leq k\leq 2^L-1$ with a probability distribution of $\lbrace p_k(S_1^{(l)}[\Theta_{\free}^{(l)}],S_2^{(l)}[\Lambda_{\free}^{(l)}]; \tilde\sigma,\mathbb{P})\rbrace_k$ independent of the order $r$ since
\begin{equation}
	\begin{split}
		p_k(S_1^{(l)}[\Theta_{\free}^{(l)}],S_2^{(l)}[\Lambda_{\free}^{(l)}];\tilde\sigma,\mathbb{P})&=\Tr{P_k\Delta F_{S_2^{(l)}[\Lambda_{\free}^{(l)}]} \mathcal{E} \bigotimes_{l}S_1^{(l)}[\Theta_{\free}^{(l)}](\tilde\sigma)}\\
		&= \Tr{P_k\Delta F_{S_2^{(l)}[\Lambda_{\free}^{(l)}]}\Delta \mathcal{E} \bigotimes_{l}S_1^{(l)}[\Theta_{\free}^{(l)}](\tilde\sigma)}\\
		&=\Tr{P_k\Delta F_{S_2^{(l)}[\Lambda_{\free}^{(l)}]} \bigotimes_{l}S_1^{(l)}[\Theta_{\free}^{(l)}](\tilde\sigma)}.
	\end{split}
\end{equation}
Without prior knowledge about the order $r$ (including factors of $r$ itself which may be obtained by considering the outcomes of multiple rounds combined; not considered here though), the on average most beneficial probability distribution $p_k$ is the uniform distribution. For all free channels $\Theta_{\free}^{(l)},\Lambda_{\free}^{(l)}$ we can always choose a pair $S_1^{(l)},S_2^{(l)}$ that achieves this uniform distribution. In fact, such super-channels can even be chosen independently of $\Theta_{\free}^{(l)}$ and $\Lambda_{\free}^{(l)}$ in the classical limit, even in the case of non-unital $\Lambda_{\text{free}}^{(l)}$. 
A simple example would be to choose a suitable replacement channel as the post-processing of $\Lambda_{\text{free}}^{(l)}$.
The resulting uniformly random measurement outcome is forwarded to the continued fraction algorithm producing an estimate on $r$.  Thus, the overall success probability (assuming no prior knowledge of $r$) in the classical limit is given by
\begin{equation}\label{eq:ClassicalSucc}
	P^{\suc}(\Theta_{\free}^{(l)},\Lambda_{\free}^{(l)})= \frac{f(N,r)}{2^L},
\end{equation}
where we define $f(N,r)=\sum_k P(k\to r \,|\, \text{CFA})$. The function $f(N,r)$ mitigates the exponential term in the success probability to some extend and to quantify this notion we proceed to derive bounds on this function.
\begin{prop}\label{prop:FBound}
	The function $f(N,r)$ is bounded by 
	\begin{equation}
		2\varphi(r)\floor*{\frac{q-1}{2r^2}}\le f(N,r)\le \varphi(r)\left(1+2\floor*{\frac{q}{r^2}} \right),
	\end{equation}
	where $\varphi(r)$ denotes Euler's totient function and $q$ is uniquely given by $N^2<q=2^L<2N^2$. 
\end{prop}
\begin{proof}
	Consider a single coprime pair $(j,r)$. Let $\tilde{f}_j(N,r)$ denote the function that counts the number of measurement outcomes that lead to this particular convergent $j/r$, i.e., 
	\begin{equation}
		\tilde{f}_j(N,r)= \sum_k P(k\to (j,r)\, |\,\text{CFA}).
	\end{equation}
	For a lower bound, recall Thm.~\ref{thm::Approximation} and Cor.~\ref{cor:goodEstimate}. Let us define the set $\mathcal{K}_1^j(\beta)$ which contains all integers that surely allow for a successful post-processing, i.e., 
	\begin{equation}\label{eq:fUpper}
		\mathcal{K}_1^j(\beta)=\left\lbrace k:0\leq k<q \wedge \,  \left|\frac{j}{r}-\frac{k}{q}\right|\le \frac{\beta}{2q} \right\rbrace,
	\end{equation}
	where $\beta=\tfrac{q-1}{r^2}$. A lower bound on $\tilde{f}_j(N,r)$ is then given by 
	\begin{equation}\label{eq:fLower}
		\begin{split}
			\tilde{f}_j(N,r)&=\sum_k P(k\to (j,r)\, |\,\text{CFA}) \ge \sum_{k\in \mathcal{K}_1^j(\beta)} P(k\to (j,r)\, |\,\text{CFA})= \sum_{k\in \mathcal{K}_1^j(\beta)} 1 = |\mathcal{K}_1^j(\beta)|.
		\end{split}
	\end{equation}
	Furthermore, according to the second part of Thm.~\ref{thm::Approximation} all measurement outcomes that yield the pair $(j,r)$ as a convergent of $k/q$ are contained in the set $\mathcal{K}_2^j$, as introduced in Eq.~\eqref{eq:Sets}. Thus the function can be upper bounded as
	\begin{equation}
		\tilde{f}_j(N,r)= \sum_k P(k\to (j,r)\, |\,\text{CFA}) = \sum_{k \in \mathcal{K}_2^j } P(k\to (j,r)\, |\,\text{CFA}) \le \sum_{k \in \mathcal{K}_2^j } 1= |\mathcal{K}_2^j| .
	\end{equation}
	To further simplify these bounds, we proceed to bound the cardinalities of $\mathcal{K}_2^j$ and $\mathcal{K}_1^j(\beta)$. We start with a lower bound on $|\mathcal{K}_1^j(\beta)|$.
	Consider the closest fraction $\tfrac{k^\prime}{q}$ defined by the smallest distance to the fraction $\tfrac{j}{r}$. This integer $k^\prime$ is roughly in the center of the set defined by $\mathcal{K}_1^j(\beta)$, and also the set $\mathcal{K}_2^j$.  Now consider the adjacent integers $k=k^\prime \pm n$ to the closets integer $k^\prime$. First assume $\tfrac{j}{r}-\tfrac{k^\prime}{q}> 0$. Then for the elements to the left of $k^\prime$, i.e., $k=k^\prime-n$ contained in $\mathcal{K}_1^j(\beta)$ we have
	\begin{equation}\label{eq:boundK1j}
		\frac{\beta}{2q}\geq \frac{j}{r}-\frac{k}{q}=\left(\frac{j}{r}-\frac{k^\prime}{q}\right)+\frac{n}{q}>\frac{n}{q},
	\end{equation}
	and therefore
	\begin{align}
		n<\tfrac{q-1}{2r^2}.
	\end{align}
	If $\tfrac{q-1}{2r^2}$ is an integer, all natural numbers $n\le n_\text{max}=\tfrac{q-1}{2r^2}-1=$ satisfy this equation and therefore lead to a $k$ in $\mathcal{K}_1^j(\beta)$. Moreover, due to our closeness assumption of $k^\prime$, larger $n$ cannot be in $\mathcal{K}_1^j(\beta)$. In this case, also all $k=k'+m$ with $m\le m_\text{max}=\tfrac{q-1}{2r^2}$ are contained in $\mathcal{K}_1^j(\beta)$, because we assumed $\tfrac{j}{r}-\tfrac{k^\prime}{q}> 0$, i.e., there cannot be less integers $k>k'$ in $\mathcal{K}_1^j(\beta)$ than integers $k<k'$ and we cannot hit the boundary twice exactly. In sum, we find $|\mathcal{K}_1^j(\beta)|=1+n_\text{max}+m_\text{max}=2\tfrac{q-1}{2r^2}$. If $\tfrac{q-1}{2r^2}$ is not an integer, we take $n_\text{max}=\floor{\frac{q-1}{2r^2}}$ instead, and $|\mathcal{K}_1^j(\beta)|=1+2n_\text{max}=1+2\floor{\frac{q-1}{2r^2}}$. Combining both cases, we have $|\mathcal{K}_1^j(\beta)|\ge2\floor{\frac{q-1}{2r^2}}$.
	
	If $\tfrac{j}{r}-\tfrac{k^\prime}{q}< 0$, the same bound holds true, which can be seen by switching the role of $n_\text{max}$ and $m_\text{max}$, i.e., switching the intervals to the left and right. Lastly, consider the case of $\tfrac{j}{r}-\tfrac{k^\prime}{q}= 0$ and $k=k^\prime  \pm n$. From 
	\begin{equation}
		\frac{\beta}{2q}\geq\left| \frac{j}{r}-\frac{k}{q}\right|=\left|\left(\frac{j}{r}-\frac{k^\prime}{q}\right)+\frac{n}{q}\right|=\frac{|n|}{q},
	\end{equation}
	follows that all integers $k=k'+n$ with $|n|\le  \floor{\tfrac{q-1}{2r^2}}$ are contained in $\mathcal{K}_1^j(\beta)$, i.e., $|\mathcal{K}_1^j(\beta)|\ge 1+2 \floor{\tfrac{q-1}{2r^2}}$.  Combining all cases, we we find
	\begin{equation}\label{eq:CardLower}
		|\mathcal{K}_1^j(\beta)|\ge\min\left\lbrace 2\floor{\tfrac{q-1}{2r^2}}, 1+2\floor{\tfrac{q-1}{2r^2}} \right\rbrace= 2\floor{\tfrac{q-1}{2r^2}}.
	\end{equation}

	Now we continue to obtain an upper bound on the cardinality of $|\mathcal{K}_2^j|$. All integers $k\in \mathcal{K}_2^j$ obey $|\tfrac{j}{r}-\tfrac{k}{q}|\leq \tfrac{1}{r^2}$ by definition.
	Again consider the closest fraction $\tfrac{k^\prime}{q}$ defined as the one with the smallest difference to the fraction $\tfrac{j}{r}$. As in the discussion for the lower bound, assume $\tfrac{j}{r}-\tfrac{k'}{q}>0$. From the analogue of Eq.~\eqref{eq:boundK1j} follows that $k=k'-n$ is an element of $\mathcal{K}_2^j$ if  $n<\tfrac{q}{r^2}$. If $q/r^2$ is an integer then $n_\text{max}=\tfrac{q}{r^2}-1$, and for the same arguments as before, $m_\text{max}=\tfrac{q}{r^2}$. If $q/r^2$ is not an integer,  $n_\text{max}=\floor{\tfrac{q}{r^2}}$. In addition,  $m_\text{max}=\floor{\tfrac{q}{r^2}}$. Depending on $\tfrac{q}{r^2}$ being an integer or not, the cardinality is given by $|\mathcal{K}_2^j|=1+2\floor{\tfrac{q}{r^2}}$ or $|\mathcal{K}_2^j|=2\floor{\tfrac{q}{r^2}}$, thereby the cardinality is bounded by
	\begin{equation}\label{eq:CardUpper}
		|\mathcal{K}_2^j|\leq 1+2\floor{\tfrac{q}{r^2}}.
	\end{equation}
	In the remaining case, i.e., if $\tfrac{j}{r}-\tfrac{k}{q}=0$, we find $|\mathcal{K}_2^j|=1+2\floor{\tfrac{q}{r^2}}$, whether or not $\tfrac{q}{r^2}$ is an integer.

	To conclude the proof, take into account all possible integers $j$ that are smaller then and coprime to $r$. There are exactly $\varphi(r)$ such values for $j$, and correspondingly for each such $j$, there is a range of possible outcomes $k$ that lead to the respective pair $(j,r)$. Inserting the expressions Eqs.~\eqref{eq:CardLower} and~\eqref{eq:CardUpper} into
	Eqs.~\eqref{eq:fUpper} and~\eqref{eq:fLower} respectively, we see that the function $\tilde{f}_j(n,r)$ can be bounded by
	\begin{align}
		f(N,r)=\sum_k P(k\to r\, |\,\text{CFA})=\sum_{j \text{ coprime to }r} \tilde{f}_j(N,r) \ge \sum_{j \text{ coprime to }r} 2\floor*{\frac{q-1}{2r^2}} =2\varphi(r)\floor*{\frac{q-1}{2r^2}} \nonumber ,
	\end{align}
	and
	\begin{align}
		f(N,r)=\sum_k P(k\to r\, |\,\text{CFA})=\sum_{j \text{ coprime to }r} \tilde{f}_j(N,r) \le \sum_{j \text{ coprime to }r} \left(1+2\floor*{\frac{q}{r^2}}\right) =\varphi(r)\left(1+2\floor*{\frac{q}{r^2}} \right) \nonumber.
	\end{align}
	Note that for all $(j,r)$ with $j$ not coprime to $r$, the continued fraction algorithm will yield a factor of $r$. Whilst this information can be used in principle, it is not relevant for the fixed post-processing strategy that we chose. \qedhere
\end{proof}

Let us conclude this section by noting that with the result of Prop.~\ref{prop:FBound} and Eq.~\eqref{eq:ClassicalSucc}, we can provide bounds on the classical limit of the success probability, i.e., 
\begin{equation}
	2\frac{\varphi(r)}{2^L} \floor*{\frac{2^L-1}{2r^2}}\leq P^{\suc}(\Theta_{\free}^{(l)},\Lambda_{\free}^{(l)})\le \frac{\varphi(r)}{2^L}\left(1+2\floor*{\frac{2^L}{r^2}} \right),
\end{equation}
where we used $q=2^L$. The classical limit of the success probability is thus sensibly dependent on the ratio between $2^L$ and $r^2$. Since $N^2<2^L<2N^2$, this is a purely problem specific expression, in the sense that it only depends on the number $N$ to factor and a corresponding order $r$.

\section{An upper bound}
From a complexity theoretic perspective, providing an upper bound on the success probability is rather uninteresting, since it corresponds to a best-case scenario. On the other hand, an upper bound is an interesting question, if we want to attribute a potential speed-up to a resource, i.e., in our case coherence. For this reason we use a similar technique as in the classical limit to provide a sufficiently general upper bound on the performance of the protocol. Nevertheless, the bound is general enough to provide quantitative insights on the role of coherence in the algorithm.

\begin{thm}\label{thm:UpperBoundSM}
	The success probability of the order-finding protocol with qubit channels $\Theta_l$ and unital $\Lambda_l$ is bounded by
	\begin{equation}
		P^\suc(\Theta_l,\Lambda_l)\leq \min \left\lbrace \frac{\varphi(r)}{2^L}\left(1+2\floor{\tfrac{2^L}{r^2}} \right)\prod_{l=1}^{L}\left( 1 + \mathscr{C}(\Theta_l) \tilde{M}_\diamond(\Lambda_l)\right)  ,\, 1 \right\rbrace,
	\end{equation}
	where $\mathscr{C}$ denotes the cohering power with respect to the robustness of coherence, $\tilde{M}_\diamond$ is the NSID-measure, and $\varphi(r)$ is Euler's totient function.
\end{thm}
\begin{proof}
	Again, consider an idealized protocol with only a single rotation $\mathcal{E}_j$ first. Recall the notations introduced around Eq.~\eqref{eq:succProb}. We now need to be more careful than in the lower bound and use a similar technique as in the classical limit. From the quantum part of the protocol we obtain a measurement outcome, i.e., an integer $k$, with a probability depending on the rotation $\mathcal{E}_j$. The classical post-processing succeeds by definition if it returns a coprime pair $(j^\prime,r)$, where $r$ is the order we are looking for. Even if a measurement outcome $k$ does not lead to $j/r$, it could still be close enough to another coprime fraction $j^\prime/r$ such that the post-processing succeeds. This means we have to account for all possible coprime integers $j^\prime$ and hence their corresponding integers $k$ that allow to estimate $j^\prime/r$. The success probability is given by
	\begin{equation}\label{eq:UpperBoundCalc}
		\begin{split}
			\tilde{P}_j^{\suc}(\Theta_l,\Lambda_l) =&\sup_{\substack{S_1^{(l)} \in \CIS \\ S_2^{(l)}\in \DIS}} \max_{\tilde \sigma \in \mathcal{I}} \sum_k P(k\to r \, |\, \text{CFA} )p_k^{(j)}(S_1^{(l)}[\Theta_l],S_2^{(l)}[\Lambda_l];\tilde\sigma, \mathbb{P}) \\
			&=\sup_{\substack{S_1^{(l)} \in \CIS \\ S_2^{(l)}\in \DIS}} \max_{\tilde \sigma \in \mathcal{I}} \sum_{j^\prime \perp r} \sum_k P(k\to (j^\prime,r) \, |\, \text{CFA} )p_k^{(j)}(S_1^{(l)}[\Theta_l],S_2^{(l)}[\Lambda_l];\tilde\sigma, \mathbb{P}),
		\end{split}
	\end{equation}
	where in the second line we used that we only consider direct estimates of $r$ but not factors of $r$, thus we sum only over all coprime $j^\prime$. Recall the necessary condition that integers $k$ leading to $j^\prime/r$ are contained in $\mathcal{K}_2^{j^\prime}$, hence
	\begin{equation}
		\begin{split}
			\tilde{P}_j^{\suc}(\Theta_l,\Lambda_l) &=\sup_{\substack{S_1^{(l)} \in \CIS \\ S_2^{(l)}\in \DIS}} \max_{\tilde \sigma \in \mathcal{I}} \sum_{j^\prime \perp r} \sum_k P(k\to (j^\prime,r) \, |\, \text{CFA} )p_k^{(j)}(S_1^{(l)}[\Theta_l],S_2^{(l)}[\Lambda_l];\tilde\sigma, \mathbb{P}) \\
			&\leq \sup_{\substack{S_1^{(l)} \in \CIS \\ S_2^{(l)}\in \DIS}} \max_{\tilde \sigma \in \mathcal{I}} \sum_{j^\prime \perp r} \sum_{k \in \mathcal{K}_2^{j^\prime}} p_k^{(j)}(S_1^{(l)}[\Theta_l],S_2^{(l)}[\Lambda_l];\tilde\sigma, \mathbb{P}) \\
			&\leq  \sup_{\substack{S_1^{(l)} \in \CIS \\ S_2^{(l)}\in \DIS}} \max_{\tilde \sigma \in \mathcal{I}} \sum_{j^\prime \perp r} \max_{j^\prime \perp r} \sum_{k \in \mathcal{K}_2^{j^\prime}}   p_k^{(j)}(S_1^{(l)}[\Theta_l],S_2^{(l)}[\Lambda_l];\tilde\sigma, \mathbb{P})\\
			&=\varphi(r)\sup_{\substack{S_1^{(l)} \in \CIS \\ S_2^{(l)}\in \DIS}} \max_{\tilde \sigma \in \mathcal{I}}  \max_{j^\prime \perp r}\sum_{k \in \mathcal{K}_2^{j^\prime}}  p_k^{(j)}(S_1^{(l)}[\Theta_l],S_2^{(l)}[\Lambda_l];\tilde\sigma, \mathbb{P}).
		\end{split}
	\end{equation}
	Recall from the proof of Thm.~\ref{Thm:LowerBound} that
	\begin{align}
		p_k^{(j)}(S_1^{(l)}[\Theta_l],S_2^{(l)}[\Lambda_l];\tilde\sigma, \mathbb{P})=\prod_{l=1}^L \Tr{P_{k}^{(l)} \Delta S_2^{(l)}[\Lambda_l] R'_l \mathcal{E}_j^{(l)} S_1^{(l)}[\Theta_l] \sigma_l}.
	\end{align}
	As in the proof for the lower bound, we derive an expression for the measurement statistics for arbitrary super-channels. We can write
	\begin{align}
		\Delta S_2^{(l)}[\Lambda_l]  R'_l \mathcal{E}_j^{(l)} S_1^{(l)}[\Theta_l] \sigma_l&= \Delta S_2^{(l)}[\Lambda_l]\Delta R'_l \mathcal{E}_j^{(l)} S_1^{(l)}[\Theta_l] \sigma_l+\Delta S_2^{(l)}[\Lambda_l](\id-\Delta) R'_l \mathcal{E}_j^{(l)} S_1^{(l)}[\Theta_l] \sigma_l\nonumber\\
		&=\Delta S_2^{(l)}[\Lambda_l]\Delta S_1^{(l)}[\Theta_l] \sigma_l+\Delta S_2^{(l)}[\Lambda_l](\id-\Delta) R'_l \mathcal{E}_j^{(l)} S_1^{(l)}[\Theta_l] \sigma_l \nonumber\\
		&=\Delta S_2^{(l)}[\Lambda_l]\Delta S_1^{(l)}[\Theta_l] \sigma_l+\Delta S_2^{(l)}[\Lambda_l] R'_l \mathcal{E}_j^{(l)}(\id-\Delta) S_1^{(l)}[\Theta_l] \sigma_l,
	\end{align}
	where in the last line we used that $\id-\Delta$ commutes with the rotations.
	Let us focus on the second term first. We define
	\begin{align}
		& \left[S_1^{(l)}[\Theta_l](\sigma_l)\right]_{01}=\left| \left[S_1^{(l)}[\Theta_l](\sigma_l)\right]_{01}\right| e^{i\phi_1}=c_l e^{i\phi_1}, \nonumber\\
		& (S_2^{(l)}[\Lambda_l])_{b_l,b_l}^{01}=|(S_2^{(l)}[\Lambda_l])_{b_l,b_l}^{01}| e^{i\lambda_{b_l,b_l}^{01}}. 
	\end{align}
	Then we can write the second term as
	\begin{align}\label{eq:CoherenceTermBound}
		&\Delta S_2^{(l)}[\Lambda_l] R'_l \mathcal{E}_j^{(l)}(\id-\Delta) S_1^{(l)}[\Theta_l] \sigma_l\nonumber\\
		&= \Delta S_2^{(l)}[\Lambda_l] R'_l \left(c_l e^{-2 \pi i\frac{j}{r} 2^{L-l}}e^{i \phi_1} \ketbra{0}{1}+h.c.  \right) \nonumber \\
		&=\Delta S_2^{(l)}[\Lambda_l]\left( c_l e^{-2 \pi i \left(\frac{j}{r} 2^{L-l}-\sum_{a=2}^l k'_{l-a}/2^a\right)} e^{i\phi_1} \ketbra{0}{1}+h.c. \right) \nonumber\\
		&= c_l\left(e^{-2 \pi i \left(\frac{j}{r} 2^{L-l}-\sum_{a=2}^l k'_{l-a}/2^a\right)}e^{i\phi_1} \sum_{b_l=0}^1  \left|(S_2^{(l)}[\Lambda_l])_{b_l b_l}^{01}\right| e^{i\lambda_{b_l,b_l}^{01}} \ketbra{b_l}{b_l}+h.c. \right)\nonumber\\
		&= \left| \left[S_1^{(l)}[\Theta_l](\sigma_l)\right]_{01}\right| \left|(S_2^{(l)}[\Lambda_l])_{00}^{01}\right|\left(e^{-2 \pi i \left(\frac{j}{r} 2^{L-l}-\sum_{a=2}^l k'_{l-a}/2^a\right)}e^{i\phi_1} \sum_{b_l=0}^1  e^{i\lambda_{b_l,b_l}^{01}} \ketbra{b_l}{b_l}+h.c. \right).
	\end{align}
	where we used in the last line that $\left|(S_2^{(l)}[\Lambda_l])_{0 0}^{01}\right|=\left|(S_2^{(l)}[\Lambda_l])_{1 1}^{01}\right|$~\cite[Prop.~6]{Masini2021}. Let us introduce the abbreviation $q_{k_l}(S_1^{(l)}[\Theta_l],S_2^{(l)}[\Lambda_l];\sigma_l, \mathbb{P} )=\Tr{P_{k}^{(l)} \Delta S_2^{(l)}[\Lambda_l]\Delta S_1^{(l)}[\Theta_l]\sigma_l}$. We evaluate the projective measurement and bound the phase dependent terms by two to find
	\begin{align}
		&\tilde{p}_k^{(j)}(S_1^{(l)}[\Theta_l],S_2^{(l)}[\Lambda_l];\tilde\sigma, \mathbb{P})\nonumber\\
		&\leq \prod_{l=1}^{L} \left(q_{k_l}(S_1^{(l)}[\Theta_l],S_2^{(l)}[\Lambda_l];\sigma_l, \mathbb{P}) +2\left| \left[S_1^{(l)}[\Theta_l](\sigma_l)\right]_{01}\right| \left|(S_2^{(l)}[\Lambda_l])_{00}^{01}\right|  \right)\nonumber\\
		&= \prod_{l=1}^{L} \left(q_{k_l}(S_1[\Theta],S_2[\Lambda];\sigma_l, \mathbb{P}) + \tfrac{1}{2}C(S_1^{(l)}[\Theta_l](\sigma_l)) \tilde{M}_\diamond(S_2^{(l)}[\Lambda])\right),
	\end{align}
	where we used the expressions for the measures from Eq.~\eqref{eq:RobustnessQubit} and Lem.~\ref{lem:DetectabilityNSID}. Using this bound on the probability to measure an outcome $k$ results in a bound on the success probability given by  
	\begin{align}
		\tilde{P}_j^{\suc}(\Theta_l,\Lambda_l)&\leq \varphi(r)\sup_{\substack{S_1 \in \CIS \\ S_2\in \DIS}} \max_{\tilde \sigma \in \mathcal{I}}  \max_{j^\prime \perp r}\\
		& \sum_{k \in \mathcal{K}_2^{j^\prime}} \prod_{l=1}^{L} \left(q_{k_l}(S_1[\Theta],S_2[\Lambda];\sigma_l, \mathbb{P}) + \tfrac{1}{2}C(S_1^{(l)}[\Theta_l](\sigma_l)) \tilde{M}_\diamond(S_2^{(l)}[\Lambda])\right).\nonumber
	\end{align}
	Only the contribution $q_{k_l}(S_1^{(l)}[\Theta_l],S_2^{(l)}[\Lambda_l];\sigma_l, \mathbb{P})$ emerging from the distribution of the initial populations depends implicitly on the integer $j^\prime$, in the sense that the integer $j^\prime$ determines which measurement outcomes $k$ are contained in the respective set $\mathcal{K}_2^{j^\prime}$. We assume no prior knowledge about $r$ and therefore no knowledge about the fraction $\tfrac{j^\prime}{r}$. Therefore, the integers $k\in \mathcal{K}_2^{j^\prime}$, cannot be known prior to the experiment and we have to assume that the interval of integers defined by $\mathcal{K}_2^{j^\prime}$ is distributed  uniformly across the range $0\leq k^\prime<q$. Hence, the optimal initial population distribution is uniform, i.e., $ q_{k_l}(S_1^{(l)}[\Theta_l],S_2^{(l)}[\Lambda_l];\sigma_l, \mathbb{P} )=\tfrac{1}{2} \, \forall k_l$. Recall that we assume that all $\Lambda_l$ are unital and all $S_2^{(l)}$ are unitality-preserving. Thus, we can construct super-operations $S_1^{(l)},S_2^{(l)}$ that can achieves this uniform distribution without influencing the second term involving the measures, see for example the super-operations utilized in the proof of the lower bound. Therefore, we find the upper bound
	\begin{align}
		\tilde{P}_j^{\suc}(\Theta_l,\Lambda_l)&\leq \varphi(r)\sup_{\substack{S_1 \in \CIS \\ S_2\in \DIS}} \max_{\tilde \sigma \in \mathcal{I}}  \max_{j^\prime \perp r}\sum_{k \in \mathcal{K}_2^{j^\prime}} \frac{1}{2^L}\prod_{l=1}^{L} \left( 1 + C(S_1^{(l)}[\Theta_l](\sigma_l)) \tilde{M}_\diamond(S_2^{(l)}[\Lambda_l])\right)\nonumber\\
		&= \frac{\varphi(r)}{2^L}\prod_{l=1}^{L}\left( 1 + \mathscr{C}(\Theta_l) \tilde{M}_\diamond(\Lambda_l)\right)  \max_{j^\prime \perp r}\sum_{k \in \mathcal{K}_2^{j^\prime}} 1 \nonumber\\
		&= \frac{\varphi(r)}{2^L}\prod_{l=1}^{L}\left( 1 + \mathscr{C}(\Theta_l) \tilde{M}_\diamond(\Lambda_l)\right)  \max_{j^\prime \perp r} |\mathcal{K}_2^{j^\prime}|\nonumber\\
		&\leq \frac{\varphi(r)}{2^L}\prod_{l=1}^{L}\left( 1 + \mathscr{C}(\Theta_l) \tilde{M}_\diamond(\Lambda_l)\right) \left(1+2\floor{\tfrac{2^L}{r^2}} \right),
	\end{align}
	where we used the results from the proof concerning the classical limit in the last line. So far we only used a single rotation. Note that since this bound holds for an arbitrary rotation $\mathcal{E}_j$ and since it is independent of the $j$ labeling a rotation $\mathcal{E}_j$, which is applied probabilistically, we find
	\begin{align}
		P^{\suc}(\Theta_l,\Lambda_l)&=\frac{1}{r}\sum_{j=0}^{r-1} \tilde{P}_j^{\suc}(\Theta_l,\Lambda_l)\nonumber\\
		&\leq \frac{1}{r}\sum_{j=0}^{r-1} \frac{\varphi(r)}{2^L}\prod_{l=1}^{L}\left( 1 + \mathscr{C}(\Theta_l) \tilde{M}_\diamond(\Lambda_l)\right) \left(1+2\floor{\tfrac{2^L}{r^2}} \right) \nonumber\\
		&=\frac{\varphi(r)}{2^L}\prod_{l=1}^{L}\left( 1 + \mathscr{C}(\Theta_l) \tilde{M}_\diamond(\Lambda_l)\right) \left(1+2\floor{\tfrac{2^L}{r^2}} \right).
	\end{align}
	Lastly, note that this bound can exceed unit probability and for this reason, we decide to formulate a bound of the form
	\begin{align}
		P^\suc(\Theta_l,\Lambda_l)\leq \min \left\lbrace \frac{\varphi(r)}{2^L}\left(1+2\floor{\tfrac{2^L}{r^2}} \right)\prod_{l=1}^{L}\left( 1 + \mathscr{C}(\Theta_l) \tilde{M}_\diamond(\Lambda_l)\right)  ,\, 1 \right\rbrace,
	\end{align}
	which concludes the proof.
\end{proof}
Again, note that for identical operations in each block we obtain the important special case of
\begin{align}
	P^\suc(\Theta_l,\Lambda_l)\leq \min \left\lbrace \frac{\varphi(r)}{2^L}\left(1+2\floor{\tfrac{2^L}{r^2}} \right)\left( 1 + \mathscr{C}(\Theta_l) \tilde{M}_\diamond(\Lambda_l)\right)^L  ,\, 1 \right\rbrace.
\end{align}
If the bound involving the resources measures exceeds unity, the upper bound reduces to a trivial bound. However, we emphasize that this is not only a trivial bound on the success probability. Most importantly, the expression exceeds unity if the prefactor becomes large. Comparing it with the classical success probability, we see that the prefactors are the same. This leads us to the conclusion that our bound is relevant whenever the order-finding problem is hard in the classical limit. Then the bound involving the dynamical resource measures is indeed useful. In that sense, we can argue that coherence is the resource that provides an advantage whenever there is an actual advantage to grant.


\begin{thebibliography}{96}%
	\makeatletter
	\providecommand \@ifxundefined [1]{%
		\@ifx{#1\undefined}
	}%
	\providecommand \@ifnum [1]{%
		\ifnum #1\expandafter \@firstoftwo
		\else \expandafter \@secondoftwo
		\fi
	}%
	\providecommand \@ifx [1]{%
		\ifx #1\expandafter \@firstoftwo
		\else \expandafter \@secondoftwo
		\fi
	}%
	\providecommand \natexlab [1]{#1}%
	\providecommand \enquote  [1]{``#1''}%
	\providecommand \bibnamefont  [1]{#1}%
	\providecommand \bibfnamefont [1]{#1}%
	\providecommand \citenamefont [1]{#1}%
	\providecommand \href@noop [0]{\@secondoftwo}%
	\providecommand \href [0]{\begingroup \@sanitize@url \@href}%
	\providecommand \@href[1]{\@@startlink{#1}\@@href}%
	\providecommand \@@href[1]{\endgroup#1\@@endlink}%
	\providecommand \@sanitize@url [0]{\catcode `\\12\catcode `\$12\catcode
		`\&12\catcode `\#12\catcode `\^12\catcode `\_12\catcode `\%12\relax}%
	\providecommand \@@startlink[1]{}%
	\providecommand \@@endlink[0]{}%
	\providecommand \url  [0]{\begingroup\@sanitize@url \@url }%
	\providecommand \@url [1]{\endgroup\@href {#1}{\urlprefix }}%
	\providecommand \urlprefix  [0]{URL }%
	\providecommand \Eprint [0]{\href }%
	\providecommand \doibase [0]{https://doi.org/}%
	\providecommand \selectlanguage [0]{\@gobble}%
	\providecommand \bibinfo  [0]{\@secondoftwo}%
	\providecommand \bibfield  [0]{\@secondoftwo}%
	\providecommand \translation [1]{[#1]}%
	\providecommand \BibitemOpen [0]{}%
	\providecommand \bibitemStop [0]{}%
	\providecommand \bibitemNoStop [0]{.\EOS\space}%
	\providecommand \EOS [0]{\spacefactor3000\relax}%
	\providecommand \BibitemShut  [1]{\csname bibitem#1\endcsname}%
	\let\auto@bib@innerbib\@empty
	\bibitem [{\citenamefont {Rivest}\ \emph {et~al.}(1978)\citenamefont {Rivest},
		\citenamefont {Shamir},\ and\ \citenamefont {Adleman}}]{RSA}%
	\BibitemOpen
	\bibfield  {author} {\bibinfo {author} {\bibfnamefont {R.~L.}\ \bibnamefont
			{Rivest}}, \bibinfo {author} {\bibfnamefont {A.}~\bibnamefont {Shamir}},\
		and\ \bibinfo {author} {\bibfnamefont {L.}~\bibnamefont {Adleman}},\
	}\bibfield  {title} {\bibinfo {title} {A method for obtaining digital
			signatures and public-key cryptosystems},\ }\href
	{https://doi.org/10.1145/359340.359342} {\bibfield  {journal} {\bibinfo
			{journal} {Commun. ACM}\ }\textbf {\bibinfo {volume} {21}},\ \bibinfo {pages}
		{120} (\bibinfo {year} {1978})}\BibitemShut {NoStop}%
	\bibitem [{\citenamefont {Shor}(1997)}]{Shor1995}%
	\BibitemOpen
	\bibfield  {author} {\bibinfo {author} {\bibfnamefont {P.~W.}\ \bibnamefont
			{Shor}},\ }\bibfield  {title} {\bibinfo {title} {Polynomial-time algorithms
			for prime factorization and discrete logarithms on a quantum computer},\
	}\href {https://doi.org/10.1137/S0097539795293172} {\bibfield  {journal}
		{\bibinfo  {journal} {SIAM J. Comput.}\ }\textbf {\bibinfo {volume} {26}},\
		\bibinfo {pages} {1484} (\bibinfo {year} {1997})}\BibitemShut {NoStop}%
	\bibitem [{\citenamefont {Deutsch}\ and\ \citenamefont
		{Jozsa}(1992)}]{Deutsch1992}%
	\BibitemOpen
	\bibfield  {author} {\bibinfo {author} {\bibfnamefont {D.}~\bibnamefont
			{Deutsch}}\ and\ \bibinfo {author} {\bibfnamefont {R.}~\bibnamefont
			{Jozsa}},\ }\bibfield  {title} {\bibinfo {title} {Rapid solution of problems
			by quantum computation},\ }\href {https://doi.org/10.1098/rspa.1992.0167}
	{\bibfield  {journal} {\bibinfo  {journal} {Proc. R. Soc. A}\ }\textbf
		{\bibinfo {volume} {439}},\ \bibinfo {pages} {553} (\bibinfo {year}
		{1992})}\BibitemShut {NoStop}%
	\bibitem [{\citenamefont {Simon}(1997)}]{Simon1997}%
	\BibitemOpen
	\bibfield  {author} {\bibinfo {author} {\bibfnamefont {D.~R.}\ \bibnamefont
			{Simon}},\ }\bibfield  {title} {\bibinfo {title} {On the power of quantum
			computation},\ }\href {https://doi.org/10.1137/S0097539796298637} {\bibfield
		{journal} {\bibinfo  {journal} {SIAM J. Comput.}\ }\textbf {\bibinfo {volume}
			{26}},\ \bibinfo {pages} {1474} (\bibinfo {year} {1997})}\BibitemShut
	{NoStop}%
	\bibitem [{\citenamefont {{Knill}}\ and\ \citenamefont
		{{Laflamme}}(1998)}]{Knill1998}%
	\BibitemOpen
	\bibfield  {author} {\bibinfo {author} {\bibfnamefont {E.}~\bibnamefont
			{{Knill}}}\ and\ \bibinfo {author} {\bibfnamefont {R.}~\bibnamefont
			{{Laflamme}}},\ }\bibfield  {title} {\bibinfo {title} {Power of one bit of
			quantum information},\ }\href {https://doi.org/10.1103/PhysRevLett.81.5672}
	{\bibfield  {journal} {\bibinfo  {journal} {\prl}\ }\textbf {\bibinfo
			{volume} {81}},\ \bibinfo {pages} {5672} (\bibinfo {year}
		{1998})}\BibitemShut {NoStop}%
	\bibitem [{\citenamefont {Harrow}\ \emph {et~al.}(2009)\citenamefont {Harrow},
		\citenamefont {Hassidim},\ and\ \citenamefont {Lloyd}}]{Harrow2009}%
	\BibitemOpen
	\bibfield  {author} {\bibinfo {author} {\bibfnamefont {A.~W.}\ \bibnamefont
			{Harrow}}, \bibinfo {author} {\bibfnamefont {A.}~\bibnamefont {Hassidim}},\
		and\ \bibinfo {author} {\bibfnamefont {S.}~\bibnamefont {Lloyd}},\ }\bibfield
	{title} {\bibinfo {title} {Quantum algorithm for linear systems of
			equations},\ }\href {https://doi.org/10.1103/PhysRevLett.103.150502}
	{\bibfield  {journal} {\bibinfo  {journal} {Phys. Rev. Lett.}\ }\textbf
		{\bibinfo {volume} {103}},\ \bibinfo {pages} {150502} (\bibinfo {year}
		{2009})}\BibitemShut {NoStop}%
	\bibitem [{\citenamefont {Jozsa}\ and\ \citenamefont
		{Linden}(2003)}]{Josza2003}%
	\BibitemOpen
	\bibfield  {author} {\bibinfo {author} {\bibfnamefont {R.}~\bibnamefont
			{Jozsa}}\ and\ \bibinfo {author} {\bibfnamefont {N.}~\bibnamefont {Linden}},\
	}\bibfield  {title} {\bibinfo {title} {On the role of entanglement in
			quantum-computational speed-up},\ }\href
	{http://doi.org/10.1098/rspa.2002.1097} {\bibfield  {journal} {\bibinfo
			{journal} {Proc. R. Soc. A}\ }\textbf {\bibinfo {volume} {459}},\ \bibinfo
		{pages} {20112032} (\bibinfo {year} {2003})}\BibitemShut {NoStop}%
	\bibitem [{\citenamefont {Hillery}(2016)}]{Hillery2016}%
	\BibitemOpen
	\bibfield  {author} {\bibinfo {author} {\bibfnamefont {M.}~\bibnamefont
			{Hillery}},\ }\bibfield  {title} {\bibinfo {title} {Coherence as a resource
			in decision problems: The {D}eutsch-{J}ozsa algorithm and a variation},\
	}\href {https://doi.org/10.1103/PhysRevA.93.012111} {\bibfield  {journal}
		{\bibinfo  {journal} {Phys. Rev. A}\ }\textbf {\bibinfo {volume} {93}},\
		\bibinfo {pages} {012111} (\bibinfo {year} {2016})}\BibitemShut {NoStop}%
	\bibitem [{\citenamefont {Anand}\ and\ \citenamefont {Pati}(2016)}]{Anand2016}%
	\BibitemOpen
	\bibfield  {author} {\bibinfo {author} {\bibfnamefont {N.}~\bibnamefont
			{Anand}}\ and\ \bibinfo {author} {\bibfnamefont {A.~K.}\ \bibnamefont
			{Pati}},\ }\bibfield  {title} {\bibinfo {title} {Coherence and entanglement
			monogamy in the discrete analogue of analog {G}rover search},\ }\href
	{https://arxiv.org/abs/1611.04542} {\bibfield  {journal} {\bibinfo  {journal}
			{arXiv:1611.04542}\ } (\bibinfo {year} {2016})}\BibitemShut {NoStop}%
	\bibitem [{\citenamefont {Shi}\ \emph {et~al.}(2017)\citenamefont {Shi},
		\citenamefont {Liu}, \citenamefont {Wang}, \citenamefont {Yang},
		\citenamefont {Yang},\ and\ \citenamefont {Fan}}]{Shi2017}%
	\BibitemOpen
	\bibfield  {author} {\bibinfo {author} {\bibfnamefont {H.-L.}\ \bibnamefont
			{Shi}}, \bibinfo {author} {\bibfnamefont {S.-Y.}\ \bibnamefont {Liu}},
		\bibinfo {author} {\bibfnamefont {X.-H.}\ \bibnamefont {Wang}}, \bibinfo
		{author} {\bibfnamefont {W.-L.}\ \bibnamefont {Yang}}, \bibinfo {author}
		{\bibfnamefont {Z.-Y.}\ \bibnamefont {Yang}},\ and\ \bibinfo {author}
		{\bibfnamefont {H.}~\bibnamefont {Fan}},\ }\bibfield  {title} {\bibinfo
		{title} {Coherence depletion in the {G}rover quantum search algorithm},\
	}\href {https://doi.org/10.1103/PhysRevA.95.032307} {\bibfield  {journal}
		{\bibinfo  {journal} {Phys. Rev. A}\ }\textbf {\bibinfo {volume} {95}},\
		\bibinfo {pages} {032307} (\bibinfo {year} {2017})}\BibitemShut {NoStop}%
	\bibitem [{\citenamefont {Liu}\ \emph {et~al.}(2019)\citenamefont {Liu},
		\citenamefont {Shang},\ and\ \citenamefont {Zhang}}]{Liu2019YC}%
	\BibitemOpen
	\bibfield  {author} {\bibinfo {author} {\bibfnamefont {Y.-C.}\ \bibnamefont
			{Liu}}, \bibinfo {author} {\bibfnamefont {J.}~\bibnamefont {Shang}},\ and\
		\bibinfo {author} {\bibfnamefont {X.}~\bibnamefont {Zhang}},\ }\bibfield
	{title} {\bibinfo {title} {Coherence depletion in quantum algorithms},\
	}\href {https://doi.org/10.3390/e21030260} {\bibfield  {journal} {\bibinfo
			{journal} {Entropy}\ }\textbf {\bibinfo {volume} {21}},\ \bibinfo {pages}
		{260} (\bibinfo {year} {2019})}\BibitemShut {NoStop}%
	\bibitem [{\citenamefont {Parker}\ and\ \citenamefont
		{Plenio}(2000)}]{Parker2000}%
	\BibitemOpen
	\bibfield  {author} {\bibinfo {author} {\bibfnamefont {S.}~\bibnamefont
			{Parker}}\ and\ \bibinfo {author} {\bibfnamefont {M.~B.}\ \bibnamefont
			{Plenio}},\ }\bibfield  {title} {\bibinfo {title} {Efficient factorization
			with a single pure qubit and $\mathrm{log}\mathit{N}$ mixed qubits},\ }\href
	{https://doi.org/10.1103/PhysRevLett.85.3049} {\bibfield  {journal} {\bibinfo
			{journal} {Phys. Rev. Lett.}\ }\textbf {\bibinfo {volume} {85}},\ \bibinfo
		{pages} {3049} (\bibinfo {year} {2000})}\BibitemShut {NoStop}%
	\bibitem [{\citenamefont {Streltsov}\ \emph {et~al.}(2017)\citenamefont
		{Streltsov}, \citenamefont {Adesso},\ and\ \citenamefont
		{Plenio}}]{Streltsov2017}%
	\BibitemOpen
	\bibfield  {author} {\bibinfo {author} {\bibfnamefont {A.}~\bibnamefont
			{Streltsov}}, \bibinfo {author} {\bibfnamefont {G.}~\bibnamefont {Adesso}},\
		and\ \bibinfo {author} {\bibfnamefont {M.~B.}\ \bibnamefont {Plenio}},\
	}\bibfield  {title} {\bibinfo {title} {Colloquium: Quantum coherence as a
			resource},\ }\href {https://doi.org/10.1103/RevModPhys.89.041003} {\bibfield
		{journal} {\bibinfo  {journal} {Rev. Mod. Phys.}\ }\textbf {\bibinfo {volume}
			{89}},\ \bibinfo {pages} {041003} (\bibinfo {year} {2017})}\BibitemShut
	{NoStop}%
	\bibitem [{\citenamefont {Vedral}\ \emph {et~al.}(1997)\citenamefont {Vedral},
		\citenamefont {Plenio}, \citenamefont {Rippin},\ and\ \citenamefont
		{Knight}}]{Vedral1997}%
	\BibitemOpen
	\bibfield  {author} {\bibinfo {author} {\bibfnamefont {V.}~\bibnamefont
			{Vedral}}, \bibinfo {author} {\bibfnamefont {M.~B.}\ \bibnamefont {Plenio}},
		\bibinfo {author} {\bibfnamefont {M.~A.}\ \bibnamefont {Rippin}},\ and\
		\bibinfo {author} {\bibfnamefont {P.~L.}\ \bibnamefont {Knight}},\ }\bibfield
	{title} {\bibinfo {title} {Quantifying entanglement},\ }\href
	{https://doi.org/10.1103/PhysRevLett.78.2275} {\bibfield  {journal} {\bibinfo
			{journal} {Phys. Rev. Lett.}\ }\textbf {\bibinfo {volume} {78}},\ \bibinfo
		{pages} {2275} (\bibinfo {year} {1997})}\BibitemShut {NoStop}%
	\bibitem [{\citenamefont {Vedral}\ and\ \citenamefont
		{Plenio}(1998)}]{Vedral998}%
	\BibitemOpen
	\bibfield  {author} {\bibinfo {author} {\bibfnamefont {V.}~\bibnamefont
			{Vedral}}\ and\ \bibinfo {author} {\bibfnamefont {M.~B.}\ \bibnamefont
			{Plenio}},\ }\bibfield  {title} {\bibinfo {title} {Entanglement measures and
			purification procedures},\ }\href {https://doi.org/10.1103/PhysRevA.57.1619}
	{\bibfield  {journal} {\bibinfo  {journal} {Phys. Rev. A}\ }\textbf {\bibinfo
			{volume} {57}},\ \bibinfo {pages} {1619} (\bibinfo {year}
		{1998})}\BibitemShut {NoStop}%
	\bibitem [{\citenamefont {Horodecki}\ \emph {et~al.}(2003)\citenamefont
		{Horodecki}, \citenamefont {Horodecki},\ and\ \citenamefont
		{Oppenheim}}]{Horodecki2003}%
	\BibitemOpen
	\bibfield  {author} {\bibinfo {author} {\bibfnamefont {M.}~\bibnamefont
			{Horodecki}}, \bibinfo {author} {\bibfnamefont {P.}~\bibnamefont
			{Horodecki}},\ and\ \bibinfo {author} {\bibfnamefont {J.}~\bibnamefont
			{Oppenheim}},\ }\bibfield  {title} {\bibinfo {title} {Reversible
			transformations from pure to mixed states and the unique measure of
			information},\ }\href {https://doi.org/10.1103/PhysRevA.67.062104} {\bibfield
		{journal} {\bibinfo  {journal} {Phys. Rev. A}\ }\textbf {\bibinfo {volume}
			{67}},\ \bibinfo {pages} {062104} (\bibinfo {year} {2003})}\BibitemShut
	{NoStop}%
	\bibitem [{\citenamefont {Aberg}(2006)}]{Aberg2006}%
	\BibitemOpen
	\bibfield  {author} {\bibinfo {author} {\bibfnamefont {J.}~\bibnamefont
			{Aberg}},\ }\bibfield  {title} {\bibinfo {title} {Quantifying
			superposition},\ }\href {http://arxiv.org/abs/quant-ph/0612146} {\bibfield
		{journal} {\bibinfo  {journal} {arXiv:quant-ph/0612146}\ } (\bibinfo {year}
		{2006})}\BibitemShut {NoStop}%
	\bibitem [{\citenamefont {Gour}\ and\ \citenamefont
		{Spekkens}(2008{\natexlab{a}})}]{Gour2008}%
	\BibitemOpen
	\bibfield  {author} {\bibinfo {author} {\bibfnamefont {G.}~\bibnamefont
			{Gour}}\ and\ \bibinfo {author} {\bibfnamefont {R.~W.}\ \bibnamefont
			{Spekkens}},\ }\bibfield  {title} {\bibinfo {title} {The resource theory of
			quantum reference frames: manipulations and monotones},\ }\href
	{https://doi.org/10.1088/1367-2630/10/3/033023} {\bibfield  {journal}
		{\bibinfo  {journal} {New J. Phys.}\ }\textbf {\bibinfo {volume} {10}},\
		\bibinfo {pages} {033023} (\bibinfo {year} {2008}{\natexlab{a}})}\BibitemShut
	{NoStop}%
	\bibitem [{\citenamefont {Horodecki}\ and\ \citenamefont
		{Oppenheim}(2013)}]{Horodecki2013}%
	\BibitemOpen
	\bibfield  {author} {\bibinfo {author} {\bibfnamefont {M.}~\bibnamefont
			{Horodecki}}\ and\ \bibinfo {author} {\bibfnamefont {J.}~\bibnamefont
			{Oppenheim}},\ }\bibfield  {title} {\bibinfo {title} {(quantumness in the
			context of) resource theories},\ }\href
	{https://doi.org/10.1142/S0217979213450197} {\bibfield  {journal} {\bibinfo
			{journal} {Int. J. Mod. Phys. B}\ }\textbf {\bibinfo {volume} {27}},\
		\bibinfo {pages} {1345019} (\bibinfo {year} {2013})}\BibitemShut {NoStop}%
	\bibitem [{\citenamefont {Brand\~ao}\ \emph {et~al.}(2013)\citenamefont
		{Brand\~ao}, \citenamefont {Horodecki}, \citenamefont {Oppenheim},
		\citenamefont {Renes},\ and\ \citenamefont {Spekkens}}]{Brandao2013}%
	\BibitemOpen
	\bibfield  {author} {\bibinfo {author} {\bibfnamefont {F.~G. S.~L.}\
			\bibnamefont {Brand\~ao}}, \bibinfo {author} {\bibfnamefont {M.}~\bibnamefont
			{Horodecki}}, \bibinfo {author} {\bibfnamefont {J.}~\bibnamefont
			{Oppenheim}}, \bibinfo {author} {\bibfnamefont {J.~M.}\ \bibnamefont
			{Renes}},\ and\ \bibinfo {author} {\bibfnamefont {R.~W.}\ \bibnamefont
			{Spekkens}},\ }\bibfield  {title} {\bibinfo {title} {Resource theory of
			quantum states out of thermal equilibrium},\ }\href
	{https://doi.org/10.1103/PhysRevLett.111.250404} {\bibfield  {journal}
		{\bibinfo  {journal} {Phys. Rev. Lett.}\ }\textbf {\bibinfo {volume} {111}},\
		\bibinfo {pages} {250404} (\bibinfo {year} {2013})}\BibitemShut {NoStop}%
	\bibitem [{\citenamefont {Veitch}\ \emph {et~al.}(2014)\citenamefont {Veitch},
		\citenamefont {Mousavian}, \citenamefont {Gottesman},\ and\ \citenamefont
		{Emerson}}]{Veitch2014}%
	\BibitemOpen
	\bibfield  {author} {\bibinfo {author} {\bibfnamefont {V.}~\bibnamefont
			{Veitch}}, \bibinfo {author} {\bibfnamefont {S.~A.~H.}\ \bibnamefont
			{Mousavian}}, \bibinfo {author} {\bibfnamefont {D.}~\bibnamefont
			{Gottesman}},\ and\ \bibinfo {author} {\bibfnamefont {J.}~\bibnamefont
			{Emerson}},\ }\bibfield  {title} {\bibinfo {title} {The resource theory of
			stabilizer quantum computation},\ }\href
	{https://doi.org/10.1088/1367-2630/16/1/013009} {\bibfield  {journal}
		{\bibinfo  {journal} {New J. Phys.}\ }\textbf {\bibinfo {volume} {16}},\
		\bibinfo {pages} {013009} (\bibinfo {year} {2014})}\BibitemShut {NoStop}%
	\bibitem [{\citenamefont {Grudka}\ \emph {et~al.}(2014)\citenamefont {Grudka},
		\citenamefont {Horodecki}, \citenamefont {Horodecki}, \citenamefont
		{Horodecki}, \citenamefont {Horodecki}, \citenamefont {Joshi}, \citenamefont
		{K\l{}obus},\ and\ \citenamefont {W\'ojcik}}]{Grudka2014}%
	\BibitemOpen
	\bibfield  {author} {\bibinfo {author} {\bibfnamefont {A.}~\bibnamefont
			{Grudka}}, \bibinfo {author} {\bibfnamefont {K.}~\bibnamefont {Horodecki}},
		\bibinfo {author} {\bibfnamefont {M.}~\bibnamefont {Horodecki}}, \bibinfo
		{author} {\bibfnamefont {P.}~\bibnamefont {Horodecki}}, \bibinfo {author}
		{\bibfnamefont {R.}~\bibnamefont {Horodecki}}, \bibinfo {author}
		{\bibfnamefont {P.}~\bibnamefont {Joshi}}, \bibinfo {author} {\bibfnamefont
			{W.}~\bibnamefont {K\l{}obus}},\ and\ \bibinfo {author} {\bibfnamefont
			{A.}~\bibnamefont {W\'ojcik}},\ }\bibfield  {title} {\bibinfo {title}
		{Quantifying contextuality},\ }\href
	{https://doi.org/10.1103/PhysRevLett.112.120401} {\bibfield  {journal}
		{\bibinfo  {journal} {Phys. Rev. Lett.}\ }\textbf {\bibinfo {volume} {112}},\
		\bibinfo {pages} {120401} (\bibinfo {year} {2014})}\BibitemShut {NoStop}%
	\bibitem [{\citenamefont {Baumgratz}\ \emph {et~al.}(2014)\citenamefont
		{Baumgratz}, \citenamefont {Cramer},\ and\ \citenamefont
		{Plenio}}]{Baumgartz2014}%
	\BibitemOpen
	\bibfield  {author} {\bibinfo {author} {\bibfnamefont {T.}~\bibnamefont
			{Baumgratz}}, \bibinfo {author} {\bibfnamefont {M.}~\bibnamefont {Cramer}},\
		and\ \bibinfo {author} {\bibfnamefont {M.~B.}\ \bibnamefont {Plenio}},\
	}\bibfield  {title} {\bibinfo {title} {Quantifying coherence},\ }\href
	{https://doi.org/10.1103/PhysRevLett.113.140401} {\bibfield  {journal}
		{\bibinfo  {journal} {Phys. Rev. Lett.}\ }\textbf {\bibinfo {volume} {113}},\
		\bibinfo {pages} {140401} (\bibinfo {year} {2014})}\BibitemShut {NoStop}%
	\bibitem [{\citenamefont {del Rio}\ \emph {et~al.}(2015)\citenamefont {del
			Rio}, \citenamefont {Kraemer},\ and\ \citenamefont {Renner}}]{delRio2015}%
	\BibitemOpen
	\bibfield  {author} {\bibinfo {author} {\bibfnamefont {L.}~\bibnamefont {del
				Rio}}, \bibinfo {author} {\bibfnamefont {L.}~\bibnamefont {Kraemer}},\ and\
		\bibinfo {author} {\bibfnamefont {R.}~\bibnamefont {Renner}},\ }\bibfield
	{title} {\bibinfo {title} {Resource theories of knowledge},\ }\href
	{https://arxiv.org/abs/1511.08818} {\bibfield  {journal} {\bibinfo  {journal}
			{arXiv:1511.08818}\ } (\bibinfo {year} {2015})}\BibitemShut {NoStop}%
	\bibitem [{\citenamefont {Chitambar}\ and\ \citenamefont
		{Gour}(2019)}]{Chitamber2019}%
	\BibitemOpen
	\bibfield  {author} {\bibinfo {author} {\bibfnamefont {E.}~\bibnamefont
			{Chitambar}}\ and\ \bibinfo {author} {\bibfnamefont {G.}~\bibnamefont
			{Gour}},\ }\bibfield  {title} {\bibinfo {title} {Quantum resource theories},\
	}\href {https://doi.org/10.1103/RevModPhys.91.025001} {\bibfield  {journal}
		{\bibinfo  {journal} {Rev. Mod. Phys.}\ }\textbf {\bibinfo {volume} {91}},\
		\bibinfo {pages} {025001} (\bibinfo {year} {2019})}\BibitemShut {NoStop}%
	\bibitem [{\citenamefont {Napoli}\ \emph {et~al.}(2016)\citenamefont {Napoli},
		\citenamefont {Bromley}, \citenamefont {Cianciaruso}, \citenamefont {Piani},
		\citenamefont {Johnston},\ and\ \citenamefont {Adesso}}]{Napoli2016}%
	\BibitemOpen
	\bibfield  {author} {\bibinfo {author} {\bibfnamefont {C.}~\bibnamefont
			{Napoli}}, \bibinfo {author} {\bibfnamefont {T.~R.}\ \bibnamefont {Bromley}},
		\bibinfo {author} {\bibfnamefont {M.}~\bibnamefont {Cianciaruso}}, \bibinfo
		{author} {\bibfnamefont {M.}~\bibnamefont {Piani}}, \bibinfo {author}
		{\bibfnamefont {N.}~\bibnamefont {Johnston}},\ and\ \bibinfo {author}
		{\bibfnamefont {G.}~\bibnamefont {Adesso}},\ }\bibfield  {title} {\bibinfo
		{title} {Robustness of coherence: An operational and observable measure of
			quantum coherence},\ }\href {https://doi.org/10.1103/PhysRevLett.116.150502}
	{\bibfield  {journal} {\bibinfo  {journal} {Phys. Rev. Lett.}\ }\textbf
		{\bibinfo {volume} {116}},\ \bibinfo {pages} {150502} (\bibinfo {year}
		{2016})}\BibitemShut {NoStop}%
	\bibitem [{\citenamefont {Piani}\ \emph {et~al.}(2016)\citenamefont {Piani},
		\citenamefont {Cianciaruso}, \citenamefont {Bromley}, \citenamefont {Napoli},
		\citenamefont {Johnston},\ and\ \citenamefont {Adesso}}]{Piani2016}%
	\BibitemOpen
	\bibfield  {author} {\bibinfo {author} {\bibfnamefont {M.}~\bibnamefont
			{Piani}}, \bibinfo {author} {\bibfnamefont {M.}~\bibnamefont {Cianciaruso}},
		\bibinfo {author} {\bibfnamefont {T.~R.}\ \bibnamefont {Bromley}}, \bibinfo
		{author} {\bibfnamefont {C.}~\bibnamefont {Napoli}}, \bibinfo {author}
		{\bibfnamefont {N.}~\bibnamefont {Johnston}},\ and\ \bibinfo {author}
		{\bibfnamefont {G.}~\bibnamefont {Adesso}},\ }\bibfield  {title} {\bibinfo
		{title} {Robustness of asymmetry and coherence of quantum states},\ }\href
	{https://doi.org/10.1103/PhysRevA.93.042107} {\bibfield  {journal} {\bibinfo
			{journal} {Phys. Rev. A}\ }\textbf {\bibinfo {volume} {93}},\ \bibinfo
		{pages} {042107} (\bibinfo {year} {2016})}\BibitemShut {NoStop}%
	\bibitem [{\citenamefont {Takagi}\ \emph {et~al.}(2019)\citenamefont {Takagi},
		\citenamefont {Regula}, \citenamefont {Bu}, \citenamefont {Liu},\ and\
		\citenamefont {Adesso}}]{Takagi2019b}%
	\BibitemOpen
	\bibfield  {author} {\bibinfo {author} {\bibfnamefont {R.}~\bibnamefont
			{Takagi}}, \bibinfo {author} {\bibfnamefont {B.}~\bibnamefont {Regula}},
		\bibinfo {author} {\bibfnamefont {K.}~\bibnamefont {Bu}}, \bibinfo {author}
		{\bibfnamefont {Z.-W.}\ \bibnamefont {Liu}},\ and\ \bibinfo {author}
		{\bibfnamefont {G.}~\bibnamefont {Adesso}},\ }\bibfield  {title} {\bibinfo
		{title} {Operational advantage of quantum resources in subchannel
			discrimination},\ }\href {https://doi.org/10.1103/PhysRevLett.122.140402}
	{\bibfield  {journal} {\bibinfo  {journal} {Phys. Rev. Lett.}\ }\textbf
		{\bibinfo {volume} {122}},\ \bibinfo {pages} {140402} (\bibinfo {year}
		{2019})}\BibitemShut {NoStop}%
	\bibitem [{\citenamefont {Takagi}\ and\ \citenamefont
		{Regula}(2019)}]{Takagi2019}%
	\BibitemOpen
	\bibfield  {author} {\bibinfo {author} {\bibfnamefont {R.}~\bibnamefont
			{Takagi}}\ and\ \bibinfo {author} {\bibfnamefont {B.}~\bibnamefont
			{Regula}},\ }\bibfield  {title} {\bibinfo {title} {General resource theories
			in quantum mechanics and beyond: Operational characterization via
			discrimination tasks},\ }\href {https://doi.org/10.1103/PhysRevX.9.031053}
	{\bibfield  {journal} {\bibinfo  {journal} {Phys. Rev. X}\ }\textbf {\bibinfo
			{volume} {9}},\ \bibinfo {pages} {031053} (\bibinfo {year}
		{2019})}\BibitemShut {NoStop}%
	\bibitem [{\citenamefont {Skrzypczyk}\ and\ \citenamefont
		{Linden}(2019)}]{Skrzypczyk2019a}%
	\BibitemOpen
	\bibfield  {author} {\bibinfo {author} {\bibfnamefont {P.}~\bibnamefont
			{Skrzypczyk}}\ and\ \bibinfo {author} {\bibfnamefont {N.}~\bibnamefont
			{Linden}},\ }\bibfield  {title} {\bibinfo {title} {Robustness of measurement,
			discrimination games, and accessible information},\ }\href
	{https://doi.org/10.1103/PhysRevLett.122.140403} {\bibfield  {journal}
		{\bibinfo  {journal} {Phys. Rev. Lett.}\ }\textbf {\bibinfo {volume} {122}},\
		\bibinfo {pages} {140403} (\bibinfo {year} {2019})}\BibitemShut {NoStop}%
	\bibitem [{\citenamefont {Skrzypczyk}\ \emph {et~al.}(2019)\citenamefont
		{Skrzypczyk}, \citenamefont {\ifmmode \check{S}\else
			\v{S}\fi{}upi\ifmmode~\acute{c}\else \'{c}\fi{}},\ and\ \citenamefont
		{Cavalcanti}}]{Skrzypczyk2019b}%
	\BibitemOpen
	\bibfield  {author} {\bibinfo {author} {\bibfnamefont {P.}~\bibnamefont
			{Skrzypczyk}}, \bibinfo {author} {\bibfnamefont {I.}~\bibnamefont {\ifmmode
				\check{S}\else \v{S}\fi{}upi\ifmmode~\acute{c}\else \'{c}\fi{}}},\ and\
		\bibinfo {author} {\bibfnamefont {D.}~\bibnamefont {Cavalcanti}},\ }\bibfield
	{title} {\bibinfo {title} {All sets of incompatible measurements give an
			advantage in quantum state discrimination},\ }\href
	{https://doi.org/10.1103/PhysRevLett.122.130403} {\bibfield  {journal}
		{\bibinfo  {journal} {Phys. Rev. Lett.}\ }\textbf {\bibinfo {volume} {122}},\
		\bibinfo {pages} {130403} (\bibinfo {year} {2019})}\BibitemShut {NoStop}%
	\bibitem [{\citenamefont {Uola}\ \emph {et~al.}(2019)\citenamefont {Uola},
		\citenamefont {Kraft}, \citenamefont {Shang}, \citenamefont {Yu},\ and\
		\citenamefont {G\"uhne}}]{Uola2019}%
	\BibitemOpen
	\bibfield  {author} {\bibinfo {author} {\bibfnamefont {R.}~\bibnamefont
			{Uola}}, \bibinfo {author} {\bibfnamefont {T.}~\bibnamefont {Kraft}},
		\bibinfo {author} {\bibfnamefont {J.}~\bibnamefont {Shang}}, \bibinfo
		{author} {\bibfnamefont {X.-D.}\ \bibnamefont {Yu}},\ and\ \bibinfo {author}
		{\bibfnamefont {O.}~\bibnamefont {G\"uhne}},\ }\bibfield  {title} {\bibinfo
		{title} {Quantifying quantum resources with conic programming},\ }\href
	{https://doi.org/10.1103/PhysRevLett.122.130404} {\bibfield  {journal}
		{\bibinfo  {journal} {Phys. Rev. Lett.}\ }\textbf {\bibinfo {volume} {122}},\
		\bibinfo {pages} {130404} (\bibinfo {year} {2019})}\BibitemShut {NoStop}%
	\bibitem [{\citenamefont {Mori}(2020)}]{Mori2020}%
	\BibitemOpen
	\bibfield  {author} {\bibinfo {author} {\bibfnamefont {J.}~\bibnamefont
			{Mori}},\ }\bibfield  {title} {\bibinfo {title} {Operational characterization
			of incompatibility of quantum channels with quantum state discrimination},\
	}\href {https://doi.org/10.1103/PhysRevA.101.032331} {\bibfield  {journal}
		{\bibinfo  {journal} {Phys. Rev. A}\ }\textbf {\bibinfo {volume} {101}},\
		\bibinfo {pages} {032331} (\bibinfo {year} {2020})}\BibitemShut {NoStop}%
	\bibitem [{\citenamefont {Ducuara}\ and\ \citenamefont
		{Skrzypczyk}(2020)}]{Ducuara2020a}%
	\BibitemOpen
	\bibfield  {author} {\bibinfo {author} {\bibfnamefont {A.~F.}\ \bibnamefont
			{Ducuara}}\ and\ \bibinfo {author} {\bibfnamefont {P.}~\bibnamefont
			{Skrzypczyk}},\ }\bibfield  {title} {\bibinfo {title} {Operational
			interpretation of weight-based resource quantifiers in convex quantum
			resource theories},\ }\href {https://doi.org/10.1103/PhysRevLett.125.110401}
	{\bibfield  {journal} {\bibinfo  {journal} {Phys. Rev. Lett.}\ }\textbf
		{\bibinfo {volume} {125}},\ \bibinfo {pages} {110401} (\bibinfo {year}
		{2020})}\BibitemShut {NoStop}%
	\bibitem [{\citenamefont {Ducuara}\ \emph {et~al.}(2020)\citenamefont
		{Ducuara}, \citenamefont {Lipka-Bartosik},\ and\ \citenamefont
		{Skrzypczyk}}]{Ducuara2020b}%
	\BibitemOpen
	\bibfield  {author} {\bibinfo {author} {\bibfnamefont {A.~F.}\ \bibnamefont
			{Ducuara}}, \bibinfo {author} {\bibfnamefont {P.}~\bibnamefont
			{Lipka-Bartosik}},\ and\ \bibinfo {author} {\bibfnamefont {P.}~\bibnamefont
			{Skrzypczyk}},\ }\bibfield  {title} {\bibinfo {title} {Multiobject
			operational tasks for convex quantum resource theories of state-measurement
			pairs},\ }\href {https://doi.org/10.1103/PhysRevResearch.2.033374} {\bibfield
		{journal} {\bibinfo  {journal} {Phys. Rev. Research}\ }\textbf {\bibinfo
			{volume} {2}},\ \bibinfo {pages} {033374} (\bibinfo {year}
		{2020})}\BibitemShut {NoStop}%
	\bibitem [{\citenamefont {Uola}\ \emph {et~al.}(2020)\citenamefont {Uola},
		\citenamefont {Bullock}, \citenamefont {Kraft}, \citenamefont
		{Pellonp\"a\"a},\ and\ \citenamefont {Brunner}}]{Uola2020}%
	\BibitemOpen
	\bibfield  {author} {\bibinfo {author} {\bibfnamefont {R.}~\bibnamefont
			{Uola}}, \bibinfo {author} {\bibfnamefont {T.}~\bibnamefont {Bullock}},
		\bibinfo {author} {\bibfnamefont {T.}~\bibnamefont {Kraft}}, \bibinfo
		{author} {\bibfnamefont {J.-P.}\ \bibnamefont {Pellonp\"a\"a}},\ and\
		\bibinfo {author} {\bibfnamefont {N.}~\bibnamefont {Brunner}},\ }\bibfield
	{title} {\bibinfo {title} {All quantum resources provide an advantage in
			exclusion tasks},\ }\href {https://doi.org/10.1103/PhysRevLett.125.110402}
	{\bibfield  {journal} {\bibinfo  {journal} {Phys. Rev. Lett.}\ }\textbf
		{\bibinfo {volume} {125}},\ \bibinfo {pages} {110402} (\bibinfo {year}
		{2020})}\BibitemShut {NoStop}%
	\bibitem [{\citenamefont {Masini}\ \emph {et~al.}(2021)\citenamefont {Masini},
		\citenamefont {Theurer},\ and\ \citenamefont {Plenio}}]{Masini2021}%
	\BibitemOpen
	\bibfield  {author} {\bibinfo {author} {\bibfnamefont {M.}~\bibnamefont
			{Masini}}, \bibinfo {author} {\bibfnamefont {T.}~\bibnamefont {Theurer}},\
		and\ \bibinfo {author} {\bibfnamefont {M.~B.}\ \bibnamefont {Plenio}},\
	}\bibfield  {title} {\bibinfo {title} {Coherence of operations and
			interferometry},\ }\href {https://doi.org/10.1103/PhysRevA.103.042426}
	{\bibfield  {journal} {\bibinfo  {journal} {Phys. Rev. A}\ }\textbf {\bibinfo
			{volume} {103}},\ \bibinfo {pages} {042426} (\bibinfo {year}
		{2021})}\BibitemShut {NoStop}%
	\bibitem [{\citenamefont {Matera}\ \emph {et~al.}(2016)\citenamefont {Matera},
		\citenamefont {Egloff}, \citenamefont {Killoran},\ and\ \citenamefont
		{Plenio}}]{Matera2016}%
	\BibitemOpen
	\bibfield  {author} {\bibinfo {author} {\bibfnamefont {J.~M.}\ \bibnamefont
			{Matera}}, \bibinfo {author} {\bibfnamefont {D.}~\bibnamefont {Egloff}},
		\bibinfo {author} {\bibfnamefont {N.}~\bibnamefont {Killoran}},\ and\
		\bibinfo {author} {\bibfnamefont {M.~B.}\ \bibnamefont {Plenio}},\ }\bibfield
	{title} {\bibinfo {title} {Coherent control of quantum systems as a resource
			theory},\ }\href {https://doi.org/10.1088/2058-9565/1/1/01LT01} {\bibfield
		{journal} {\bibinfo  {journal} {Quantum Sci. Technol.}\ }\textbf {\bibinfo
			{volume} {1}},\ \bibinfo {pages} {01LT01} (\bibinfo {year}
		{2016})}\BibitemShut {NoStop}%
	\bibitem [{\citenamefont {Biswas}\ \emph {et~al.}(2017)\citenamefont {Biswas},
		\citenamefont {Garc{\'\i}a~D{\'\i}az},\ and\ \citenamefont
		{Winter}}]{Biswas2017}%
	\BibitemOpen
	\bibfield  {author} {\bibinfo {author} {\bibfnamefont {T.}~\bibnamefont
			{Biswas}}, \bibinfo {author} {\bibfnamefont {M.}~\bibnamefont
			{Garc{\'\i}a~D{\'\i}az}},\ and\ \bibinfo {author} {\bibfnamefont
			{A.}~\bibnamefont {Winter}},\ }\bibfield  {title} {\bibinfo {title}
		{Interferometric visibility and coherence},\ }\href
	{https://dx.doi.org/10.1098/rspa.2017.0170} {\bibfield  {journal} {\bibinfo
			{journal} {Proc. R. Soc. A}\ }\textbf {\bibinfo {volume} {473}},\ \bibinfo
		{pages} {20170170} (\bibinfo {year} {2017})}\BibitemShut {NoStop}%
	\bibitem [{\citenamefont {Marvian}(2020)}]{Marvian2020}%
	\BibitemOpen
	\bibfield  {author} {\bibinfo {author} {\bibfnamefont {I.}~\bibnamefont
			{Marvian}},\ }\bibfield  {title} {\bibinfo {title} {Coherence distillation
			machines are impossible in quantum thermodynamics},\ }\href
	{https://doi.org/10.1038/s41467-019-13846-3} {\bibfield  {journal} {\bibinfo
			{journal} {Nat. Commun.}\ }\textbf {\bibinfo {volume} {11}},\ \bibinfo
		{pages} {25} (\bibinfo {year} {2020})}\BibitemShut {NoStop}%
	\bibitem [{\citenamefont {Liu}\ \emph {et~al.}(2017)\citenamefont {Liu},
		\citenamefont {Hu},\ and\ \citenamefont {Lloyd}}]{Liu2017}%
	\BibitemOpen
	\bibfield  {author} {\bibinfo {author} {\bibfnamefont {Z.-W.}\ \bibnamefont
			{Liu}}, \bibinfo {author} {\bibfnamefont {X.}~\bibnamefont {Hu}},\ and\
		\bibinfo {author} {\bibfnamefont {S.}~\bibnamefont {Lloyd}},\ }\bibfield
	{title} {\bibinfo {title} {Resource destroying maps},\ }\href
	{https://doi.org/10.1103/PhysRevLett.118.060502} {\bibfield  {journal}
		{\bibinfo  {journal} {Phys. Rev. Lett.}\ }\textbf {\bibinfo {volume} {118}},\
		\bibinfo {pages} {060502} (\bibinfo {year} {2017})}\BibitemShut {NoStop}%
	\bibitem [{\citenamefont {Garc{\'{i}}a~D{\'{i}}az}\ \emph
		{et~al.}(2018)\citenamefont {Garc{\'{i}}a~D{\'{i}}az}, \citenamefont {Fang},
		\citenamefont {Wang}, \citenamefont {Rosati}, \citenamefont {Skotiniotis},
		\citenamefont {Calsamiglia},\ and\ \citenamefont {Winter}}]{Diaz2018}%
	\BibitemOpen
	\bibfield  {author} {\bibinfo {author} {\bibfnamefont {M.}~\bibnamefont
			{Garc{\'{i}}a~D{\'{i}}az}}, \bibinfo {author} {\bibfnamefont
			{K.}~\bibnamefont {Fang}}, \bibinfo {author} {\bibfnamefont {X.}~\bibnamefont
			{Wang}}, \bibinfo {author} {\bibfnamefont {M.}~\bibnamefont {Rosati}},
		\bibinfo {author} {\bibfnamefont {M.}~\bibnamefont {Skotiniotis}}, \bibinfo
		{author} {\bibfnamefont {J.}~\bibnamefont {Calsamiglia}},\ and\ \bibinfo
		{author} {\bibfnamefont {A.}~\bibnamefont {Winter}},\ }\bibfield  {title}
	{\bibinfo {title} {Using and reusing coherence to realize quantum
			processes},\ }\href {https://doi.org/10.22331/q-2018-10-19-100} {\bibfield
		{journal} {\bibinfo  {journal} {{Quantum}}\ }\textbf {\bibinfo {volume}
			{2}},\ \bibinfo {pages} {100} (\bibinfo {year} {2018})}\BibitemShut {NoStop}%
	\bibitem [{\citenamefont {Theurer}\ \emph {et~al.}(2019)\citenamefont
		{Theurer}, \citenamefont {Egloff}, \citenamefont {Zhang},\ and\ \citenamefont
		{Plenio}}]{Theurer2019}%
	\BibitemOpen
	\bibfield  {author} {\bibinfo {author} {\bibfnamefont {T.}~\bibnamefont
			{Theurer}}, \bibinfo {author} {\bibfnamefont {D.}~\bibnamefont {Egloff}},
		\bibinfo {author} {\bibfnamefont {L.}~\bibnamefont {Zhang}},\ and\ \bibinfo
		{author} {\bibfnamefont {M.~B.}\ \bibnamefont {Plenio}},\ }\bibfield  {title}
	{\bibinfo {title} {Quantifying operations with an application to coherence},\
	}\href {https://doi.org/10.1103/PhysRevLett.122.190405} {\bibfield  {journal}
		{\bibinfo  {journal} {Phys. Rev. Lett.}\ }\textbf {\bibinfo {volume} {122}},\
		\bibinfo {pages} {190405} (\bibinfo {year} {2019})}\BibitemShut {NoStop}%
	\bibitem [{\citenamefont {Yadin}\ \emph {et~al.}(2016)\citenamefont {Yadin},
		\citenamefont {Ma}, \citenamefont {Girolami}, \citenamefont {Gu},\ and\
		\citenamefont {Vedral}}]{Yadin2016}%
	\BibitemOpen
	\bibfield  {author} {\bibinfo {author} {\bibfnamefont {B.}~\bibnamefont
			{Yadin}}, \bibinfo {author} {\bibfnamefont {J.}~\bibnamefont {Ma}}, \bibinfo
		{author} {\bibfnamefont {D.}~\bibnamefont {Girolami}}, \bibinfo {author}
		{\bibfnamefont {M.}~\bibnamefont {Gu}},\ and\ \bibinfo {author}
		{\bibfnamefont {V.}~\bibnamefont {Vedral}},\ }\bibfield  {title} {\bibinfo
		{title} {Quantum processes which do not use coherence},\ }\href
	{https://doi.org/10.1103/PhysRevX.6.041028} {\bibfield  {journal} {\bibinfo
			{journal} {Phys. Rev. X}\ }\textbf {\bibinfo {volume} {6}},\ \bibinfo {pages}
		{041028} (\bibinfo {year} {2016})}\BibitemShut {NoStop}%
	\bibitem [{\citenamefont {Smirne}\ \emph {et~al.}(2018)\citenamefont {Smirne},
		\citenamefont {Egloff}, \citenamefont {D{\'{\i}}az}, \citenamefont {Plenio},\
		and\ \citenamefont {Huelga}}]{Smirne2018}%
	\BibitemOpen
	\bibfield  {author} {\bibinfo {author} {\bibfnamefont {A.}~\bibnamefont
			{Smirne}}, \bibinfo {author} {\bibfnamefont {D.}~\bibnamefont {Egloff}},
		\bibinfo {author} {\bibfnamefont {M.~G.}\ \bibnamefont {D{\'{\i}}az}},
		\bibinfo {author} {\bibfnamefont {M.~B.}\ \bibnamefont {Plenio}},\ and\
		\bibinfo {author} {\bibfnamefont {S.~F.}\ \bibnamefont {Huelga}},\ }\bibfield
	{title} {\bibinfo {title} {Coherence and non-classicality of quantum
			{M}arkov processes},\ }\href {https://doi.org/10.1088/2058-9565/aaebd5}
	{\bibfield  {journal} {\bibinfo  {journal} {Quantum Sci. Technol.}\ }\textbf
		{\bibinfo {volume} {4}},\ \bibinfo {pages} {01LT01} (\bibinfo {year}
		{2018})}\BibitemShut {NoStop}%
	\bibitem [{\citenamefont {Meznaric}\ \emph {et~al.}(2013)\citenamefont
		{Meznaric}, \citenamefont {Clark},\ and\ \citenamefont
		{Datta}}]{Meznaric2013}%
	\BibitemOpen
	\bibfield  {author} {\bibinfo {author} {\bibfnamefont {S.}~\bibnamefont
			{Meznaric}}, \bibinfo {author} {\bibfnamefont {S.~R.}\ \bibnamefont
			{Clark}},\ and\ \bibinfo {author} {\bibfnamefont {A.}~\bibnamefont {Datta}},\
	}\bibfield  {title} {\bibinfo {title} {Quantifying the nonclassicality of
			operations},\ }\href {https://doi.org/10.1103/PhysRevLett.110.070502}
	{\bibfield  {journal} {\bibinfo  {journal} {Phys. Rev. Lett.}\ }\textbf
		{\bibinfo {volume} {110}},\ \bibinfo {pages} {070502} (\bibinfo {year}
		{2013})}\BibitemShut {NoStop}%
	\bibitem [{\citenamefont {Xu}\ \emph {et~al.}(2020)\citenamefont {Xu},
		\citenamefont {Xu}, \citenamefont {Theurer}, \citenamefont {Egloff},
		\citenamefont {Liu}, \citenamefont {Yu}, \citenamefont {Plenio},\ and\
		\citenamefont {Zhang}}]{Xu2020}%
	\BibitemOpen
	\bibfield  {author} {\bibinfo {author} {\bibfnamefont {H.}~\bibnamefont
			{Xu}}, \bibinfo {author} {\bibfnamefont {F.}~\bibnamefont {Xu}}, \bibinfo
		{author} {\bibfnamefont {T.}~\bibnamefont {Theurer}}, \bibinfo {author}
		{\bibfnamefont {D.}~\bibnamefont {Egloff}}, \bibinfo {author} {\bibfnamefont
			{Z.-W.}\ \bibnamefont {Liu}}, \bibinfo {author} {\bibfnamefont
			{N.}~\bibnamefont {Yu}}, \bibinfo {author} {\bibfnamefont {M.~B.}\
			\bibnamefont {Plenio}},\ and\ \bibinfo {author} {\bibfnamefont
			{L.}~\bibnamefont {Zhang}},\ }\bibfield  {title} {\bibinfo {title}
		{Experimental quantification of coherence of a tunable quantum detector},\
	}\href {https://doi.org/10.1103/PhysRevLett.125.060404} {\bibfield  {journal}
		{\bibinfo  {journal} {Phys. Rev. Lett.}\ }\textbf {\bibinfo {volume} {125}},\
		\bibinfo {pages} {060404} (\bibinfo {year} {2020})}\BibitemShut {NoStop}%
	\bibitem [{\citenamefont {Ben~Dana}\ \emph {et~al.}(2017)\citenamefont
		{Ben~Dana}, \citenamefont {Garc\'{\i}a~D\'{\i}az}, \citenamefont {Mejatty},\
		and\ \citenamefont {Winter}}]{Dana2017}%
	\BibitemOpen
	\bibfield  {author} {\bibinfo {author} {\bibfnamefont {K.}~\bibnamefont
			{Ben~Dana}}, \bibinfo {author} {\bibfnamefont {M.}~\bibnamefont
			{Garc\'{\i}a~D\'{\i}az}}, \bibinfo {author} {\bibfnamefont {M.}~\bibnamefont
			{Mejatty}},\ and\ \bibinfo {author} {\bibfnamefont {A.}~\bibnamefont
			{Winter}},\ }\bibfield  {title} {\bibinfo {title} {Resource theory of
			coherence: Beyond states},\ }\href
	{https://doi.org/10.1103/PhysRevA.95.062327} {\bibfield  {journal} {\bibinfo
			{journal} {Phys. Rev. A}\ }\textbf {\bibinfo {volume} {95}},\ \bibinfo
		{pages} {062327} (\bibinfo {year} {2017})}\BibitemShut {NoStop}%
	\bibitem [{\citenamefont {Zhuang}\ \emph {et~al.}(2018)\citenamefont {Zhuang},
		\citenamefont {Shor},\ and\ \citenamefont {Shapiro}}]{Zhuang2018}%
	\BibitemOpen
	\bibfield  {author} {\bibinfo {author} {\bibfnamefont {Q.}~\bibnamefont
			{Zhuang}}, \bibinfo {author} {\bibfnamefont {P.~W.}\ \bibnamefont {Shor}},\
		and\ \bibinfo {author} {\bibfnamefont {J.~H.}\ \bibnamefont {Shapiro}},\
	}\bibfield  {title} {\bibinfo {title} {Resource theory of non-{G}aussian
			operations},\ }\href {https://doi.org/10.1103/PhysRevA.97.052317} {\bibfield
		{journal} {\bibinfo  {journal} {Phys. Rev. A}\ }\textbf {\bibinfo {volume}
			{97}},\ \bibinfo {pages} {052317} (\bibinfo {year} {2018})}\BibitemShut
	{NoStop}%
	\bibitem [{\citenamefont {Wang}\ \emph {et~al.}(2019)\citenamefont {Wang},
		\citenamefont {Wilde},\ and\ \citenamefont {Su}}]{Wang2019a}%
	\BibitemOpen
	\bibfield  {author} {\bibinfo {author} {\bibfnamefont {X.}~\bibnamefont
			{Wang}}, \bibinfo {author} {\bibfnamefont {M.~M.}\ \bibnamefont {Wilde}},\
		and\ \bibinfo {author} {\bibfnamefont {Y.}~\bibnamefont {Su}},\ }\bibfield
	{title} {\bibinfo {title} {Quantifying the magic of quantum channels},\
	}\href {https://doi.org/10.1088/1367-2630/ab451d} {\bibfield  {journal}
		{\bibinfo  {journal} {New J. Phys.}\ }\textbf {\bibinfo {volume} {21}},\
		\bibinfo {pages} {103002} (\bibinfo {year} {2019})}\BibitemShut {NoStop}%
	\bibitem [{\citenamefont {Wang}\ and\ \citenamefont {Wilde}(2019)}]{Wang2019b}%
	\BibitemOpen
	\bibfield  {author} {\bibinfo {author} {\bibfnamefont {X.}~\bibnamefont
			{Wang}}\ and\ \bibinfo {author} {\bibfnamefont {M.~M.}\ \bibnamefont
			{Wilde}},\ }\bibfield  {title} {\bibinfo {title} {Resource theory of
			asymmetric distinguishability for quantum channels},\ }\href
	{https://doi.org/10.1103/PhysRevResearch.1.033169} {\bibfield  {journal}
		{\bibinfo  {journal} {Phys. Rev. Research}\ }\textbf {\bibinfo {volume}
			{1}},\ \bibinfo {pages} {033169} (\bibinfo {year} {2019})}\BibitemShut
	{NoStop}%
	\bibitem [{\citenamefont {Liu}\ and\ \citenamefont {Yuan}(2020)}]{Liu2020}%
	\BibitemOpen
	\bibfield  {author} {\bibinfo {author} {\bibfnamefont {Y.}~\bibnamefont
			{Liu}}\ and\ \bibinfo {author} {\bibfnamefont {X.}~\bibnamefont {Yuan}},\
	}\bibfield  {title} {\bibinfo {title} {Operational resource theory of quantum
			channels},\ }\href {https://doi.org/10.1103/PhysRevResearch.2.012035}
	{\bibfield  {journal} {\bibinfo  {journal} {Phys. Rev. Research}\ }\textbf
		{\bibinfo {volume} {2}},\ \bibinfo {pages} {012035(R)} (\bibinfo {year}
		{2020})}\BibitemShut {NoStop}%
	\bibitem [{\citenamefont {Liu}\ and\ \citenamefont {Winter}(2019)}]{Liu2019}%
	\BibitemOpen
	\bibfield  {author} {\bibinfo {author} {\bibfnamefont {Z.-W.}\ \bibnamefont
			{Liu}}\ and\ \bibinfo {author} {\bibfnamefont {A.}~\bibnamefont {Winter}},\
	}\bibfield  {title} {\bibinfo {title} {Resource theories of quantum channels
			and the universal role of resource erasure},\ }\href
	{https://arxiv.org/abs/1904.04201} {\bibfield  {journal} {\bibinfo  {journal}
			{arXiv:1904.04201}\ } (\bibinfo {year} {2019})}\BibitemShut {NoStop}%
	\bibitem [{\citenamefont {Gour}\ and\ \citenamefont
		{Winter}(2019{\natexlab{a}})}]{Gour2019a}%
	\BibitemOpen
	\bibfield  {author} {\bibinfo {author} {\bibfnamefont {G.}~\bibnamefont
			{Gour}}\ and\ \bibinfo {author} {\bibfnamefont {A.}~\bibnamefont {Winter}},\
	}\bibfield  {title} {\bibinfo {title} {How to quantify a dynamical quantum
			resource},\ }\href {https://doi.org/10.1103/PhysRevLett.123.150401}
	{\bibfield  {journal} {\bibinfo  {journal} {Phys. Rev. Lett.}\ }\textbf
		{\bibinfo {volume} {123}},\ \bibinfo {pages} {150401} (\bibinfo {year}
		{2019}{\natexlab{a}})}\BibitemShut {NoStop}%
	\bibitem [{\citenamefont {Gour}\ and\ \citenamefont
		{Scandolo}(2021)}]{Gour2021}%
	\BibitemOpen
	\bibfield  {author} {\bibinfo {author} {\bibfnamefont {G.}~\bibnamefont
			{Gour}}\ and\ \bibinfo {author} {\bibfnamefont {C.~M.}\ \bibnamefont
			{Scandolo}},\ }\bibfield  {title} {\bibinfo {title} {Entanglement of a
			bipartite channel},\ }\href {https://doi.org/10.1103/PhysRevA.103.062422}
	{\bibfield  {journal} {\bibinfo  {journal} {Phys. Rev. A}\ }\textbf {\bibinfo
			{volume} {103}},\ \bibinfo {pages} {062422} (\bibinfo {year}
		{2021})}\BibitemShut {NoStop}%
	\bibitem [{\citenamefont {Saxena}\ \emph {et~al.}(2020)\citenamefont {Saxena},
		\citenamefont {Chitambar},\ and\ \citenamefont {Gour}}]{Saxena2020}%
	\BibitemOpen
	\bibfield  {author} {\bibinfo {author} {\bibfnamefont {G.}~\bibnamefont
			{Saxena}}, \bibinfo {author} {\bibfnamefont {E.}~\bibnamefont {Chitambar}},\
		and\ \bibinfo {author} {\bibfnamefont {G.}~\bibnamefont {Gour}},\ }\bibfield
	{title} {\bibinfo {title} {Dynamical resource theory of quantum coherence},\
	}\href {https://doi.org/10.1103/PhysRevResearch.2.023298} {\bibfield
		{journal} {\bibinfo  {journal} {Phys. Rev. Research}\ }\textbf {\bibinfo
			{volume} {2}},\ \bibinfo {pages} {023298} (\bibinfo {year}
		{2020})}\BibitemShut {NoStop}%
	\bibitem [{\citenamefont {Gour}\ and\ \citenamefont
		{Scandolo}(2020{\natexlab{a}})}]{Gour2020}%
	\BibitemOpen
	\bibfield  {author} {\bibinfo {author} {\bibfnamefont {G.}~\bibnamefont
			{Gour}}\ and\ \bibinfo {author} {\bibfnamefont {C.~M.}\ \bibnamefont
			{Scandolo}},\ }\bibfield  {title} {\bibinfo {title} {Dynamical
			entanglement},\ }\href {https://doi.org/10.1103/PhysRevLett.125.180505}
	{\bibfield  {journal} {\bibinfo  {journal} {Phys. Rev. Lett.}\ }\textbf
		{\bibinfo {volume} {125}},\ \bibinfo {pages} {180505} (\bibinfo {year}
		{2020}{\natexlab{a}})}\BibitemShut {NoStop}%
	\bibitem [{\citenamefont {B{\"a}uml}\ \emph {et~al.}(2019)\citenamefont
		{B{\"a}uml}, \citenamefont {Das}, \citenamefont {Wang},\ and\ \citenamefont
		{Wilde}}]{Bauml2019}%
	\BibitemOpen
	\bibfield  {author} {\bibinfo {author} {\bibfnamefont {S.}~\bibnamefont
			{B{\"a}uml}}, \bibinfo {author} {\bibfnamefont {S.}~\bibnamefont {Das}},
		\bibinfo {author} {\bibfnamefont {X.}~\bibnamefont {Wang}},\ and\ \bibinfo
		{author} {\bibfnamefont {M.~M.}\ \bibnamefont {Wilde}},\ }\bibfield  {title}
	{\bibinfo {title} {Resource theory of entanglement for bipartite quantum
			channels},\ }\href {https://arxiv.org/abs/1907.04181} {\bibfield  {journal}
		{\bibinfo  {journal} {arXiv:1907.04181}\ } (\bibinfo {year}
		{2019})}\BibitemShut {NoStop}%
	\bibitem [{\citenamefont {Li}\ \emph {et~al.}(2020)\citenamefont {Li},
		\citenamefont {Bu},\ and\ \citenamefont {Liu}}]{Li2020}%
	\BibitemOpen
	\bibfield  {author} {\bibinfo {author} {\bibfnamefont {L.}~\bibnamefont
			{Li}}, \bibinfo {author} {\bibfnamefont {K.}~\bibnamefont {Bu}},\ and\
		\bibinfo {author} {\bibfnamefont {Z.-W.}\ \bibnamefont {Liu}},\ }\bibfield
	{title} {\bibinfo {title} {Quantifying the resource content of quantum
			channels: An operational approach},\ }\href
	{https://doi.org/10.1103/PhysRevA.101.022335} {\bibfield  {journal} {\bibinfo
			{journal} {Phys. Rev. A}\ }\textbf {\bibinfo {volume} {101}},\ \bibinfo
		{pages} {022335} (\bibinfo {year} {2020})}\BibitemShut {NoStop}%
	\bibitem [{\citenamefont {Takagi}\ \emph {et~al.}(2020)\citenamefont {Takagi},
		\citenamefont {Wang},\ and\ \citenamefont {Hayashi}}]{Takagi2020}%
	\BibitemOpen
	\bibfield  {author} {\bibinfo {author} {\bibfnamefont {R.}~\bibnamefont
			{Takagi}}, \bibinfo {author} {\bibfnamefont {K.}~\bibnamefont {Wang}},\ and\
		\bibinfo {author} {\bibfnamefont {M.}~\bibnamefont {Hayashi}},\ }\bibfield
	{title} {\bibinfo {title} {Application of the resource theory of channels to
			communication scenarios},\ }\href
	{https://doi.org/10.1103/PhysRevLett.124.120502} {\bibfield  {journal}
		{\bibinfo  {journal} {Phys. Rev. Lett.}\ }\textbf {\bibinfo {volume} {124}},\
		\bibinfo {pages} {120502} (\bibinfo {year} {2020})}\BibitemShut {NoStop}%
	\bibitem [{\citenamefont {Takagi}(2021)}]{Takagi2021}%
	\BibitemOpen
	\bibfield  {author} {\bibinfo {author} {\bibfnamefont {R.}~\bibnamefont
			{Takagi}},\ }\bibfield  {title} {\bibinfo {title} {Optimal resource cost for
			error mitigation},\ }\href {https://doi.org/10.1103/PhysRevResearch.3.033178}
	{\bibfield  {journal} {\bibinfo  {journal} {Phys. Rev. Research}\ }\textbf
		{\bibinfo {volume} {3}},\ \bibinfo {pages} {033178} (\bibinfo {year}
		{2021})}\BibitemShut {NoStop}%
	\bibitem [{\citenamefont {Chiribella}\ \emph {et~al.}(2008)\citenamefont
		{Chiribella}, \citenamefont {D'Ariano},\ and\ \citenamefont
		{Perinotti}}]{Chiribella2008}%
	\BibitemOpen
	\bibfield  {author} {\bibinfo {author} {\bibfnamefont {G.}~\bibnamefont
			{Chiribella}}, \bibinfo {author} {\bibfnamefont {G.~M.}\ \bibnamefont
			{D'Ariano}},\ and\ \bibinfo {author} {\bibfnamefont {P.}~\bibnamefont
			{Perinotti}},\ }\bibfield  {title} {\bibinfo {title} {Transforming quantum
			operations: Quantum supermaps},\ }\href
	{https://doi.org/10.1209/0295-5075/83/30004} {\bibfield  {journal} {\bibinfo
			{journal} {Europhys. Lett.}\ }\textbf {\bibinfo {volume} {83}},\ \bibinfo
		{pages} {30004} (\bibinfo {year} {2008})}\BibitemShut {NoStop}%
	\bibitem [{\citenamefont {Gour}\ and\ \citenamefont
		{Winter}(2019{\natexlab{b}})}]{PhysRevLett.123.150401}%
	\BibitemOpen
	\bibfield  {author} {\bibinfo {author} {\bibfnamefont {G.}~\bibnamefont
			{Gour}}\ and\ \bibinfo {author} {\bibfnamefont {A.}~\bibnamefont {Winter}},\
	}\bibfield  {title} {\bibinfo {title} {How to quantify a dynamical quantum
			resource},\ }\href {https://doi.org/10.1103/PhysRevLett.123.150401}
	{\bibfield  {journal} {\bibinfo  {journal} {Phys. Rev. Lett.}\ }\textbf
		{\bibinfo {volume} {123}},\ \bibinfo {pages} {150401} (\bibinfo {year}
		{2019}{\natexlab{b}})}\BibitemShut {NoStop}%
	\bibitem [{\citenamefont {Gour}\ and\ \citenamefont
		{Scandolo}(2020{\natexlab{b}})}]{gour2020dynamical}%
	\BibitemOpen
	\bibfield  {author} {\bibinfo {author} {\bibfnamefont {G.}~\bibnamefont
			{Gour}}\ and\ \bibinfo {author} {\bibfnamefont {C.~M.}\ \bibnamefont
			{Scandolo}},\ }\bibfield  {title} {\bibinfo {title} {Dynamical resources},\
	}\href {https://arxiv.org/abs/2101.01552} {\bibfield  {journal} {\bibinfo
			{journal} {arXiv:2101.01552}\ } (\bibinfo {year}
		{2020}{\natexlab{b}})}\BibitemShut {NoStop}%
	\bibitem [{\citenamefont {Bennett}\ \emph {et~al.}(2003)\citenamefont
		{Bennett}, \citenamefont {Harrow}, \citenamefont {Leung},\ and\ \citenamefont
		{Smolin}}]{Bennett2003}%
	\BibitemOpen
	\bibfield  {author} {\bibinfo {author} {\bibfnamefont {C.~H.}\ \bibnamefont
			{Bennett}}, \bibinfo {author} {\bibfnamefont {A.~W.}\ \bibnamefont {Harrow}},
		\bibinfo {author} {\bibfnamefont {D.~W.}\ \bibnamefont {Leung}},\ and\
		\bibinfo {author} {\bibfnamefont {J.~A.}\ \bibnamefont {Smolin}},\ }\bibfield
	{title} {\bibinfo {title} {On the capacities of bipartite {H}amiltonians and
			unitary gates},\ }\href {https://doi.org/10.1109/TIT.2003.814935} {\bibfield
		{journal} {\bibinfo  {journal} {IEEE Trans. Inf. Theory}\ }\textbf {\bibinfo
			{volume} {49}},\ \bibinfo {pages} {1895} (\bibinfo {year}
		{2003})}\BibitemShut {NoStop}%
	\bibitem [{\citenamefont {Mani}\ and\ \citenamefont
		{Karimipour}(2015)}]{Mani2015}%
	\BibitemOpen
	\bibfield  {author} {\bibinfo {author} {\bibfnamefont {A.}~\bibnamefont
			{Mani}}\ and\ \bibinfo {author} {\bibfnamefont {V.}~\bibnamefont
			{Karimipour}},\ }\bibfield  {title} {\bibinfo {title} {Cohering and
			decohering power of quantum channels},\ }\href
	{https://doi.org/10.1103/PhysRevA.92.032331} {\bibfield  {journal} {\bibinfo
			{journal} {Phys. Rev. A}\ }\textbf {\bibinfo {volume} {92}},\ \bibinfo
		{pages} {032331} (\bibinfo {year} {2015})}\BibitemShut {NoStop}%
	\bibitem [{\citenamefont {Xi}\ \emph {et~al.}(2015)\citenamefont {Xi},
		\citenamefont {Hu}, \citenamefont {Li},\ and\ \citenamefont {Fan}}]{Xi2015}%
	\BibitemOpen
	\bibfield  {author} {\bibinfo {author} {\bibfnamefont {Z.}~\bibnamefont
			{Xi}}, \bibinfo {author} {\bibfnamefont {M.}~\bibnamefont {Hu}}, \bibinfo
		{author} {\bibfnamefont {Y.}~\bibnamefont {Li}},\ and\ \bibinfo {author}
		{\bibfnamefont {H.}~\bibnamefont {Fan}},\ }\bibfield  {title} {\bibinfo
		{title} {Entropic characterization of coherence in quantum evolutions},\
	}\href {https://arxiv.org/abs/1510.06473} {\bibfield  {journal} {\bibinfo
			{journal} {arXiv:1510.06473}\ } (\bibinfo {year} {2015})}\BibitemShut
	{NoStop}%
	\bibitem [{\citenamefont {Garc\'{\i}a-D\'{\i}az}\ \emph
		{et~al.}(2016)\citenamefont {Garc\'{\i}a-D\'{\i}az}, \citenamefont {Egloff},\
		and\ \citenamefont {Plenio}}]{Garcia2016}%
	\BibitemOpen
	\bibfield  {author} {\bibinfo {author} {\bibfnamefont {M.}~\bibnamefont
			{Garc\'{\i}a-D\'{\i}az}}, \bibinfo {author} {\bibfnamefont {D.}~\bibnamefont
			{Egloff}},\ and\ \bibinfo {author} {\bibfnamefont {M.~B.}\ \bibnamefont
			{Plenio}},\ }\bibfield  {title} {\bibinfo {title} {A note on coherence power
			of n-dimensional unitary operators},\ }\href
	{https://dl.acm.org/doi/abs/10.5555/3179439.3179441} {\bibfield  {journal}
		{\bibinfo  {journal} {Quantum Inf. Comput.}\ }\textbf {\bibinfo {volume}
			{16}},\ \bibinfo {pages} {12821294} (\bibinfo {year} {2016})}\BibitemShut
	{NoStop}%
	\bibitem [{\citenamefont {Bu}\ and\ \citenamefont {Xiong}(2017)}]{Bu2017a}%
	\BibitemOpen
	\bibfield  {author} {\bibinfo {author} {\bibfnamefont {K.}~\bibnamefont
			{Bu}}\ and\ \bibinfo {author} {\bibfnamefont {C.}~\bibnamefont {Xiong}},\
	}\bibfield  {title} {\bibinfo {title} {A note on cohering power and
			de-cohering power},\ }\href {https://doi.org/10.26421/QIC17.13-14-8}
	{\bibfield  {journal} {\bibinfo  {journal} {Quantum Inf. Comput.}\ }\textbf
		{\bibinfo {volume} {17}},\ \bibinfo {pages} {1206} (\bibinfo {year}
		{2017})}\BibitemShut {NoStop}%
	\bibitem [{\citenamefont {Bu}\ \emph {et~al.}(2017)\citenamefont {Bu},
		\citenamefont {Kumar}, \citenamefont {Zhang},\ and\ \citenamefont
		{Wu}}]{Bu2017b}%
	\BibitemOpen
	\bibfield  {author} {\bibinfo {author} {\bibfnamefont {K.}~\bibnamefont
			{Bu}}, \bibinfo {author} {\bibfnamefont {A.}~\bibnamefont {Kumar}}, \bibinfo
		{author} {\bibfnamefont {L.}~\bibnamefont {Zhang}},\ and\ \bibinfo {author}
		{\bibfnamefont {J.}~\bibnamefont {Wu}},\ }\bibfield  {title} {\bibinfo
		{title} {Cohering power of quantum operations},\ }\href
	{https://doi.org/10.1016/j.physleta.2017.03.022} {\bibfield  {journal}
		{\bibinfo  {journal} {Phys. Lett. A}\ }\textbf {\bibinfo {volume} {381}},\
		\bibinfo {pages} {1670 } (\bibinfo {year} {2017})}\BibitemShut {NoStop}%
	\bibitem [{\citenamefont {Hardy}\ and\ \citenamefont
		{Wright}(1979)}]{Hardy1984}%
	\BibitemOpen
	\bibfield  {author} {\bibinfo {author} {\bibfnamefont {G.~H.}\ \bibnamefont
			{Hardy}}\ and\ \bibinfo {author} {\bibfnamefont {E.~M.}\ \bibnamefont
			{Wright}},\ }\href@noop {} {\emph {\bibinfo {title} {An Introduction to the
				Theory of Numbers}}}\ (\bibinfo  {publisher} {Clarendon Press},\ \bibinfo
	{year} {1979})\BibitemShut {NoStop}%
	\bibitem [{\citenamefont {Nielsen}\ and\ \citenamefont
		{Chuang}(2016)}]{Nielsen2016}%
	\BibitemOpen
	\bibfield  {author} {\bibinfo {author} {\bibfnamefont {M.}~\bibnamefont
			{Nielsen}}\ and\ \bibinfo {author} {\bibfnamefont {I.}~\bibnamefont
			{Chuang}},\ }\href@noop {} {\emph {\bibinfo {title} {Quantum Computation and
				Quantum Information}}}\ (\bibinfo  {publisher} {Cambridge University Press},\
	\bibinfo {year} {2016})\BibitemShut {NoStop}%
	\bibitem [{\citenamefont {Griffiths}\ and\ \citenamefont
		{Niu}(1996)}]{Griffiths1995}%
	\BibitemOpen
	\bibfield  {author} {\bibinfo {author} {\bibfnamefont {R.~B.}\ \bibnamefont
			{Griffiths}}\ and\ \bibinfo {author} {\bibfnamefont {C.-S.}\ \bibnamefont
			{Niu}},\ }\bibfield  {title} {\bibinfo {title} {Semiclassical {F}ourier
			transform for quantum computation},\ }\href
	{https://doi.org/10.1103/PhysRevLett.76.3228} {\bibfield  {journal} {\bibinfo
			{journal} {Phys. Rev. Lett.}\ }\textbf {\bibinfo {volume} {76}},\ \bibinfo
		{pages} {3228} (\bibinfo {year} {1996})}\BibitemShut {NoStop}%
	\bibitem [{\citenamefont {Gottesman}(1997)}]{gottesman1997}%
	\BibitemOpen
	\bibfield  {author} {\bibinfo {author} {\bibfnamefont {D.}~\bibnamefont
			{Gottesman}},\ }\bibfield  {title} {\bibinfo {title} {Stabilizer codes and
			quantum error correction},\ }\href {https://arxiv.org/abs/quant-ph/9705052}
	{\bibfield  {journal} {\bibinfo  {journal} {arXiv:quant-ph/9705052}\ }
		(\bibinfo {year} {1997})}\BibitemShut {NoStop}%
	\bibitem [{\citenamefont {Rosset}\ \emph {et~al.}(2018)\citenamefont {Rosset},
		\citenamefont {Buscemi},\ and\ \citenamefont {Liang}}]{Rosset2018}%
	\BibitemOpen
	\bibfield  {author} {\bibinfo {author} {\bibfnamefont {D.}~\bibnamefont
			{Rosset}}, \bibinfo {author} {\bibfnamefont {F.}~\bibnamefont {Buscemi}},\
		and\ \bibinfo {author} {\bibfnamefont {Y.-C.}\ \bibnamefont {Liang}},\
	}\bibfield  {title} {\bibinfo {title} {Resource theory of quantum memories
			and their faithful verification with minimal assumptions},\ }\href
	{https://doi.org/10.1103/PhysRevX.8.021033} {\bibfield  {journal} {\bibinfo
			{journal} {Phys. Rev. X}\ }\textbf {\bibinfo {volume} {8}},\ \bibinfo {pages}
		{021033} (\bibinfo {year} {2018})}\BibitemShut {NoStop}%
	\bibitem [{\citenamefont {Streltsov}\ \emph {et~al.}(2015)\citenamefont
		{Streltsov}, \citenamefont {Singh}, \citenamefont {Dhar}, \citenamefont
		{Bera},\ and\ \citenamefont {Adesso}}]{Streltsov2015}%
	\BibitemOpen
	\bibfield  {author} {\bibinfo {author} {\bibfnamefont {A.}~\bibnamefont
			{Streltsov}}, \bibinfo {author} {\bibfnamefont {U.}~\bibnamefont {Singh}},
		\bibinfo {author} {\bibfnamefont {H.~S.}\ \bibnamefont {Dhar}}, \bibinfo
		{author} {\bibfnamefont {M.~N.}\ \bibnamefont {Bera}},\ and\ \bibinfo
		{author} {\bibfnamefont {G.}~\bibnamefont {Adesso}},\ }\bibfield  {title}
	{\bibinfo {title} {Measuring quantum coherence with entanglement},\ }\href
	{https://doi.org/10.1103/PhysRevLett.115.020403} {\bibfield  {journal}
		{\bibinfo  {journal} {Phys. Rev. Lett.}\ }\textbf {\bibinfo {volume} {115}},\
		\bibinfo {pages} {020403} (\bibinfo {year} {2015})}\BibitemShut {NoStop}%
	\bibitem [{\citenamefont {Egloff}\ \emph {et~al.}(2018)\citenamefont {Egloff},
		\citenamefont {Matera}, \citenamefont {Theurer},\ and\ \citenamefont
		{Plenio}}]{Egloff2018}%
	\BibitemOpen
	\bibfield  {author} {\bibinfo {author} {\bibfnamefont {D.}~\bibnamefont
			{Egloff}}, \bibinfo {author} {\bibfnamefont {J.~M.}\ \bibnamefont {Matera}},
		\bibinfo {author} {\bibfnamefont {T.}~\bibnamefont {Theurer}},\ and\ \bibinfo
		{author} {\bibfnamefont {M.~B.}\ \bibnamefont {Plenio}},\ }\bibfield  {title}
	{\bibinfo {title} {Of local operations and physical wires},\ }\href
	{https://doi.org/10.1103/PhysRevX.8.031005} {\bibfield  {journal} {\bibinfo
			{journal} {Phys. Rev. X}\ }\textbf {\bibinfo {volume} {8}},\ \bibinfo {pages}
		{031005} (\bibinfo {year} {2018})}\BibitemShut {NoStop}%
	\bibitem [{\citenamefont {Theurer}\ \emph {et~al.}(2020)\citenamefont
		{Theurer}, \citenamefont {Satyajit},\ and\ \citenamefont
		{Plenio}}]{Theurer2020}%
	\BibitemOpen
	\bibfield  {author} {\bibinfo {author} {\bibfnamefont {T.}~\bibnamefont
			{Theurer}}, \bibinfo {author} {\bibfnamefont {S.}~\bibnamefont {Satyajit}},\
		and\ \bibinfo {author} {\bibfnamefont {M.~B.}\ \bibnamefont {Plenio}},\
	}\bibfield  {title} {\bibinfo {title} {Quantifying dynamical coherence with
			dynamical entanglement},\ }\href
	{https://doi.org/10.1103/PhysRevLett.125.130401} {\bibfield  {journal}
		{\bibinfo  {journal} {Phys. Rev. Lett.}\ }\textbf {\bibinfo {volume} {125}},\
		\bibinfo {pages} {130401} (\bibinfo {year} {2020})}\BibitemShut {NoStop}%
	\bibitem [{\citenamefont {Kay}(2018)}]{Kay2018}%
	\BibitemOpen
	\bibfield  {author} {\bibinfo {author} {\bibfnamefont {A.}~\bibnamefont
			{Kay}},\ }\bibfield  {title} {\bibinfo {title} {Tutorial on the quantikz
			package},\ }\href {https://arxiv.org/abs/1809.03842} {\bibfield  {journal}
		{\bibinfo  {journal} {arXiv:1809.03842}\ } (\bibinfo {year}
		{2018})}\BibitemShut {NoStop}%
	\bibitem [{\citenamefont {L{\"{o}}fberg}()}]{Lofberg2004}%
	\BibitemOpen
	\bibfield  {author} {\bibinfo {author} {\bibfnamefont {J.}~\bibnamefont
			{L{\"{o}}fberg}},\ }\bibfield  {title} {\bibinfo {title} {{YALMIP} : A
			toolbox for modeling and optimization in {MATLAB}},\ }in\ \href
	{https://ieeexplore.ieee.org/abstract/document/1393890/references#references}
	{\emph {\bibinfo {booktitle} {In Proceedings of the CACSD Conference}}}\
	(\bibinfo {address} {Taipei, Taiwan})\BibitemShut {NoStop}%
	\bibitem [{\citenamefont {Sturm}(1999)}]{Sturm1999}%
	\BibitemOpen
	\bibfield  {author} {\bibinfo {author} {\bibfnamefont {J.~F.}\ \bibnamefont
			{Sturm}},\ }\bibfield  {title} {\bibinfo {title} {Using {SeDuMi} 1.02, a
			{M}atlab toolbox for optimization over symmetric cones},\ }\href
	{https://doi.org/10.1080/10556789908805766} {\bibfield  {journal} {\bibinfo
			{journal} {Optim. Methods Softw.}\ }\textbf {\bibinfo {volume} {11}},\
		\bibinfo {pages} {625} (\bibinfo {year} {1999})}\BibitemShut {NoStop}%
	\bibitem [{\citenamefont {Grant}\ and\ \citenamefont {Boyd}(2014)}]{cvx}%
	\BibitemOpen
	\bibfield  {author} {\bibinfo {author} {\bibfnamefont {M.}~\bibnamefont
			{Grant}}\ and\ \bibinfo {author} {\bibfnamefont {S.}~\bibnamefont {Boyd}},\
	}\href@noop {} {\bibinfo {title} {{CVX}: Matlab software for disciplined
			convex programming, version 2.1}},\ \bibinfo {howpublished}
	{\url{http://cvxr.com/cvx}} (\bibinfo {year} {2014})\BibitemShut {NoStop}%
	\bibitem [{\citenamefont {Vedral}\ \emph {et~al.}(1996)\citenamefont {Vedral},
		\citenamefont {Barenco},\ and\ \citenamefont {Ekert}}]{PhysRevA.54.147}%
	\BibitemOpen
	\bibfield  {author} {\bibinfo {author} {\bibfnamefont {V.}~\bibnamefont
			{Vedral}}, \bibinfo {author} {\bibfnamefont {A.}~\bibnamefont {Barenco}},\
		and\ \bibinfo {author} {\bibfnamefont {A.}~\bibnamefont {Ekert}},\ }\bibfield
	{title} {\bibinfo {title} {Quantum networks for elementary arithmetic
			operations},\ }\href {https://doi.org/10.1103/PhysRevA.54.147} {\bibfield
		{journal} {\bibinfo  {journal} {Phys. Rev. A}\ }\textbf {\bibinfo {volume}
			{54}},\ \bibinfo {pages} {147} (\bibinfo {year} {1996})}\BibitemShut
	{NoStop}%
	\bibitem [{\citenamefont {Beckman}\ \emph {et~al.}(1996)\citenamefont
		{Beckman}, \citenamefont {Chari}, \citenamefont {Devabhaktuni},\ and\
		\citenamefont {Preskill}}]{PhysRevA.54.1034}%
	\BibitemOpen
	\bibfield  {author} {\bibinfo {author} {\bibfnamefont {D.}~\bibnamefont
			{Beckman}}, \bibinfo {author} {\bibfnamefont {A.~N.}\ \bibnamefont {Chari}},
		\bibinfo {author} {\bibfnamefont {S.}~\bibnamefont {Devabhaktuni}},\ and\
		\bibinfo {author} {\bibfnamefont {J.}~\bibnamefont {Preskill}},\ }\bibfield
	{title} {\bibinfo {title} {Efficient networks for quantum factoring},\ }\href
	{https://doi.org/10.1103/PhysRevA.54.1034} {\bibfield  {journal} {\bibinfo
			{journal} {Phys. Rev. A}\ }\textbf {\bibinfo {volume} {54}},\ \bibinfo
		{pages} {1034} (\bibinfo {year} {1996})}\BibitemShut {NoStop}%
	\bibitem [{\citenamefont {Zalka}(1998)}]{zalka1998fast}%
	\BibitemOpen
	\bibfield  {author} {\bibinfo {author} {\bibfnamefont {C.}~\bibnamefont
			{Zalka}},\ }\bibfield  {title} {\bibinfo {title} {Fast versions of {S}hor's
			quantum factoring algorithm},\ }\href
	{https://arxiv.org/abs/quant-ph/9806084} {\bibfield  {journal} {\bibinfo
			{journal} {arXiv:quant-ph/9806084}\ } (\bibinfo {year} {1998})}\BibitemShut
	{NoStop}%
	\bibitem [{\citenamefont {Gour}\ and\ \citenamefont
		{Spekkens}(2008{\natexlab{b}})}]{Gour_2008}%
	\BibitemOpen
	\bibfield  {author} {\bibinfo {author} {\bibfnamefont {G.}~\bibnamefont
			{Gour}}\ and\ \bibinfo {author} {\bibfnamefont {R.~W.}\ \bibnamefont
			{Spekkens}},\ }\bibfield  {title} {\bibinfo {title} {The resource theory of
			quantum reference frames: manipulations and monotones},\ }\href
	{https://doi.org/10.1088/1367-2630/10/3/033023} {\bibfield  {journal}
		{\bibinfo  {journal} {New J. Phys.}\ }\textbf {\bibinfo {volume} {10}},\
		\bibinfo {pages} {033023} (\bibinfo {year} {2008}{\natexlab{b}})}\BibitemShut
	{NoStop}%
	\bibitem [{\citenamefont {Gour}\ \emph {et~al.}(2009)\citenamefont {Gour},
		\citenamefont {Marvian},\ and\ \citenamefont
		{Spekkens}}]{PhysRevA.80.012307}%
	\BibitemOpen
	\bibfield  {author} {\bibinfo {author} {\bibfnamefont {G.}~\bibnamefont
			{Gour}}, \bibinfo {author} {\bibfnamefont {I.}~\bibnamefont {Marvian}},\ and\
		\bibinfo {author} {\bibfnamefont {R.~W.}\ \bibnamefont {Spekkens}},\
	}\bibfield  {title} {\bibinfo {title} {Measuring the quality of a quantum
			reference frame: The relative entropy of frameness},\ }\href
	{https://doi.org/10.1103/PhysRevA.80.012307} {\bibfield  {journal} {\bibinfo
			{journal} {Phys. Rev. A}\ }\textbf {\bibinfo {volume} {80}},\ \bibinfo
		{pages} {012307} (\bibinfo {year} {2009})}\BibitemShut {NoStop}%
	\bibitem [{\citenamefont {Marvian}\ and\ \citenamefont
		{Spekkens}(2013)}]{Marvian_2013}%
	\BibitemOpen
	\bibfield  {author} {\bibinfo {author} {\bibfnamefont {I.}~\bibnamefont
			{Marvian}}\ and\ \bibinfo {author} {\bibfnamefont {R.~W.}\ \bibnamefont
			{Spekkens}},\ }\bibfield  {title} {\bibinfo {title} {The theory of
			manipulations of pure state asymmetry: I. basic tools, equivalence classes
			and single copy transformations},\ }\href
	{https://doi.org/10.1088/1367-2630/15/3/033001} {\bibfield  {journal}
		{\bibinfo  {journal} {New J. Phys.}\ }\textbf {\bibinfo {volume} {15}},\
		\bibinfo {pages} {033001} (\bibinfo {year} {2013})}\BibitemShut {NoStop}%
	\bibitem [{\citenamefont {Marvian}\ \emph {et~al.}(2016)\citenamefont
		{Marvian}, \citenamefont {Spekkens},\ and\ \citenamefont
		{Zanardi}}]{Marvian_2016}%
	\BibitemOpen
	\bibfield  {author} {\bibinfo {author} {\bibfnamefont {I.}~\bibnamefont
			{Marvian}}, \bibinfo {author} {\bibfnamefont {R.~W.}\ \bibnamefont
			{Spekkens}},\ and\ \bibinfo {author} {\bibfnamefont {P.}~\bibnamefont
			{Zanardi}},\ }\bibfield  {title} {\bibinfo {title} {Quantum speed limits,
			coherence, and asymmetry},\ }\href
	{https://doi.org/10.1103/PhysRevA.93.052331} {\bibfield  {journal} {\bibinfo
			{journal} {Phys. Rev. A}\ }\textbf {\bibinfo {volume} {93}},\ \bibinfo
		{pages} {052331} (\bibinfo {year} {2016})}\BibitemShut {NoStop}%
	\bibitem [{\citenamefont {Miller}\ and\ \citenamefont
		{Takloo-Bighash}(2006)}]{Miller2006}%
	\BibitemOpen
	\bibfield  {author} {\bibinfo {author} {\bibfnamefont {S.~J.}\ \bibnamefont
			{Miller}}\ and\ \bibinfo {author} {\bibfnamefont {R.}~\bibnamefont
			{Takloo-Bighash}},\ }\href@noop {} {\emph {\bibinfo {title} {An Invitation to
				Modern Number Theory}}}\ (\bibinfo  {publisher} {Princeton University
		Press},\ \bibinfo {year} {2006})\BibitemShut {NoStop}%
	\bibitem [{\citenamefont {Gerjuoy}(2005)}]{Gerjuoy2005}%
	\BibitemOpen
	\bibfield  {author} {\bibinfo {author} {\bibfnamefont {E.}~\bibnamefont
			{Gerjuoy}},\ }\bibfield  {title} {\bibinfo {title} {{S}hor’s factoring
			algorithm and modern cryptography. {A}n illustration of the capabilities
			inherent in quantum computers},\ }\href {https://doi.org/10.1119/1.1891170}
	{\bibfield  {journal} {\bibinfo  {journal} {Am. J. Phys.}\ }\textbf {\bibinfo
			{volume} {73}},\ \bibinfo {pages} {521} (\bibinfo {year} {2005})}\BibitemShut
	{NoStop}%
	\bibitem [{\citenamefont {Bourdon}\ and\ \citenamefont
		{Williams}(2006)}]{Bourdon}%
	\BibitemOpen
	\bibfield  {author} {\bibinfo {author} {\bibfnamefont {P.~S.}\ \bibnamefont
			{Bourdon}}\ and\ \bibinfo {author} {\bibfnamefont {H.~T.}\ \bibnamefont
			{Williams}},\ }\bibfield  {title} {\bibinfo {title} {Sharp probability
			estimates for {S}hor's order-finding algorithm},\ }\href
	{https://arxiv.org/abs/quant-ph/0607148} {\bibfield  {journal} {\bibinfo
			{journal} {arXiv:quant-ph/0607148}\ } (\bibinfo {year} {2006})}\BibitemShut
	{NoStop}%
	\bibitem [{\citenamefont {Chitambar}\ and\ \citenamefont
		{Gour}(2016{\natexlab{a}})}]{Chitamber2016Comp}%
	\BibitemOpen
	\bibfield  {author} {\bibinfo {author} {\bibfnamefont {E.}~\bibnamefont
			{Chitambar}}\ and\ \bibinfo {author} {\bibfnamefont {G.}~\bibnamefont
			{Gour}},\ }\bibfield  {title} {\bibinfo {title} {Comparison of incoherent
			operations and measures of coherence},\ }\href
	{https://doi.org/10.1103/PhysRevA.94.052336} {\bibfield  {journal} {\bibinfo
			{journal} {Phys. Rev. A}\ }\textbf {\bibinfo {volume} {94}},\ \bibinfo
		{pages} {052336} (\bibinfo {year} {2016}{\natexlab{a}})}\BibitemShut
	{NoStop}%
	\bibitem [{\citenamefont {Marvian}\ and\ \citenamefont
		{Spekkens}(2016)}]{Marvian2016}%
	\BibitemOpen
	\bibfield  {author} {\bibinfo {author} {\bibfnamefont {I.}~\bibnamefont
			{Marvian}}\ and\ \bibinfo {author} {\bibfnamefont {R.~W.}\ \bibnamefont
			{Spekkens}},\ }\bibfield  {title} {\bibinfo {title} {How to quantify
			coherence: Distinguishing speakable and unspeakable notions},\ }\href
	{https://doi.org/10.1103/PhysRevA.94.052324} {\bibfield  {journal} {\bibinfo
			{journal} {Phys. Rev. A}\ }\textbf {\bibinfo {volume} {94}},\ \bibinfo
		{pages} {052324} (\bibinfo {year} {2016})}\BibitemShut {NoStop}%
	\bibitem [{\citenamefont {Chitambar}\ and\ \citenamefont
		{Gour}(2016{\natexlab{b}})}]{Chitamber2016}%
	\BibitemOpen
	\bibfield  {author} {\bibinfo {author} {\bibfnamefont {E.}~\bibnamefont
			{Chitambar}}\ and\ \bibinfo {author} {\bibfnamefont {G.}~\bibnamefont
			{Gour}},\ }\bibfield  {title} {\bibinfo {title} {Critical examination of
			incoherent operations and a physically consistent resource theory of quantum
			coherence},\ }\href {https://doi.org/10.1103/PhysRevLett.117.030401}
	{\bibfield  {journal} {\bibinfo  {journal} {Phys. Rev. Lett.}\ }\textbf
		{\bibinfo {volume} {117}},\ \bibinfo {pages} {030401} (\bibinfo {year}
		{2016}{\natexlab{b}})}\BibitemShut {NoStop}%
	\bibitem [{\citenamefont {Moreno}\ and\ \citenamefont
		{García-Caballero}(2013)}]{MORENO201390}%
	\BibitemOpen
	\bibfield  {author} {\bibinfo {author} {\bibfnamefont {S.~G.}\ \bibnamefont
			{Moreno}}\ and\ \bibinfo {author} {\bibfnamefont {E.~M.}\ \bibnamefont
			{García-Caballero}},\ }\bibfield  {title} {\bibinfo {title} {On
			{V}iète-like formulas},\ }\href
	{https://doi.org/https://doi.org/10.1016/j.jat.2013.06.006} {\bibfield
		{journal} {\bibinfo  {journal} {J. Approx. Theory}\ }\textbf {\bibinfo
			{volume} {174}},\ \bibinfo {pages} {90} (\bibinfo {year} {2013})}\BibitemShut
	{NoStop}%
\end{thebibliography}
\end{document}